\documentclass[11p]{article}
\usepackage{lineno}
\modulolinenumbers[1]

\usepackage{tikz}
\usetikzlibrary{arrows,calc,positioning}

\usepackage[utf8]{inputenc}
\usepackage{amsmath}
\usepackage{amsfonts}
\usepackage{amssymb}
\usepackage{amsfonts}
\usepackage{amsthm}
\usepackage{amscd}
\usepackage{color}
\usepackage{graphics}
\usepackage{graphicx}
\usepackage{overpic}
\usepackage[toc,page]{appendix}

\usepackage[all]{xy}

\usepackage{algorithm,algorithmicx,algpseudocode}

\algrenewcommand\algorithmicrequire{\textbf{INPUT:}}
\algrenewcommand\algorithmicensure{\textbf{OUTPUT:}}

\usepackage{concmath}
\usepackage[T1]{fontenc}

\usepackage{tikz}

\newcommand{\C}{\mathbb{C}}
\newcommand{\R}{\mathbb{R}}

\newcommand{\Z}{\mathbb{Z}}
\newcommand{\N}{\mathbb{N}}

\newcommand{\iso}{\mathrm{Iso}\,}
\newcommand{\sym}{\mathrm{Sym}\,}

\newcommand{\I}{\mathrm{i}}
\newcommand{\e}{\mathrm{e}}

\newcommand{\f}[1]{\mathbf{#1}}

\newcommand{\OO}{\mathbf{O}}

\newtheorem{theorem}{\bf Theorem}
\newtheorem{lemma}[theorem]{\bf Lemma}

\newtheorem{remark}[theorem]{\bf Remark}

\newtheorem{example}[theorem]{\bf Example}
\newtheorem{corollary}[theorem]{\bf Corollary}

\begin{document}

\sloppy

\title{Symmetries of discrete curves and point clouds via trigonometric interpolation}

\author{Michal Bizzarri\footnote{\tt bizzarri@kma.zcu.cz}, Miroslav L\'avi\v{c}ka\footnote{\tt lavicka@kma.zcu.cz}, Jan Vr\v{s}ek\footnote{\tt vrsekjan@kma.zcu.cz}}

\maketitle

\begin{abstract}
We formulate a simple algorithm for computing global exact symmetries of closed discrete curves in plane.  The method is based on a suitable trigonometric interpolation of vertices of the given polyline and consequent computation of the symmetry group of the obtained trigonometric curve. The algorithm exploits the fact that the introduced unique assigning of the trigonometric curve to each closed discrete curve commutes with isometries. For understandable reasons, an essential part of the paper is devoted to determining rotational and axial symmetries of trigonometric curves. We also show that the formulated approach can be easily applied on nonorganized clouds of points. A functionality of the designed detection method is presented on several examples.
\end{abstract}

\section{Introduction and motivation}\label{sec:intro}

Symmetries have always been an indisputable phenomenon in the real world, considered often as an aspect of perfection, balance and harmony. Being symmetric is a potentially very useful feature which many real shapes exhibit, and thus symmetries in the natural world have significantly inspired people to incorporate them when producing tools, buildings or artwork. However, investigation of symmetries has also become a powerful and effective tool of theoretical research in various disciplines. For instance, the notion of symmetry is a basic concept in Felix Klein's Erlangen Program devoted to the characterization of classes of geometry based on the underlying symmetry groups, see \cite{Kl1893}. Nowadays, symmetry detection belongs among standard problems in geometry, computer graphics, computer vision, geometric modelling, pattern recognition etc. It is very important for many applications to know how to detect symmetries in geometrical models and a number of subsequent algorithms is based on this.

Let us recall that an object possesses a symmetry if there is a geometric transformation that maps this object onto itself, i.e., the object is invariant under this transformation. Symmetry is global, when it concerns the whole object or local, when only part of the object has this property. Symmetry can also be defined according to the type of transformation used (reflectional, rotational, central,\, \ldots). Symmetry can be exact (strong) or approximate (weak) with a given tolerance. Especially in many applications the considered geometric shapes are only approximate models of real objects (described often in a floating point representation) and therefore the symmetries are not perfect but only approximate. To develop symmetry detection methods, able to work with artificial as well as real-life data is a challenge, regardless whether the research is just theoretical or it concentrates on a certain application area. 

In this paper we will focus on global exact symmetries of investigated objects, in particular of planar curves. This research area has recently (or relatively recently) experienced a renaissance in geometric modelling and related disciplines and we can find many papers focused on the detection and computation of symmetries or equivalences of curves given implicitly or by their (most often) rational parameterization. We recall e.g. \cite{BrKn04,LeRG08}, or recent series of papers \cite{Al14,AlHeMu14,AlHeMu14b,AlHeMu18}. In \cite{HaJu18}, the authors present a method more general than the previous  ones as they study general affine and projective equivalences in an arbitrary space dimension. The problem of deterministically computing the symmetries of a given planar algebraic curve, implicitly defined, and the problem of deterministically checking whether or not two implicitly given, planar algebraic curves are similar, i.e., equal up to a similarity transformation, was investigated in \cite{AlLaVr19}.  Projective equivalence of special algebraic varieties, including projective (and other) equivalences of rational curves were studied and presented in \cite{BiLaVr20a}. A simple algorithm for an approximate reconstruction of an  inexact  planar curve which is assumed to be a perturbation of some unknown planar curve with symmetry was formulated in \cite{BiLaVr20}. 
Exact and approximate similarities of non-necessarily rational planar, parameterized curves, using centers of gravity and inertia tensors were studied in \cite{AlQu21}.

Another problem is detection of symmetries of objects in 2D or 3D which are represented discretely as sets of points.
So the goal is to find the symmetry transformation (if any) under which a certain point set is (exactly or approximately) mapped onto itself. This problem is addressed mainly by computer scientists as they typically process set of points or meshes. It is interesting that some approaches suggest to compute first an implicit equation associated to the given set and then to apply mathematical methods formulated for algebraic curves and surfaces, see e.g. \cite{TaCo98}, \cite{TaCo00}. We recall once more that methods for detecting symmetries of point sets and objects have been widely studied in computer graphics,  computer vision and pattern recognition. For the sake of brevity we recall at least some algorithms for detecting exact global symmetries from this point of view, see \cite{Wolter1985,Alt1988,Aguilar2015}. Of course, much more papers are devoted to approximate symmetries as the mentioned scientific disciplines work mainly with real-world, i.e., non-perfect, data.

In this paper, we continue with our previous research on symmetries of geometric shapes form the mathematical point of view, i.e., we are mainly interested in objects described by equations. On the other hand, we must admit that motivation for our approach originated in computer graphics. We have noticed that the iterative process of Laplacian smoothing of polygonal mesh used  for removing rough features of input shapes preserves symmetries of the original non-smoothed object. And this is only a small step to signal processing, discrete Fourier transform and thus also to trigonometric curves. We show that trigonometric interpolation makes perfect sense when one wants to determine exact global symmetries of closed discrete curves, as the symmetries of trigonometric curves can be easily and directly determined from their trigonometric parameterization without any need to switch to their implicit or rational description. A~special role is played by the so called higher cycloid curves considered as generalizations of standard cycloids.

The structure of the paper is the following. Section \ref{sec:prelim} provides some general results, including some background on isometries and symmetries, trigonometric curves and trigonometric interpolation, to be used later in the paper. Symmetry computation of trigonometric curves is addressed in Section \ref{sec:trigcurves_sym}, first for rotations including central symmetries, then for the axial symmetries. The next Section \ref{sec:sym_polyline} is devoted to the application on discrete curves in plane. By unique assigning to each discrete curve a certain trigonometric parameterization we transform the problem of computing symmetries of discrete curves to the previous case. In Section~\ref{sec:convexhull} we show that the formulated approach can be immediately used also for unorganized clouds of points. Our conclusions and some lines for future research are presented in Section \ref{sec:sum}.

\section{Preliminaries}\label{sec:prelim}

\subsection{Isometries and symmetries}\label{sec:iso_sym}

An~\emph{isometry} of the plane $\R^2$ is a~mapping $\phi: \R^2\rightarrow\R^2$ which preserves distances. It can be always written in the form
\begin{equation}\label{eq isometry}
  \f x\mapsto \f A\cdot \f x+\f b,
\end{equation}
where $\f A\in\OO(2)$ is an orthogonal matrix and $\f{b}\in\R^2$. The group of isometries in $\R^2$ will be denoted by $\iso(\R^2)$. The isometries preserving also the orientation of the plane are called \emph{direct} isometries and it holds $\det(\f A)=1$ in this case. For $\det(\f A)=-1$  we speak about {\em indirect} isometries. Identity, translation and rotation are examples of direct isometries, whereas reflection (i.e., axial symmetry) belongs among indirect isometries.

Two point sets $X,X'\subset \R^2$ are isometric (or {\em congruent}\/) if there exists $\phi\in \iso(\R^2)$ such that $X'=\phi(X)$. Furthermore, $X$ is said to possess a \emph{symmetry} if there exists a non-identical isometry $\phi$ such that $X=\phi(X)$. The {\em group of symmetries} of $X$ will be denoted by $\sym(X)$. The subgroup of orientation-preserving (i.e., direct) symmetries is called the {\em proper symmetry group}. 

\medskip
It is well known that any finite subgroup of $\iso(\R^2)$ falls into the following classes:
{\em cyclic groups} $C_m$ consisting of all rotations about the same center by multiples of the angle  $\frac{2\pi}{m}$; and {\em dihedral groups} $D_m$ consisting of all the rotations in $C_m$ together with all reflections in $m$ axes that go through the same point, the center. Especially, $C_1$ is the trivial group containing only the identity;
$D_1$ is the group containing the identity and one reflection; 
$C_2$ is the symmetry group of the letter `S' consisting of the identity and one central symmetry; $D_2$ is the symmetry group of a rectangle which is not a square containing the identity, one central symmetry and two reflections with perpendicular axes. And for $m\geq 3$ we have the symmetry groups ($D_m$) or the proper symmetry groups ($C_m$) of some non-oriented or oriented regular $m$-gon, respectively. 

\smallskip
Hence, to determine such a symmetry group it is sufficient to find:
\begin{enumerate}
  \vspace{-1ex}
  \item the center of symmetry (barycenter),
  \item the rotation angle $\frac{2\pi}{m}$ (where $m=1$ corresponds to the object without any rotational symmetry),
  \item the direction of one symmetry axis (if the $m$-gon is non-oriented, i.e., the symmetry group $D_m$).
\end{enumerate}

\medskip

The strategy used in the next sections for determining symmetries  is based on application of a suitable operator $\Psi$ which commutes with isometries. Then the unique assigning of some object to a given geometric shape $X$ via $\Psi$ enables to study symmetries of $\Psi(X) $ instead of symmetries of $X$. This is formalized in the following lemma. Of course, it is essential to find $\Psi$ wisely 
in order to simplify the situation, i.e., $\sym(\Psi(X))$ must be easier to determine than direct computing $\sym(X)$.

\begin{lemma}\label{lem:komut}
Let be given $\mathcal{X}\subset\mathcal{P}(\R^2)$, where $
\mathcal{P}(\R^2)$ is the set of all subsets of $\R^2$,  and $\Psi:\mathcal{X}\rightarrow\mathcal{P}(\R^2)$ such that
for all $X\in\mathcal{X}$ and for all $\phi\in\iso(\R^2)$ it holds
\begin{equation}
\Psi(\phi(X))=\phi(\Psi(X)),
\end{equation}
then $\sym(X)\subset\sym(\Psi(X))$.
\end{lemma}

\begin{proof}
 Let $\phi\in\sym(X)$, then $\phi(\Psi(X))=\Psi(\phi(X))=\Psi(X)$.
\end{proof}

\begin{example}\rm
We recall \cite{AlLaVr19} where $\mathcal{X}$ is a set of real planar algebraic curves and the mapping $\Psi$ is the Laplace operator $\Delta$. As known this operator commutes with isometries. This property was used in the referred paper for finding symmetries of an algebraic curve, say $f$, via the chain of symmetry groups 
\begin{equation}
\sym(f)\subset\sym(\Delta\,f)\subset\sym(\Delta^2\, f)\subset\cdots\subset\sym(\Delta^{\ell}\, f)=\sym(h),
\end{equation}
where $\sym(h)$ is finite and easier to determine. 
\end{example}

\begin{example}\rm\label{ex:Discrete_Laplace}
In this paper, we are interested in discrete curves, see Section~\ref{sec:sym_polyline}. So we mention also a discrete analogy of $\Delta$. Consider a closed planar polyline $X$ formed by a sequence of ordered points $\f v_0,\f v_1,\ldots, \f v_n=\f v_0$ connected with line segments. 

A~\emph{discrete Laplacian operator} is a map associating to each $\f v_i$ the vector 
\begin{equation}
   L(\f v_i) = \f v_i-\frac{1}{2}(\f v_{i-1}+\f v_{i+1}),
\end{equation}
This is a linear operator which can be represented by the circular matrix 
\begin{equation}
   L=
   \begin{pmatrix}
      1 &-\frac{1}{2}&0&\dots&\dots&0&-\frac{1}{2}\\
      -\frac{1}{2}&1&-\frac{1}{2}&0&\dots&\dots&0\\
      \vdots&\vdots&\vdots&\vdots&\vdots&\vdots&\vdots\\
      0&\dots&\dots&0&-\frac{1}{2}&1&-\frac{1}{2}\\
      -\frac{1}{2}&0&\dots&\dots&0&-\frac{1}{2}&1
   \end{pmatrix}.
\end{equation}
Next, a~\emph{smoothing operator} is defined as 
\begin{equation}
  S=I-\lambda L,
\end{equation}
where $\lambda\in (0,1)$, i.e.,  each vertex $\f v_i$ is replaced by the affine combination $(1-\lambda)\f v_i+\frac{\lambda}{2}(\f v_{i-1}+\f v_{i+1})$. The smoothing operator commutes with the affine (and thus also Euclidean) transformations and thus we arrive to the conclusion that it holds $\sym(X)\subset\sym(S(X))$, cf. Fig.~\ref{fig:smoothing}.
\end{example}

\begin{figure}[t]
\begin{center}
\includegraphics[height=0.23\textwidth]{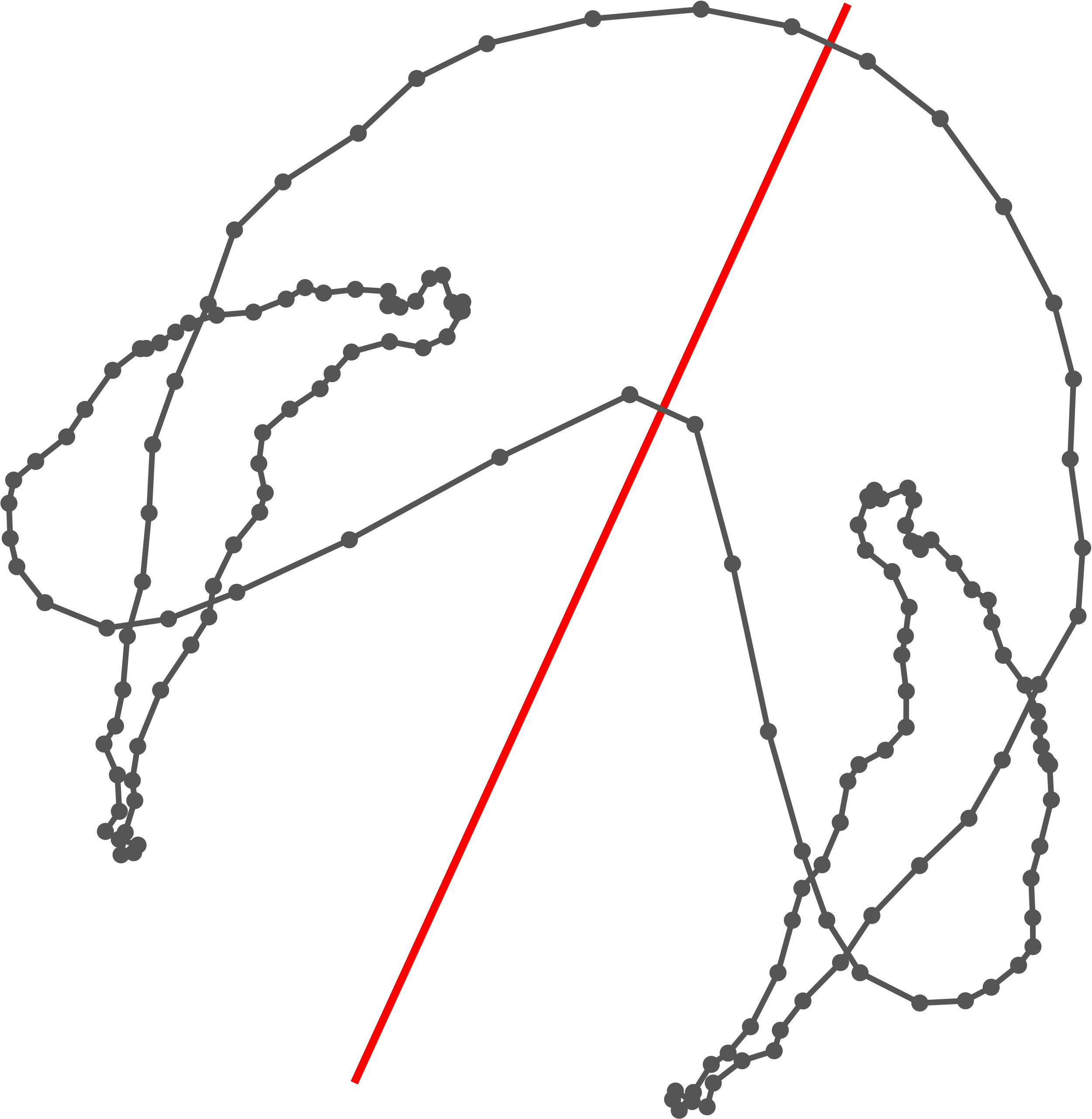}
\hspace{3ex}
\includegraphics[height=0.23\textwidth]{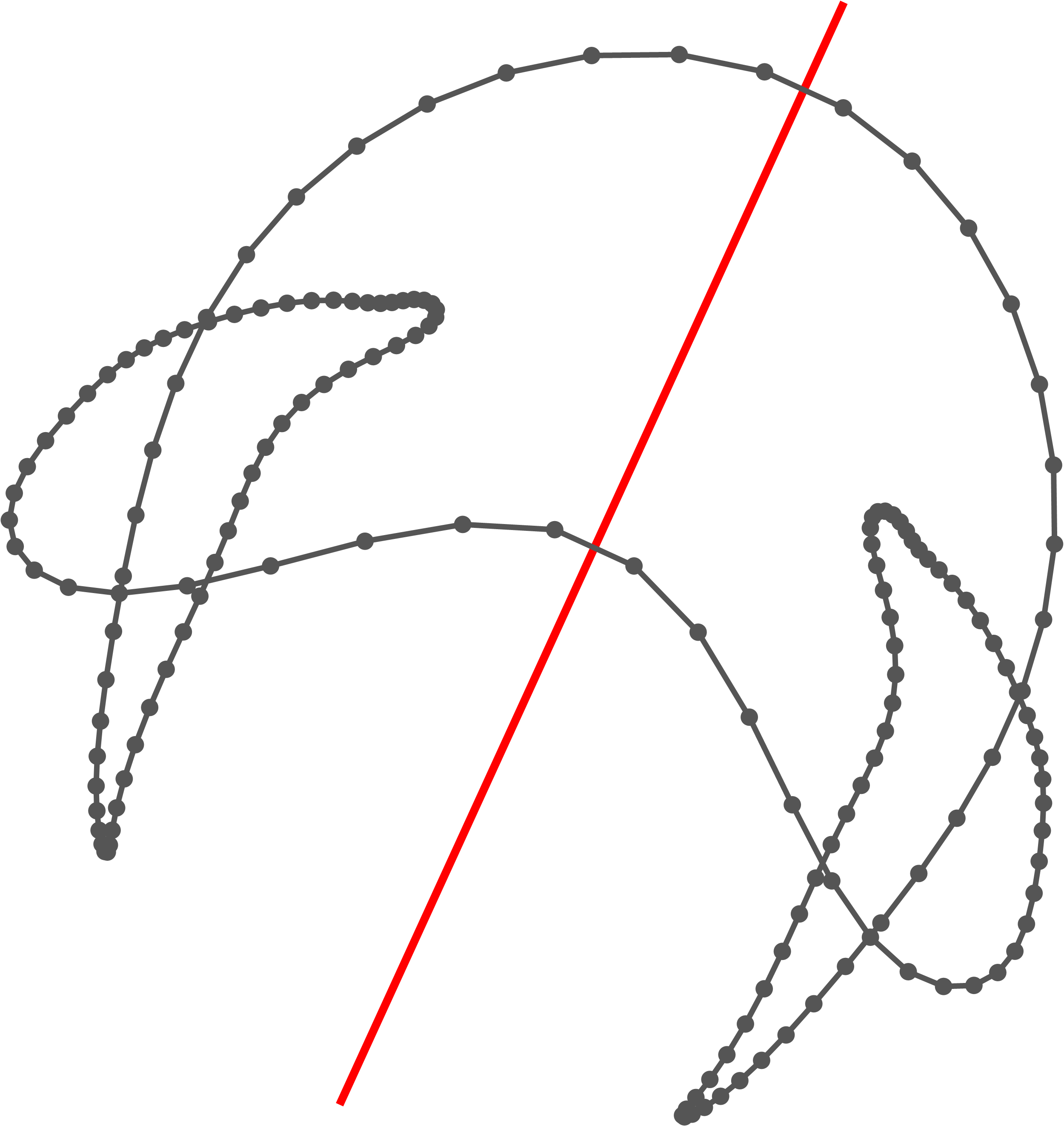}
\hspace{3ex}
\includegraphics[height=0.23\textwidth]{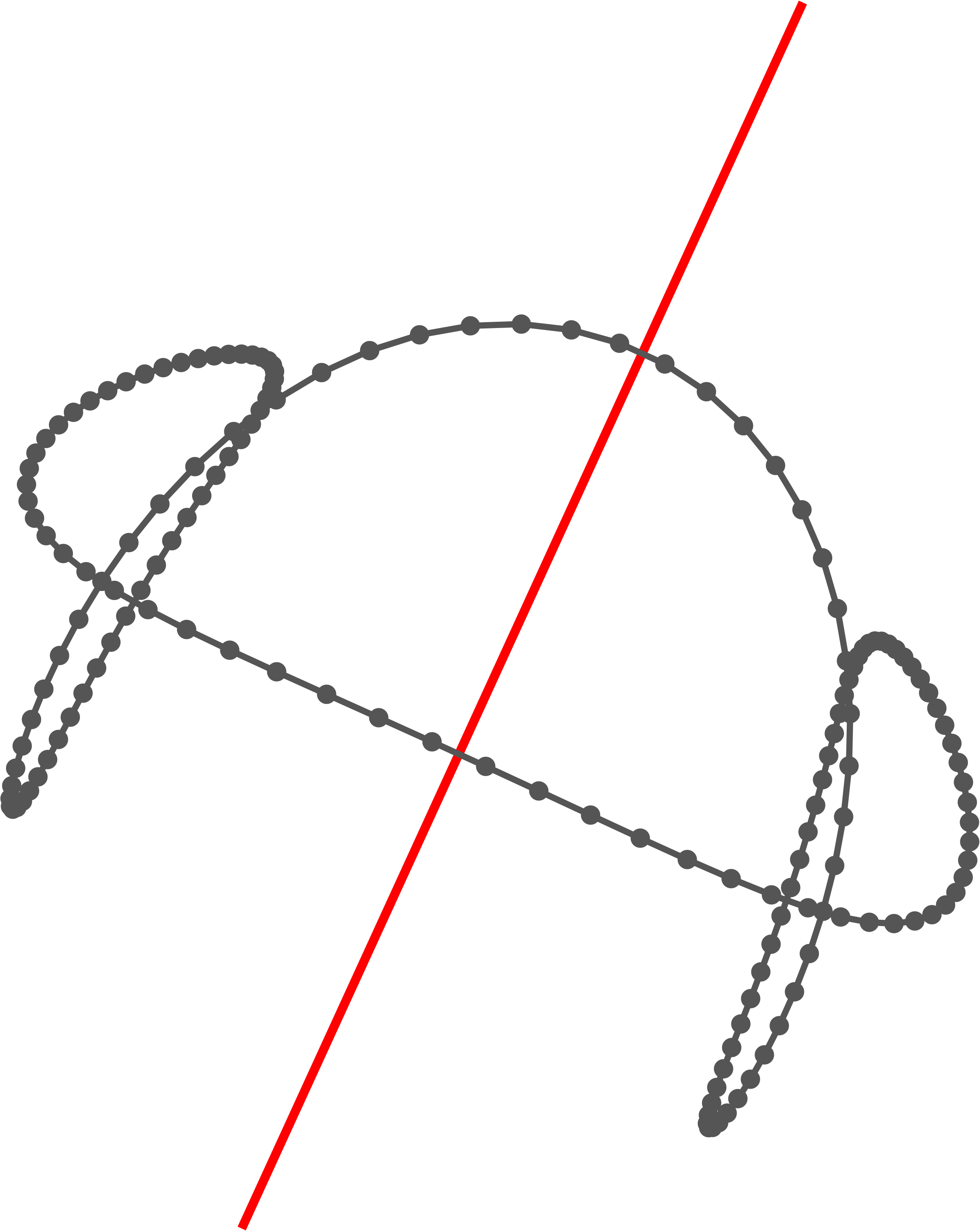}
\hspace{3ex}
\includegraphics[height=0.23\textwidth]{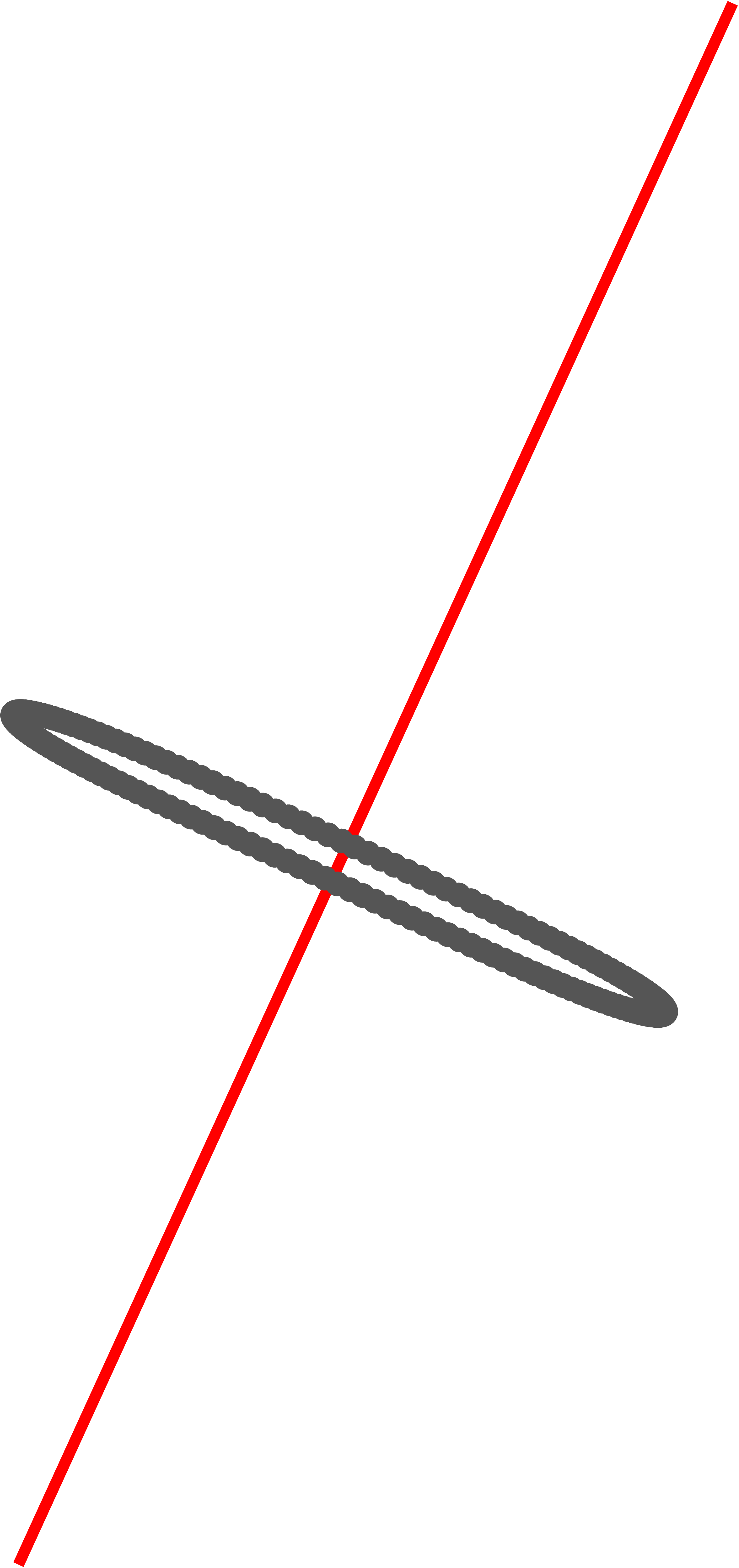}
\caption{Laplace smoothing and preservation of symmetries. From left to right: The original polyline and the polylines after Laplace smoothing w.r.t. $10$, $100$ and $1000$ step. All the polylines are symmetric w.r.t. the red axis.}\label{fig:smoothing}
\end{center}
\end{figure}

\begin{remark}\rm  
In fact, the content of Example~\ref{ex:Discrete_Laplace} was our first step to the topic of this paper. We have noticed that 
the procedure of smoothing used in computer graphics for removing rough features of a given shape exhibits the property that it preserves symmetries of the original non-smoothed version. And from the Laplacian smoothing and its analysis is a short way to  signal processing and discrete Fourier transform and thus to trigonometric curves, see \cite{Co-Or_etal15} and the following sections.
\end{remark}

\begin{figure}[t]
\begin{center}
\includegraphics[height=0.2\textwidth]{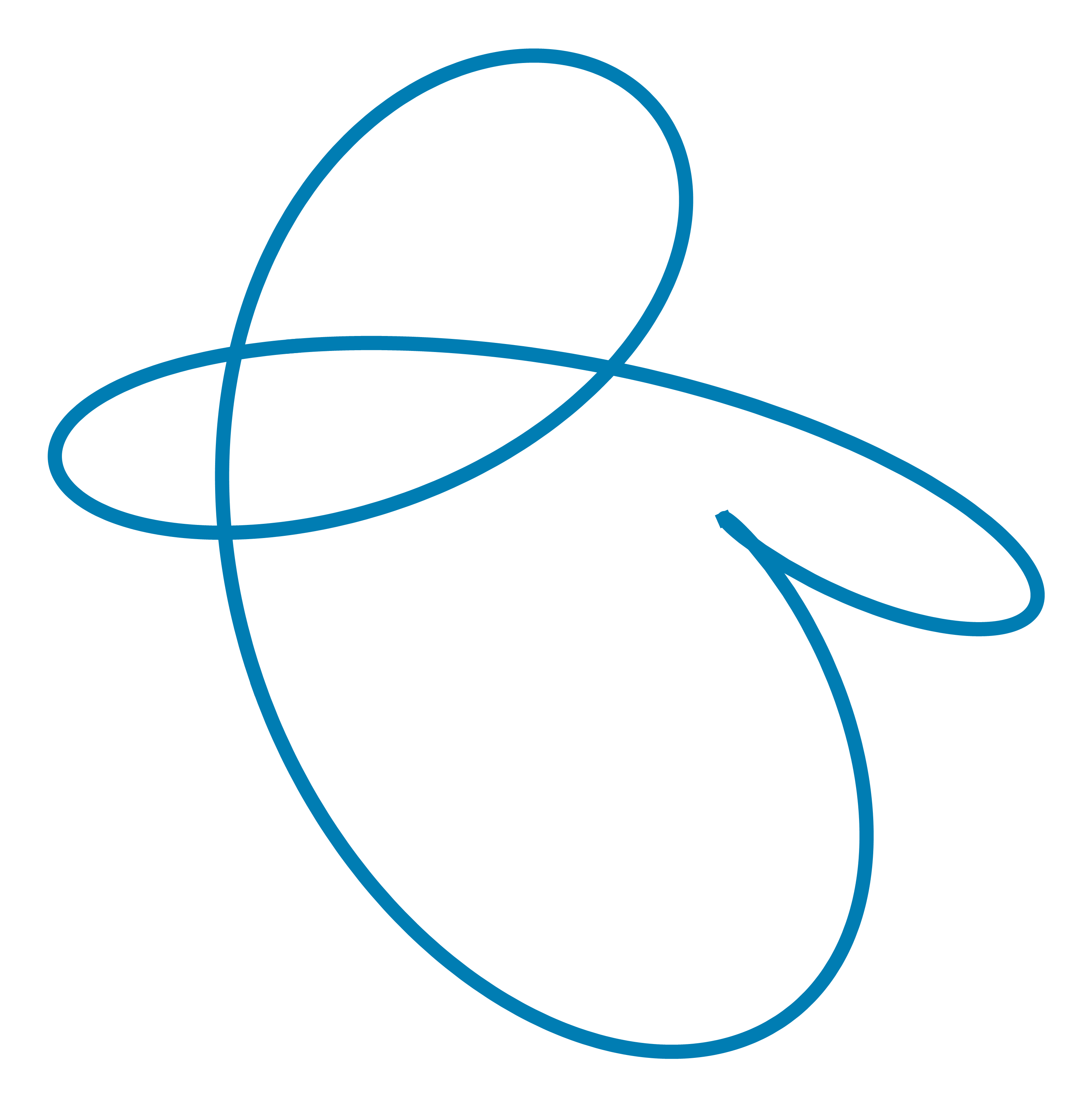}
\hspace{1ex}
\includegraphics[height=0.2\textwidth]{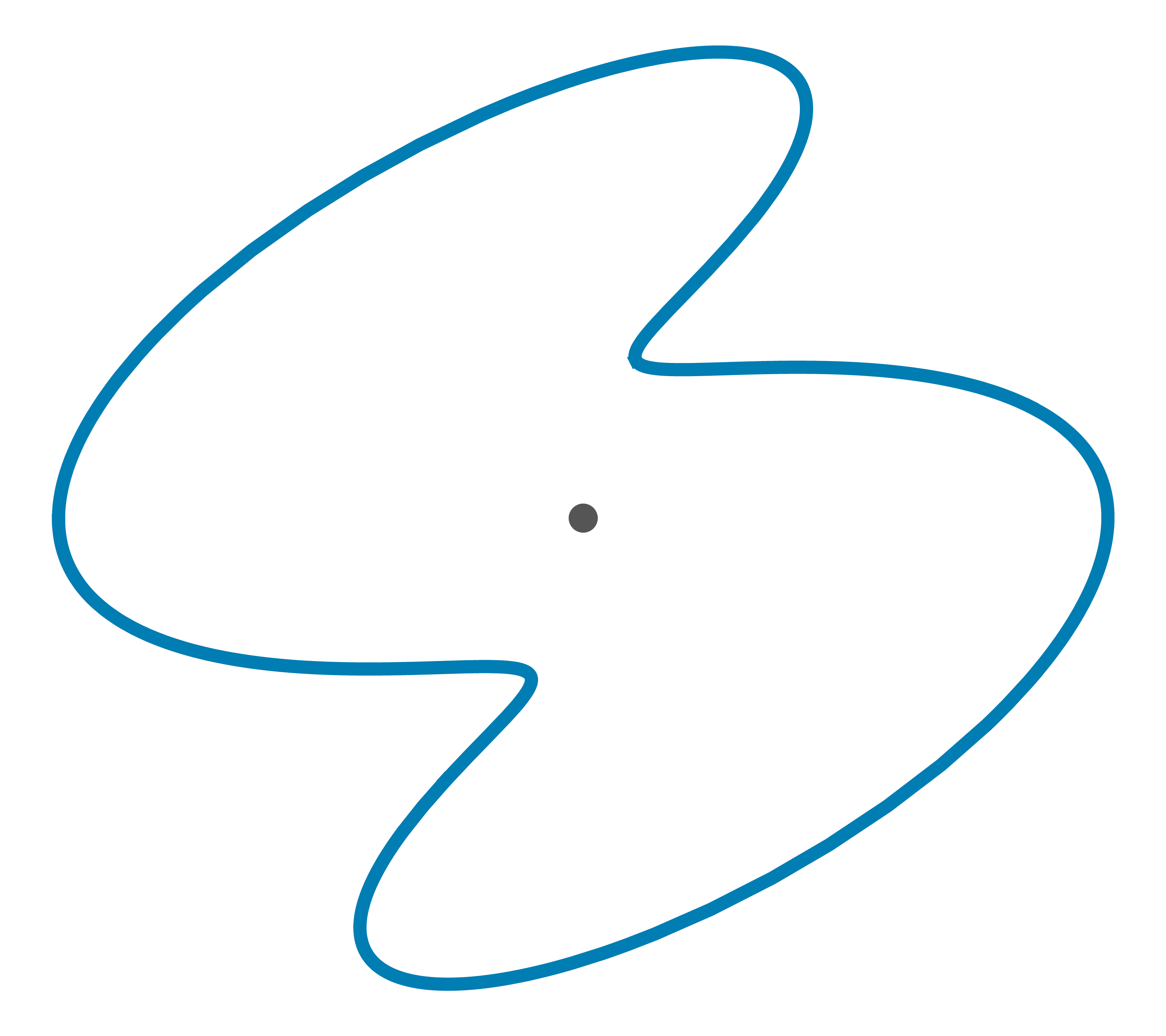}
\hspace{1ex}
\includegraphics[height=0.2\textwidth]{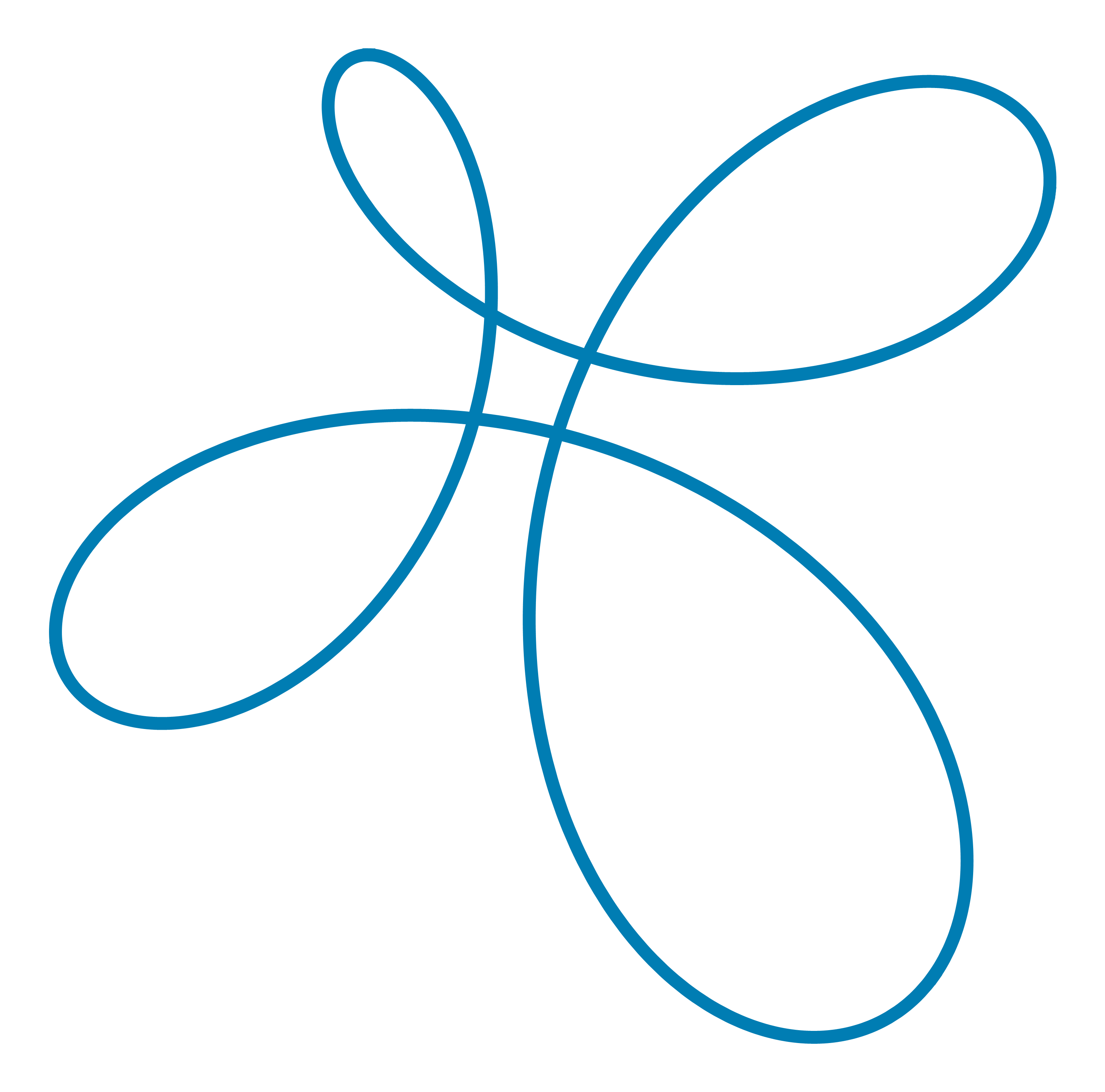}
\hspace{1ex}
\includegraphics[height=0.2\textwidth]{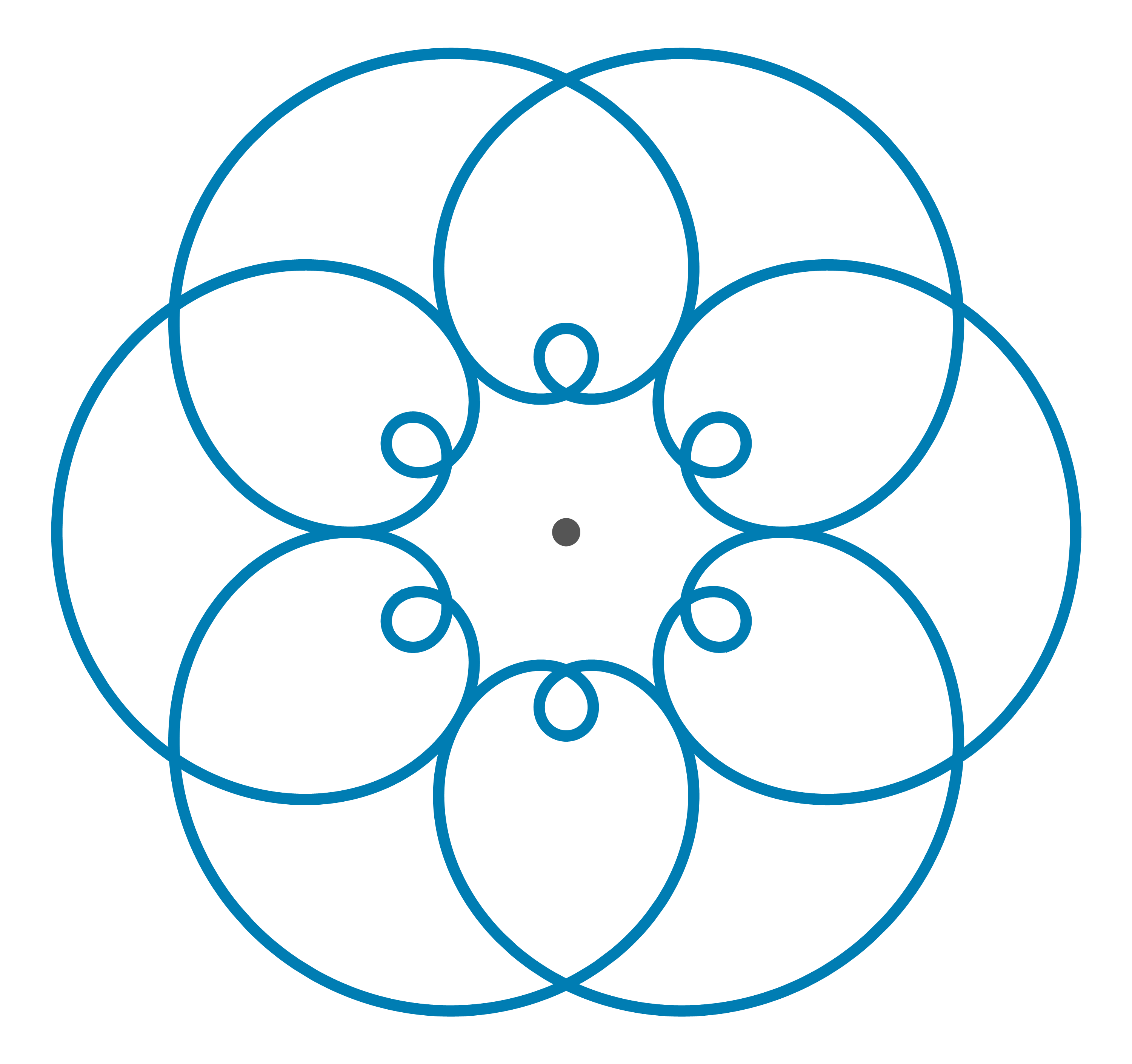}
\caption{Examples of trigonometric curves. From left to right: a generic non-symmetric trigonometric curve (all $\f p_k$'s are ellipses), a trigonometric curve with a central symmetry (all $\f p_k$'s are ellipses), a non-symmetric higher cycloid curve, and a higher cycloid curve with rotational and axial symmetries.}\label{fig:trig_curves}
\end{center}
\end{figure}

\subsection{Trigonometric curves}\label{sec:trigcurves}

By a~\emph{trigonometric curve} in $\R^2$ we understand a real curve possessing a parameterization in the form
\begin{equation}\label{eq:C_m}
  \f p(t)=\sum_{k=0}^N\big[\f a_k \cos(k t)+\f b_k\sin(k t)\big]=
  \f a_0+\sum_{k=1}^N\big[\f a_k \cos(k t)+\f b_k\sin(k t)\big]=\f a_0+\sum_{k=1}^{N}\f p_k(t),
\end{equation}
where $\f a_k,\f b_k\in\R^2$ and at least one of $\f a_N,\f b_N$ is non-zero, i.e., the coordinates are given by trigonometric polynomials of degree at most $N$ (truncated Fourier series). In short we will call $N$ the {\em trigonometric degree} of $\f p(t)$. If the $\gcd$ of all $k\in\{0,\ldots,N\}$ for which $\f a_k\neq 0$ or $\f b_k\neq 0$ is one then $\f p(t)$ is called \emph{primitive} (otherwise it can be reparameterized to obtain a primitive trigonometric curve). Notice that $\f p_k(t)$, where $t\in[0,2\pi]$, is a parameterization of a $k$-times traced ellipse.
It can be shown that \eqref{eq:C_m} parameterizes  an~algebraic irreducible curve of degree at most 2N and thus one can apply a suitable tool from algebraic geometry to get useful information about these shapes. For more details about properties of trigonometric curves see e.g. \cite{Wu47,HoSch98,RoJuSchHo09}. Some examples of trigonometric curves are shown in Fig.~\ref{fig:trig_curves}.

\smallskip
We recall that special instances of trigonometric curves also appear in the literature under the name  \emph{higher cycloid curves} or {\em cycloids of $N$th stage}, see e.g. \cite{Wu47}. They belong among kinematic curves generated by the trace of a fixed point on a circle that rolls without slipping around another circle that rolls without slipping around further circle etc. (i.e., they are generalizations of well knowns hypo-/epicycloids). The condition on \eqref{eq:C_m} being a higher cycloid curve has the form
\begin{equation}
\f (b_{k1},b_{k2})=\sigma_k(-a_{k2},a_{k1}),\quad \mbox{where } \sigma_k\in\{-1,+1\},
\end{equation}
i.e., considered $\f a_k,\f b_k$ as vectors in $\R^2$ then they are perpendicular. Thus for $\f p_k$ we arrive at
\begin{equation}\label{eq:circle_cond}
\f p_k=(a_{k1},a_{k2})\cos(kt)+\sigma_k(-a_{k2},a_{k1})\sin(kt).
\end{equation}
For the sake of simplicity, we naturally identify $\R^2$ with $\C$ and then \eqref{eq:circle_cond} can be rewritten as
\begin{equation}\label{eq:circle_cond2}
\f p_k=(a_{k1}\cos(kt)-\sigma_ka_{k2}\sin(kt),a_{k2}\cos(kt)+\sigma_ka_{k1}\sin(kt))=(a_{k1}+a_{k2}\I)\cdot[\cos(\sigma_k kt)+\I\sin(\sigma_k kt)].
\end{equation}
When setting $\cos(\sigma_k kt)+\I\sin(\sigma_k kt)=\e^{\sigma_k kt\I}$ and $a_{k1}+a_{k2}\I=\lambda_k\e^{\psi_k\I}$ we obtain the parameterization of the higher cycloid  curve in the form
\begin{equation}\label{eq:complex_param_higher_cycloid}
   \f p(t) = \f a_0+\sum_{k=1}^N \lambda_k \e^{(\sigma_k kt+\psi_k)\I},
\end{equation}
where $\lambda_k\in\R^{\geq 0}$ are the radii of the circles, $\psi_k\in[0,2\pi)$ are the ``phase shifts'' and $\sigma_k\in\{-1,+1\}$ determine the \emph{orientation} of the parameterization $\f p_k$,

Wunderlich discussed a lot of geometric properties of higher cycloids, e.g., he proved the generalization of Euler's theorem: {\em Every higher cycloid of $N$th stage can be generated in $N$ different ways by rolling of two higher cycloids of stage $N-1$}. Wunderlich also used higher cycloids  for  curve  approximation  of planar curves, see \cite{Hu04}. Some examples of higher cycloid curves are shown in Fig.~\ref{fig:trig_curves} (the last two instances).

\subsection{Periodic data and trigonometric interpolation}\label{sec:trig_interpol}

For later use we shortly recall the trigonometric interpolation problem as a very useful method for periodic data processing. For the sake of brevity we do not go into details and mention only the computational part of this problem -- the readers more interested in this topic are kindly referred e.g. to \cite{Br19,Zy02,StBu13}. The connection to the determination of the symmetries of discrete curves will be discussed in Section~\ref{sec:sym_polyline}. 

Let be given a set of $n$ ordered data $(t_j,x_j)$ that is periodic. In addition, we assume that $\{t_j\}$ are rescaled such that they lie in the interval $[0,2\pi)$, i.e.,
we have
\begin{equation}
0\leq t_0 < t_1 <\ldots < t_{n-1}< 2\pi.
\end{equation}
Our goal is to construct a trigonometric interpolation of these data, i.e., to find an interpolant in the form
\begin{equation}
 p(t)= \sum_{k=0}^N\big[a_k \cos(k t)+b_k\sin(k t)\big]=a_0+\sum_{k=1}^N\big[a_k \cos(k t)+b_k\sin(k t)\big],
\end{equation}
where $a_k,b_k\in\R$, going through $(t_j,x_j)$. In other words it is necessary to determine $2N+1$ real coefficients from $n$ input points. The existence of solution is guaranteed by the condition $n\leq 2N+1$. Specially, for $n=2N+1$ the interpolant is determined uniquely. If we have even number of input data, say $n=2N$, we arrive at one one-parametric solution. Then an additional requirement that the term $\sin(Nt)$ vanishes is often considered, which also results in the unique solution. The interpolation trigonometric polynomial has in this separate case the form
\begin{equation}
 p(t)= a_0+\sum_{k=1}^{N-1}\big[a_k \cos(k t)+b_k\sin(k t)\big]+a_N\cos(Nx).
\end{equation}

\smallskip
For non-uniformly spaced data $(t_j,x_j)$ one can just solve the associated system of linear equations. However when the nodes are uniformly spaced, i.e.,
\begin{equation}
t_j=j\frac{2\pi}{n}, \ j=0,\ldots,n-1.
\end{equation}
we obtain for odd number of input data $n=2N+1$ the closed-form expressions
\begin{equation}\label{eq:interpol_coef_odd}
\begin{array}{ll}
\displaystyle 
a_0  =  \frac{\sum_{j=0}^{n-1} x_j}{n}, &
\displaystyle 
a_k = \frac{2}{n}\sum_{j=0}^{n-1} x_j \cos(kt_j),\quad k=1,\dots,N,\\[3ex]
&
\displaystyle 
b_k =\frac{2}{n}\sum_{j=0}^{n-1} x_j \sin(kt_j),\quad k=0,1,\dots,N.
\end{array}
\end{equation}
The coefficients $a_k,b_k$ are called {\em discrete Fourier coefficients}. 
 
For even number of input data $n=2N$ and assuming that the term $\sin(Nt)$ vanishes, 
the coefficients are computed as
\begin{equation}\label{eq:interpol_coef_even}
\begin{array}{ll}
\displaystyle 
a_0  =  \frac{\sum_{j=0}^{n-1} x_j}{n}, \quad 
a_N  =  \frac{\sum_{j=0}^{n-1} x_j \cos(Nt_j)}{n},\,
&
\displaystyle 
a_k =  \frac{2}{n}\sum_{j=0}^{n-1} x_j \cos(kt_j),\quad k=1,\dots,N-1, \\[3ex]
&
\displaystyle 
b_k = \frac{2}{n}\sum_{j=0}^{n-1} x_j \sin(kt_j),\quad k=0,1,\dots,N.
\end{array}
\end{equation}

\section{Symmetries of trigonometric curves}\label{sec:trigcurves_sym}

First we would like to emphasize that any trigonometric curve $C$ is algebraic and thus it has a~finite group of symmetries unless it is a~circle or a line. In addition, these curves are also rational which suggests to apply for detecting their equivalences (including symmetries) approaches formulated for general rational curves. In this context we recall that special properties of rational parameterizations of trigonometric curves were exploited e.g. in \cite{AlQu20} for formulating a computational method for determining affine equivalences of trigonometric curves.
As our plan leads to a different goal (symmetries of discrete curves via trigonometric interpolation) we formulate a simple algorithm for determining symmetries of trigonometric curves directly from their trigonometric parameterizations.

\smallskip
Let us recall that there are infinitely many trigonometric parameterizations of the curve $C$, when one exists. However Theorem 4.6 from \cite{HoSch98} says that these parameterizations are `essentially' the same. For subsequent use we recall this theorem; for its proof the readers are kindly referred to the original source.

\begin{theorem}\label{thm:schicho}
  A trigonometric parameterization of a curve is unique up to a linear parameter change.
\end{theorem}

If $\phi\in\iso(\R^2)$ is a symmetry of a curve with trigonometric parameterization $\f p(t)$ then by Theorem~\ref{thm:schicho} there exists a linear reparameterization $t\mapsto \alpha t+\beta$, $\alpha,\beta \in \R$, such that 
\begin{equation}\label{eq:reparam}
   \phi(\f p(t))=\f p(\alpha t+\beta).
\end{equation}
Corollaries of this observation will allow us to determine the group of symmetries of any~trigonometric curve easily.

\begin{lemma}\label{lem:center}
  Let a trigonometric curve parameterized by $\f p(t)$ possesses a rotational symmetry with center $\f c$. Then $\f c=\f a_0$.
\end{lemma}

\begin{proof}
If $\f c =\f o$ is the origin, then the rotation $\rho$ is a~linear map an thus $\rho(\sum\f p_k(t))=\sum\rho(\f p_k(t))$. In particular, \eqref{eq:reparam} implies that $\rho(\f a_0)=\f a_0$ and thus $\f a_0=\f o$. The statement with general center is an immediate consequence.
\end{proof}

In the same way we can prove an analogous lemma dealing with axial symmetries.

\begin{lemma}\label{lem:point on L}
   Let a trigonometric curve parameterized by $\f p(t)$ possesses an axial symmetry with axis $L$. Then $\f a_0\in L$.
\end{lemma}
 
For the sake of simplicity we will assume in what follows that $\f a_0=\f o$ and thus in order to determine the symmetries of the curve it is enough to find the rotation angle and/or the direction of the axis. Let us start with the rotational symmetries.

\begin{lemma}\label{lem:rot_sym}
Let $\rho_{\frac{2\pi}{m}}:\R^2\rightarrow\R^2$ be the~rotation around the origin by the angle $\frac{2\pi}{m}$. Then it is a symmetry of $\f p(t)$ iff 
there exists $d\in\Z$ such that 
\begin{equation}
\rho_{\frac{2\pi}{m}}(\f p(t))=\f p\left(t+\frac{2\pi d}{m}\right).
\end{equation}\end{lemma}

\begin{proof}
By~\eqref{eq:reparam} we have $\rho_{\frac{2\pi}{m}}(\f p(t))=\f p(\alpha t+\beta)$. Since $\rho_{\frac{2\pi}{m}}$ is a linear orientation preserving mapping then $\alpha$ must equal to $1$.  Applying the rotation $m$ times we arrive at the original curve, i.e., $\f p(t)=\f p(t+m\beta)$. As $\f p(t)$ is $2\pi$-periodic we can conclude that $m\beta = 2\pi d$.
\end{proof}

\begin{remark}\rm
In the case $m=2$, the curve is centrally symmetric.  If the center is at the point $\f o$ then Lemma~\ref{lem:rot_sym} says that $-\f p(t)=\f p(t+\pi)$.  Moreover, this identity must hold for every  $\f p_k(t)=\f a_k \cos(kt)+\f b_k\sin(kt)$ which implies $\f p_k(t)=0$ for all even $k$. A particular example of a centrally symmetric trigonometric curve 
is shown in Fig.~\ref{fig:trig_curves} (second from the left).  
\end{remark}


Since a general ellipse does not possess an $m$-fold rotational symmetry for $m>2$, the~requirement for the existence of such a symmetry of the whole curve $\f p(t)$ forces each ellipse $\f p_k(t)$ to become a circle. In this case it is possible to write $\f p(t)$ (centered at $\f o$) in the form
\begin{equation}\label{eq:complex_param}
   \f p(t) = \sum_{k=1}^N \lambda_k \e^{(\sigma_k kt+\psi_k)\I},
\end{equation}
see \eqref{eq:complex_param_higher_cycloid}.

By Lemma~\ref{lem:rot_sym} the curve parameterized by $\f p(t)$ has an $m$-fold rotational symmetry if and only if
\begin{equation}\label{eq:complex_sym}
  \e^\frac{2\pi\I}{m}\cdot \f p_k(t)=\f p_k\left(t+\frac{2\pi d}{m}\right) \quad \mbox{for all } k.
\end{equation}
Combining \eqref{eq:complex_param} with \eqref{eq:complex_sym} we arrive at
\begin{equation}
\lambda_k\e^{\frac{2\pi\I}{m}}\cdot\e^{(\sigma_k kt+\psi_k)\I}=\lambda_k\e^{(\sigma_k k(t+\frac{2\pi d}{m})+\psi_k)\I}
\end{equation}
which yields
\begin{equation}\label{eq:symmetry_cyclic_curves}
  \frac{2\pi}{m}(1-\sigma_k\cdot k\cdot d)\equiv 0 \mod 2\pi.
\end{equation}
This forces some $\f p_k$'s to disappear. In particular the only non-zero frequencies occur when $\sigma_k\cdot k\cdot d -1$ is divisible by $m$. Let us formulate this as a lemma.

\begin{lemma}\label{lem:rot_symmetry}
   Let $\f p(t)$ be a primitive trigonometric curve. Then it possesses an $m$-fold rotational  symmetry ($m>2$) with the center $\f c$ if and only if $\f a_0=\f c$ and there exists an integer $d$ such that $\f p(t)$ can be written in the form~\eqref{eq:complex_param}, where the only non-zero terms $\lambda_k$ are those fulfilling the condition $\sigma_k\cdot k\cdot d = j\cdot m +1$ for some $j\in\Z$.
\end{lemma}

\begin{remark}\rm
We recall that \eqref{eq:complex_param} are higher cycloid curves, cf. Section~\ref{sec:trigcurves}. They are non-symmetric in general, and Lemma~\ref{lem:rot_symmetry} is a condition for their rotational symmetry.
\end{remark}

\medskip
Based on Lemma~\ref{lem:rot_symmetry} we define \emph{$(m,d)$-sequence} $\vartheta^{m,d}:\N\longrightarrow \mathcal{P}(\{-1,1\})$ associating $\{\pm 1\}$ to $k$ if $m$ divides $\pm kd-1$ and $\emptyset$ otherwise. We can observe that the following properties  hold for all $k\in\N$:
\begin{itemize}
  \item $\vartheta^{m,d}(k)=\emptyset$ whenever $\gcd(m,d)\neq 1$ (in particular $\vartheta^{m,0}(k)=\emptyset$),
  \item $\vartheta^{m,d+m}(k)=\vartheta^{m,d}(k)$,
  \item $\vartheta^{m,-d}(k)=-\vartheta^{m,d}(k)$,
  \item $\vartheta^{m,d}(k)=\vartheta^{m,d}(k+m)$.
\end{itemize}
Especially for $m=1,2$ we arrive at $\vartheta^{1,d}(k)=\{\pm 1\}$, $\vartheta^{2,1}(2k)=\emptyset$ and $\vartheta^{2,1}(2k+1)=\{\pm 1\}$. 

\medskip
Let be given two sequences $\alpha,\beta:\N\longrightarrow\mathcal{P}(\{-1,1\})$. 
We define 
\begin{equation}
  \alpha\preccurlyeq\beta\text{ if and only if for all }k\in\N\text{ it holds }\alpha(k)\subseteq\beta(k).
\end{equation}
Next, we assign to given parameterization \eqref{eq:complex_param} a sequence by the formula 
\begin{equation}
\sigma_\f p(k)=
\begin{cases}
\{\sigma_k\} & \text{if } \lambda_k\neq 0,\\[1ex]
\ \emptyset & \text{otherwise}.
\end{cases}
\end{equation}
The parameterized curve possesses an $m$-fold rotational symmetry if and only if  there exists $d\in\Z$  such that $\sigma_\f p\preccurlyeq\vartheta^{m,d}$. Among these sequences we are looking for the one with the maximal $m$. Such sequence is called the {\em maximal $(m,d)$-sequence} associated to $\f p(t)$.  Let us emphasize that if two trigonometric parameterizations are related by a reparameterization $\f q(t)=\f p(\pm t+\beta)$ then $\sigma_\f q = \pm\sigma_\f p$.

\begin{remark}\rm
Both considered types of sequences $\vartheta^{m,d}$ and $\sigma_{\f p}$ are infinite. Especially, $\vartheta^{m,d}$ is $m$-periodic whose elements repeat over and over. For the sake of simplicity we will write it down by just stating the period, i.e., the listing will end with the element $\vartheta^{m,d}(m)$. As concerns the sequence $\sigma_{\f p}$ then the last non-zero element is $\sigma_{\f p}(N)$, where $N$ is the trigonometric degree of the curve, and all the following elements are~$\emptyset$. Hence our convention is that its listing will end with the element $\sigma_{\f p}(N)$. Then it makes sense to write for instance 
$$
\big(\{1\},\{-1\},\emptyset,\emptyset,\{-1\}\big)\preccurlyeq \big(\{1\},\{-1\},\emptyset,\ldots\big)=\vartheta^{3,1}.
$$
\end{remark}

\begin{example}\rm
Set $m=7$ then there exist $3$ different non-trivial types of $(7,d)$-sequences. Since all of them are $7$-periodic, they are determined by their first 7 terms, in particular
\begin{align*}
   \vartheta^{7,1} &= (\{1\}, \emptyset, \emptyset, \emptyset, \emptyset, \{-1\}, \emptyset,\ldots),\\
   \vartheta^{7,2} &= (\emptyset, \emptyset, \{-1\}, \{1\}, \emptyset, \emptyset, \emptyset,\ldots),\\
   \vartheta^{7,3} &= (\emptyset, \{-1\}, \emptyset, \emptyset, \{1\}, \emptyset, \emptyset,\ldots).
\end{align*}
Some particular instances of curves with these maximal $(m,d)$-sequences are depicted in Fig.~\ref{fig:7fold}. Notice that since $\vartheta^{7,i}=-\vartheta^{7,7-i}$ for $i=1,2,3$, the curves with sequence $\vartheta^{7,i}$ are related to those with sequence $\vartheta^{7,7-i}$ by the~reparameterization $t\mapsto -t$  and thus they describe the same curves.

\begin{figure}
\begin{center}
\includegraphics[width=0.25\textwidth]{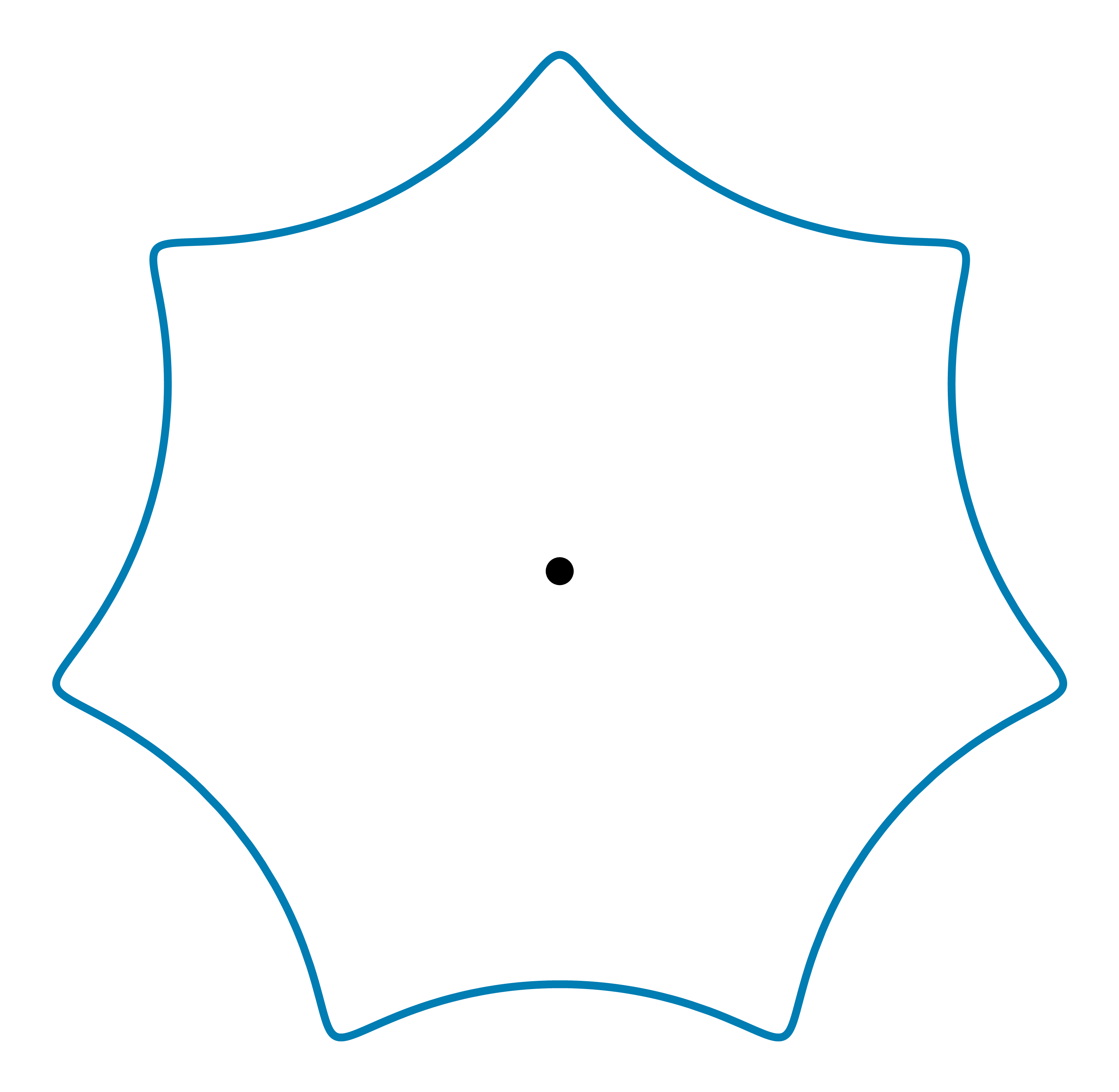}
\hspace{5ex}
\includegraphics[width=0.25\textwidth]{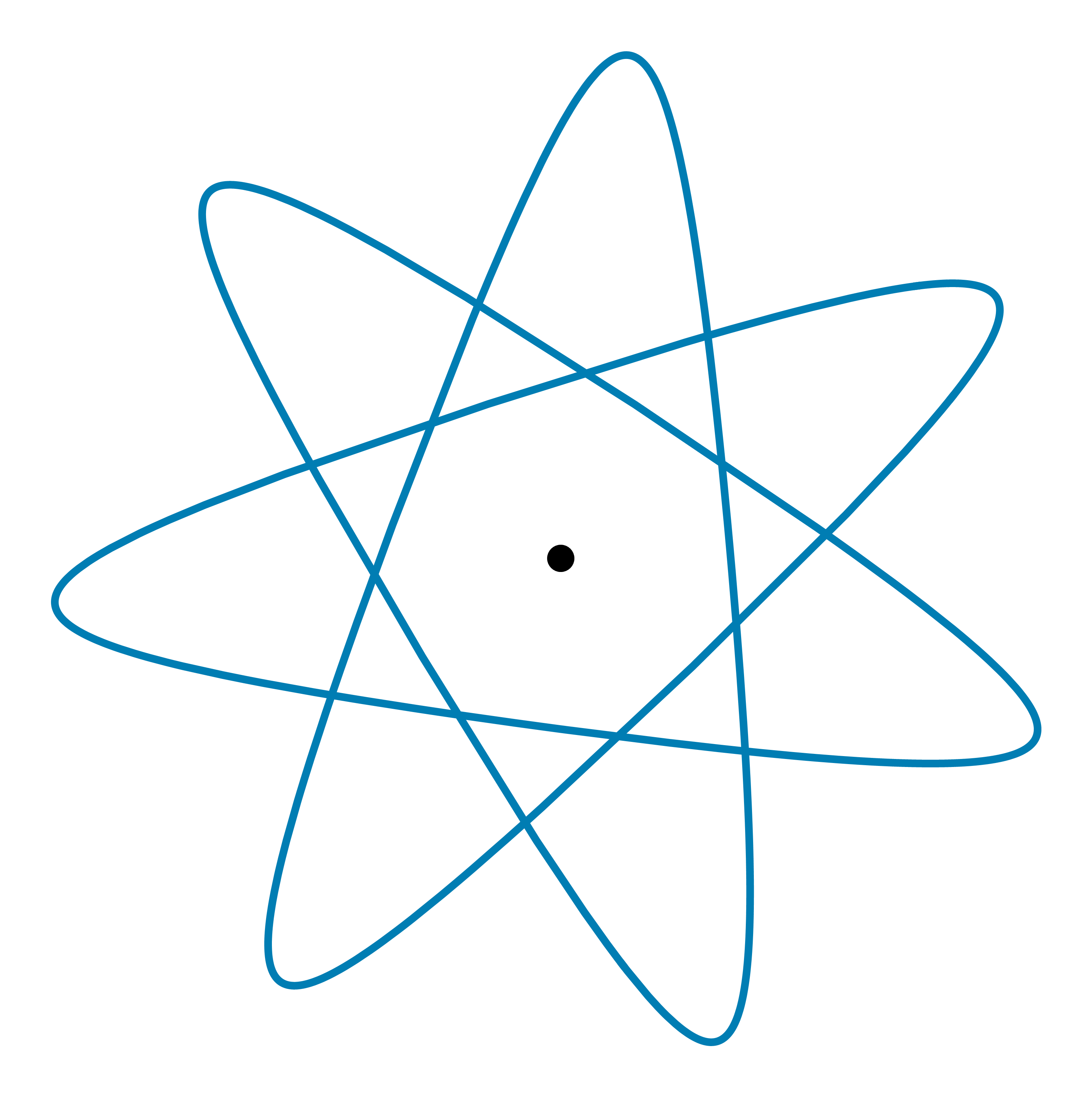}
\hspace{5ex}
\includegraphics[width=0.25\textwidth]{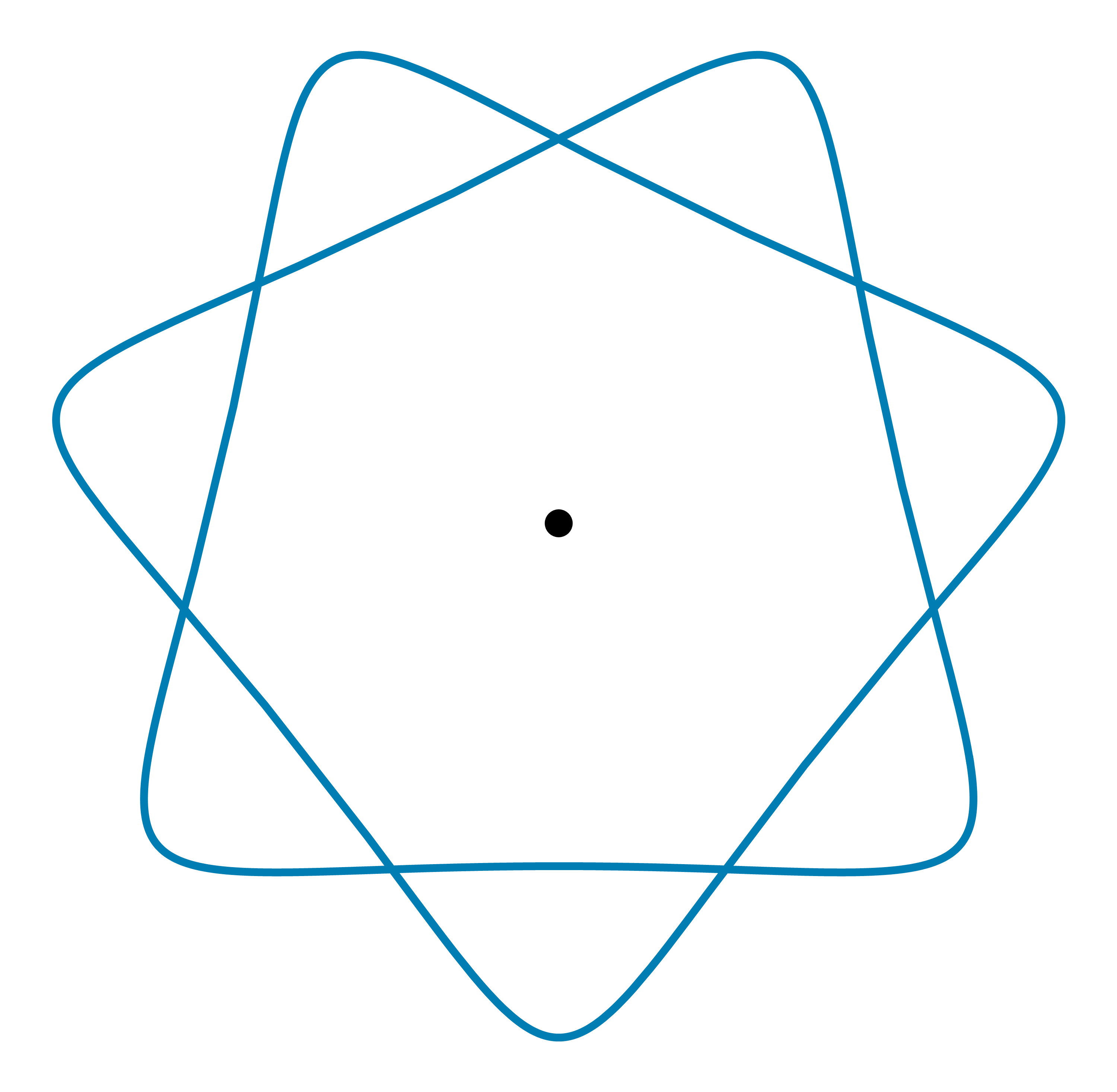}
\caption{Examples of rotationally symmetric curves with maximal sequences (from left to right) $\vartheta^{7,1}$, $\vartheta^{7,2}$ and $\vartheta^{7,3}$.}\label{fig:7fold}
\end{center}
\end{figure}
\end{example}

\begin{remark}\rm
We emphasize the importance of finding the maximal $(m,d)$-sequence associated to the given trigonometric curve. For instance, if the sequence is of the form
$$
\sigma_{\f p}=(\{1\}, \emptyset, \emptyset, \emptyset,\{-1\}),
$$
then it holds $\sigma_{\f p}\preccurlyeq \vartheta^{3,1}=(\{1\}, \{-1\},\emptyset,\ldots)$. However, the maximal sequence which is to be found is $\vartheta^{6,1}=(\{1\}, \emptyset, \emptyset, \emptyset, \{-1\},\emptyset,\ldots)$ and hence the curve possesses rotational symmetries of 6-gon.
\end{remark}

%


Now let us move to reflections. We can formulate a lemma analogous to Lemma~\ref{lem:rot_sym}.

\begin{lemma}\label{lem:refl_sym}
Let $\omega_{L}:\R^2\rightarrow\R^2$ be a reflection across a line $L$. Then it is a symmetry of $\f p(t)$ iff there exists a constant $t_0$ such that
\begin{equation}\label{eq:refl_sym}
\omega_{L}(\f p(t_0+t))=\f p(t_0-t).
\end{equation}
\end{lemma}

By Lemma~\ref{lem:point on L} we already know one point on the axis $L$.    Hence it remains to determine its direction, or equivalently we need to find the point $\f p(t_0)$. Assume first that  $t_0=0$ and the axis $L$ is identical to the $x$-axis, then by Lemma~\ref{lem:refl_sym} we can write
\begin{equation}
  \omega_L(\f p_k(t))=\f p_k(-t)\quad \mbox{for all } k.
\end{equation}
Since $\f p_k(t)=\f a_k \cos(k t)+\f b_k\sin(k t)$ and $\omega_L:[x,y]\mapsto[ x,-y]$ we arrive at the conditions $\f a_k=(a_{k1},0)$ and $\f b_k=(0,b_{k2})$.  Geometrically, this says that all the ellipses  $\f p_k(t)$ are equally aligned with respect to their axes and in addition they are coherently parameterized, which means that $\f p_k(0)$ is for all of them a vertex lying on the symmetry axis. Let us summarize this in a lemma

\begin{lemma}\label{lem:ellipse_syzygy}
 Let $\f p(t)$ be a primitive trigonometric curve.  Then it possesses a reflectional symmetry with axis $L$ if and only if the line $L$ passes through the point $\f a_0$ and there exists $t_0\in[0,2\pi)$ such that for all $k$ the point $\f p_k(t_0)$ is the vertex of the ellipse and if $\f p_k(t_0)-\f a_0\neq\f o$ then it is the direction of $L$. 
\end{lemma}

\begin{figure}[t]
\begin{center}
\begin{overpic}[height=0.3\textwidth]{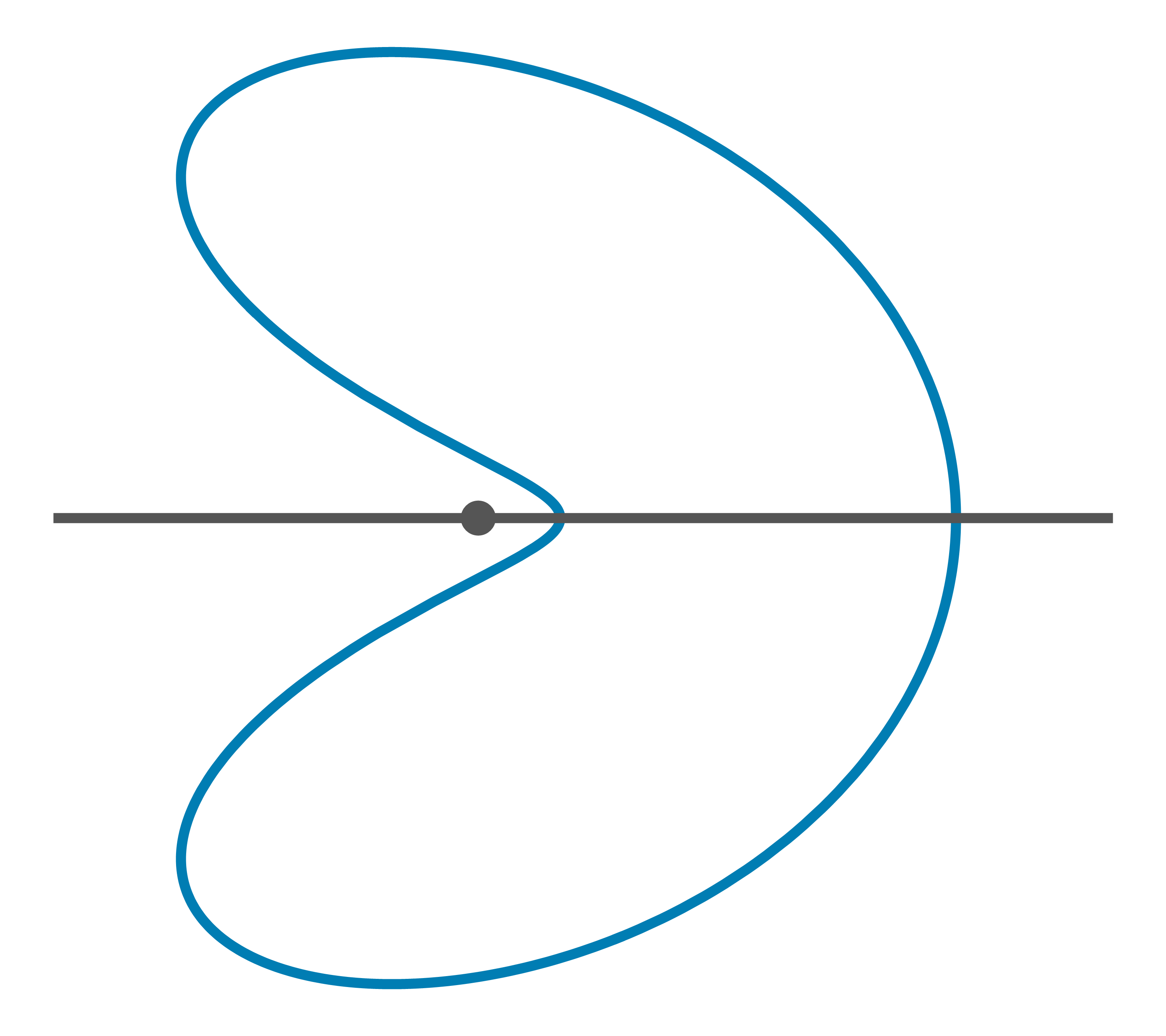}
\put(35,40){$\f o$}
\put(90,38){$L$}
\put(65,10){$\f p$}
\end{overpic}
\hspace{5ex}
\begin{overpic}[height=0.3\textwidth]{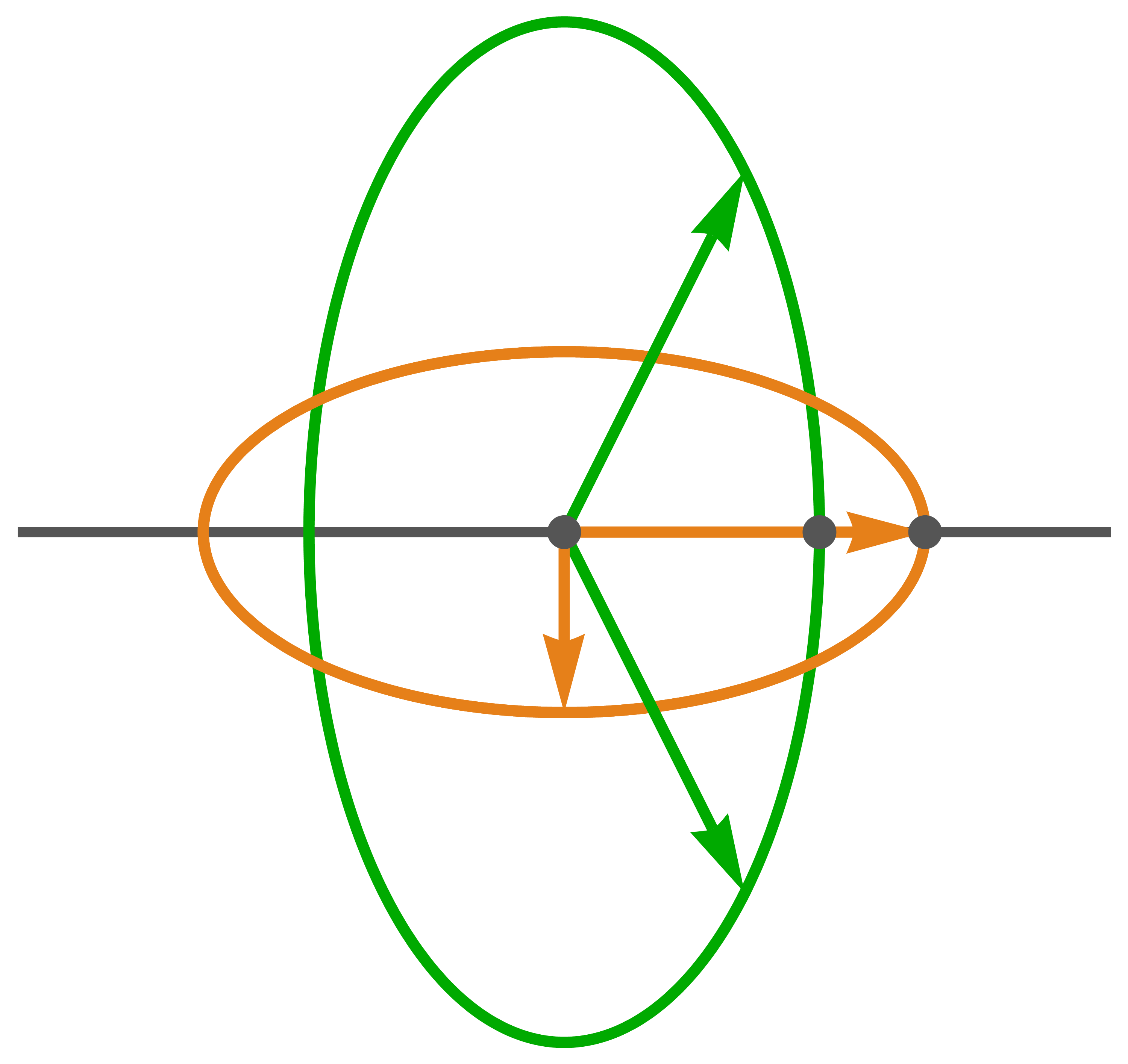}
\put(45,42){$\f o$}
\put(95,40){$L$}
\put(25,15){$\f p_1$}
\put(12,38){$\f p_2$}
\put(52,25){$\f a_1$}
\put(52,70){$\f b_1$}
\put(55,50){$\f b_2$}
\put(40,35){$\f a_2$}
\put(83,50){$\f p_2(t_0)$}
\put(54,41){$\f p_1(t_0)$}
\end{overpic}
\caption{Left: Symmetric trigonometric curve $\f p$ w.r.t. axis $L: y=0$. Right: Its decomposition into two ellipses $\f p_1$ and $\f p_2$. The conjugated half-diameters $\f a_1$, $\f b_1$ and $\f a_2$, $\f b_2$ of $\f p_1$ and $\f p_2$, respectively, together with the vertices $\f p_1(t_0)$ and $\f p_2(t_0)$ determining the axis of symmetry are shown.}\label{fig:axis_sym}
\end{center}
\end{figure}
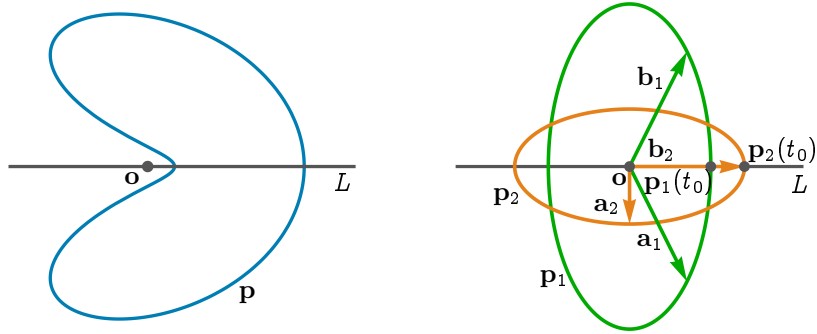

\begin{remark}\rm
Recall that the {\em Rytz axis construction} is an elementary  method from descriptive geometry that can be used to find the axes and the vertices of an ellipse (in our case $\f p_k$) starting from its two conjugated half-diameters (in our case determined by $\f a_k,\f b_k$). Based on this construction the parameter value corresponding to the vertex is computed via the formula
\begin{equation}\label{eq:t_k0}
  t_{k,0}  = \frac{1}{2k}\mathrm{arccot}\left(\frac{\f a_k\cdot\f a_k-\f b_k\cdot \f b_k}{2\f a_k\cdot \f b_k}\right).
\end{equation}
The remaining vertices of $\f p_k(t)$ are obtained for the parameter values $\{t_{k,0}+\frac{\pi}{2k}j\}$, $j=1,2,3$.  If there exists non-empty intersection of these sets for all non-zero $\f p_k$'s then this value is the sought $t_0$. 
\end{remark}

\begin{example}\rm
Consider a trigonometric curve given by the following parameterization 
$$
\f p(t)=[\sin (t)+2 \sin (2 t)+\cos (t),2 \sin (t)-2 \cos (t)-\cos (2 t)].
$$
It is a sum of two ellipses
$$
\f p_1(t) =  [1,-2] \cos (t) + [1,2] \sin (t) \quad \mathrm{and} \quad \f p_2(t)=   [0,-1] \cos (2 t) + [2,0] \sin (2 t).
$$
From \eqref{eq:t_k0}, all vertices of $\f p_1(t)$ occur at the following values of the parameter $\left\{\frac{\pi }{4},\frac{3 \pi }{4},\frac{5 \pi }{4},\frac{7 \pi }{4}\right\}$, whereas the vertices of $\f p_2(t)$ correspond to $\left\{0,\frac{\pi }{4},\frac{\pi }{2},\frac{3 \pi }{4}\right\}$. Hence $t_0 = \frac{\pi }{4}$  yields the sough-after axis of symmetry, see Fig.~\ref{fig:axis_sym}. Note, that $t_0 = \frac{3\pi }{4}$ corresponds to the same line (axis of symmetry) but with the opposite directional vector.
\end{example}

Of course, when an ellipse becomes a circle then any line through the center is its axis and we have infinitely many candidates for the vertex.  Hence if the parameterization $\f p(t)$ consists of some circles and at least one  ellipse then it is most reasonable to ignore the circles, find the tentative axis  and the parameter $t_0$ for this simplified parameterization and finally test whether the found symmetry is the symmetry of the original curve or not.

\medskip
The above discussed simplification does not help when all the $\f p_k$'s are circles. In this case we must determine the parameter values $t_0$  such that all $\f p_k(t_0)$ are linearly dependent. Motivated by the notion from astronomy, we will call this straight-line configuration of three or more points a {\em syzygy}, see Fig.~\ref{fig:syzygy}.

\begin{figure}[t]
\begin{center}
\begin{overpic}[height=0.35\textwidth]{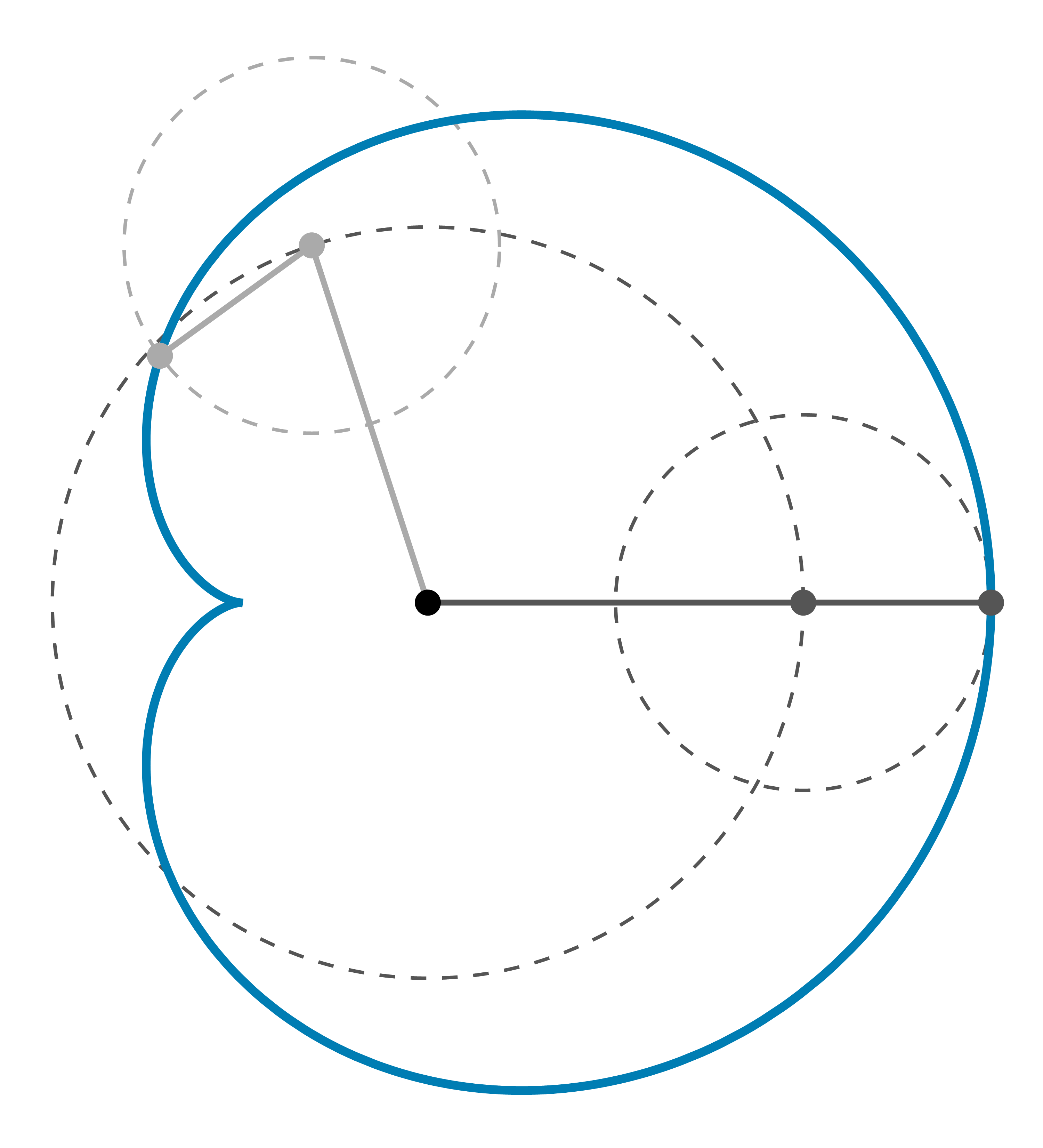}
\put(88,45){$\f p(t_0^1)$}
\put(50,13){$\f p_1(t)$}
\put(42,66){$\f p_2(t)+\f p_1(t_0^1)$}
\put(35,42){$\f o$}
\put(72,10){$\f p(t)$}
\put(90,70){\fcolorbox{gray}{white}{\includegraphics[width=0.1\textwidth]{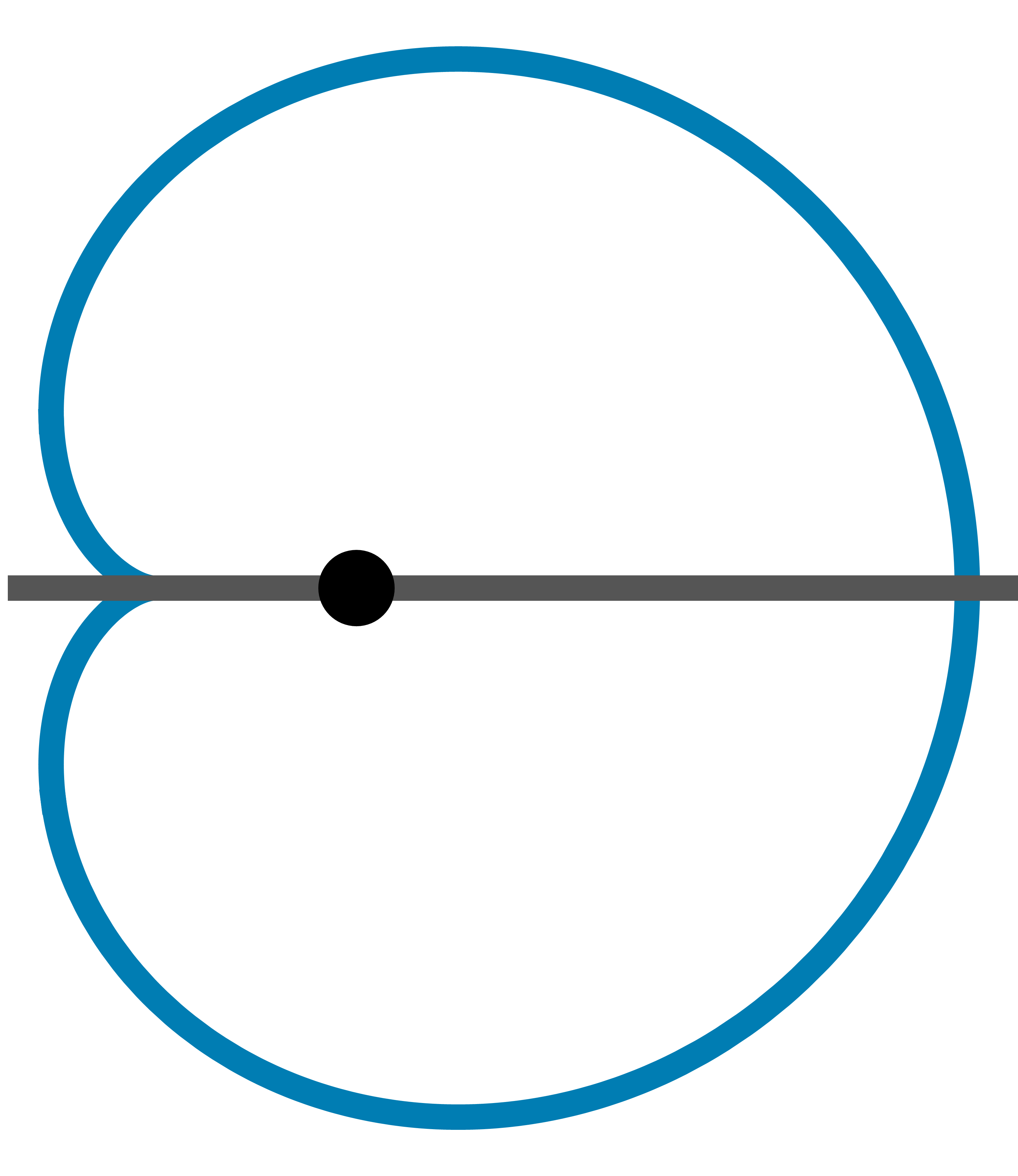}}}
\end{overpic}
\hspace{15ex}
\begin{overpic}[height=0.35\textwidth]{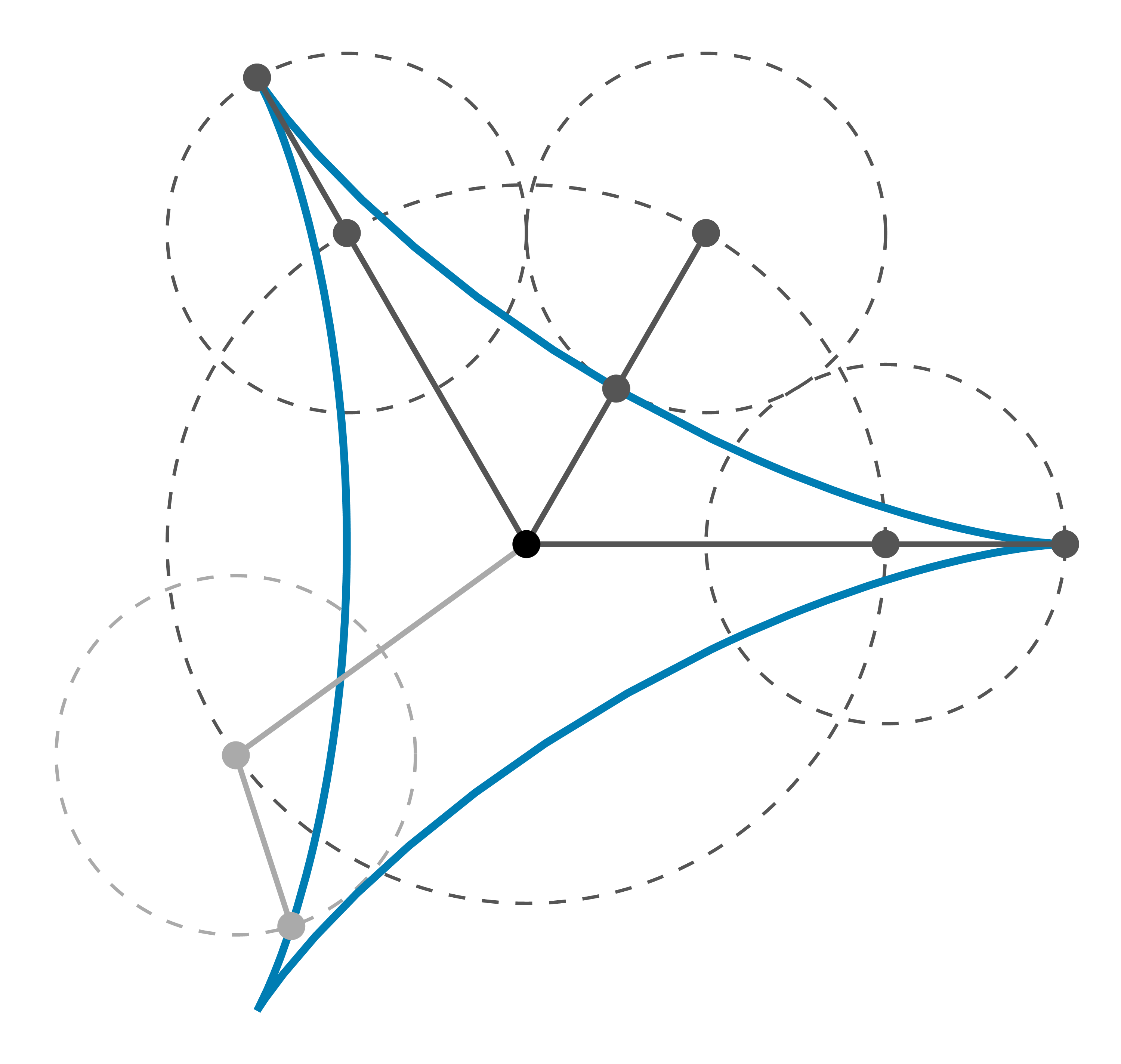}
\put(90,50){$\f p(t_0^1)$}
\put(58,62){$\f p(t_0^2)$}
\put(15,92){$\f p(t_0^3)$}
\put(60,13){$\f p_1(t)$}
\put(75,25){$\f p_2(t)+\f p_1(t_0^1)$}
\put(50,92){$\f p_2(t)+\f p_1(t_0^2)$}
\put(-10,55){$\f p_2(t)+\f p_1(t_0^3)$}
\put(45,41){$\f o$}
\put(50,25){$\f p(t)$}
\put(90,70){\fcolorbox{gray}{white}{\includegraphics[width=0.1\textwidth]{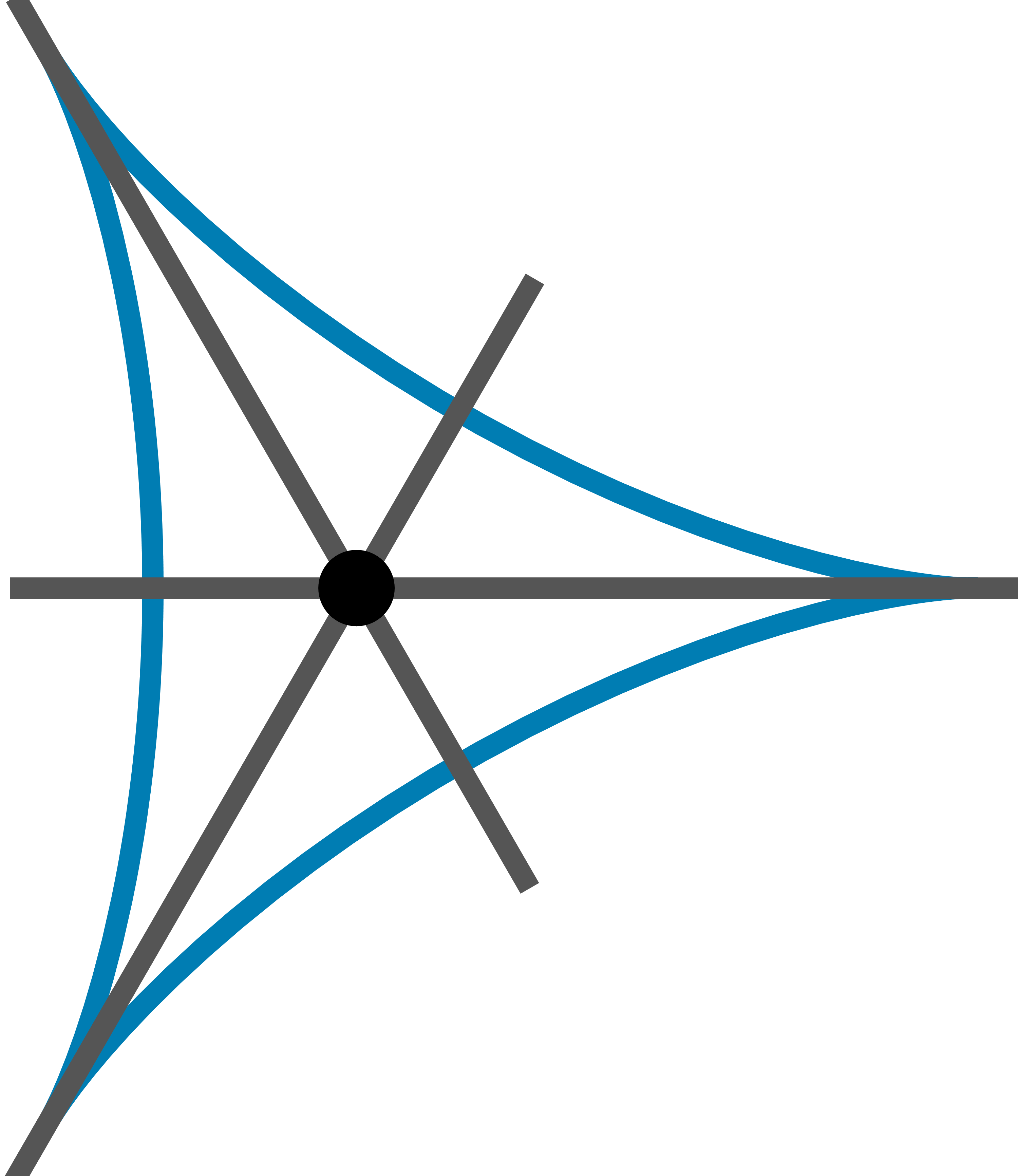}}}
\end{overpic}
\caption{The syzygy of the points on the~circles $\f p_k$ corresponds to the directions of the axes (framed) of symmetry. A generic configuration (non-syzygy) is also shown in light gray. Left: A cycloidal curve (cardioid) with one symmetry axis. Right: A cuspidal curve (deltoid) with three axes of symmetry.}\label{fig:syzygy}
\end{center}
\end{figure}
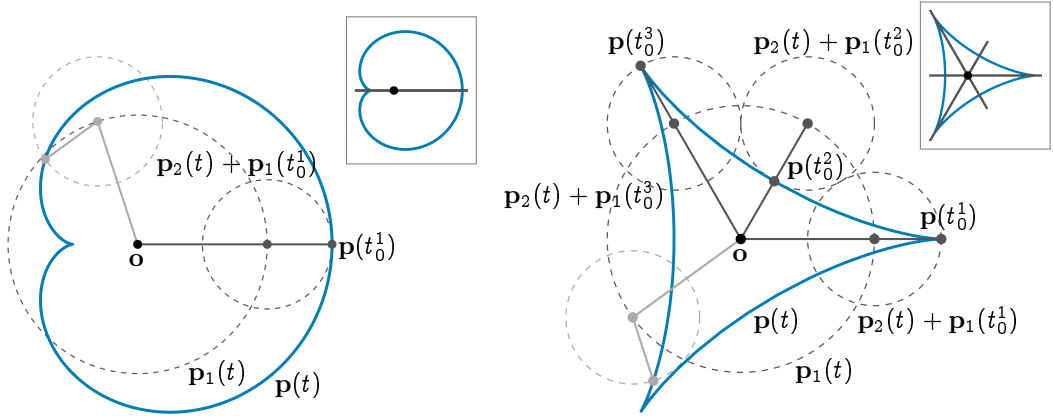

Consider the simplest example $\f p(t) = \lambda_{k_1}\e^{\I(\sigma_{k_1}k_1t+\psi_{k_1})}+\lambda_{k_2}\e^{\I(\sigma_{k_2}k_2t+\psi_{k_2})}$, where $\gcd(k_1,k_2)=1$. Then the condition that $\f p_{k_1}(t)$ and $\f p_{k_2}(t)$ are linearly dependent can be rewritten as
\begin{equation}
    \sigma_{k_1}k_1t+\psi_{k_1} =   \sigma_{k_2}k_2t+\psi_{k_2}+ j\pi
\end{equation}
for some $j\in\Z$, i.e., we  arrive at
\begin{equation}\label{eq:volba_tecek}
  t_0= \frac{\psi_{k_2}-\psi_{k_1}}{\sigma_{k_1}k_1-\sigma_{k_2}k_2}+j\frac{\pi}{\sigma_{k_1}k_1-\sigma_{k_2}k_2}.
\end{equation}
There exist $\sigma_{k_1}k_1-\sigma_{k_2}k_2$ possibilities for choosing $t_0$ providing different axis, see Fig.~\ref{fig:syzygy}  , left. If $\f p(t)$ contains more than two circles then the possible values of $t_0$ must be found as the common values arising from all the~pairs  of the~circles $\f p_k(t)$.

\begin{example}\rm
Consider a deltoid (a hypocycloid known also as a tricuspoid curve) given by the trigonometric parameterization 
$$
\f p(t)=[2\cos t + \cos(2t),2\sin t - \sin(2t)],\qquad \text{or}\qquad
\f p(t)=2\e^{t\I}+\e^{-2t\I}.
$$
It follows from \eqref{eq:volba_tecek} that we have three possible choices for $t_0$, in particular 
$$
t_0=\frac{\pi}{3}j,\quad j=0,1,2,
$$
and thus there exist three potential syzygy configurations, see Fig.~\ref{fig:syzygy}, right. The cycloidal curve (cardioid) possessing one axis of symmetry is shown in see Fig.~\ref{fig:syzygy}, left.
\end{example}

\begin{remark}\rm
We can return to the situations where $\f p_k$'s are ellipses and also in this case we will refer to the configuration guaranteeing the existence of the symmetry axis (given by a suitable value $t_0$, cf. Lemma~\ref{lem:ellipse_syzygy}) as a syzygy. We recall that for ellipses in this configuration not only all $\f p_k(t_0)$ are linearly dependent but in addition they are vertices for all $\f p_k$'s which are ellipses.
\end{remark}

\medskip
Finally, let us summarize the results of this section in a decision tree, see Diagram~\ref{fig:symmetry_groups}.

\tikzstyle{intt}=[draw,text centered,minimum size=6em,text width=5.25cm]
\tikzstyle{intl}=[draw,text centered,text width=4cm,text height=0.34cm]
\tikzstyle{int}=[draw,minimum size=2.5em,text centered,text width=4cm]
\tikzstyle{intg}=[draw,shape=rectangle,minimum size=3em,text centered,text width=5cm]
\tikzstyle{intb}=[draw,shape=rectangle,rounded corners=1.5ex,minimum size=3em,text centered,text width=4cm]
\tikzstyle{sum}=[draw,shape=circle,text centered]
\tikzstyle{summ}=[drawshape=circle,inner sep=4pt,text centered,node distance=3.cm]

{\renewcommand\figurename{Diagram}
\begin{figure}[!htb]
\centering
\begin{tikzpicture}[
 >=latex',
 auto
 ]
 \node [intg] (kp)  {$\exists k$ such that $\f p_k$ is an ellipse?};
 \node [int]  (ki1) [node distance=1.2cm and -0.2cm,below left=of kp,xshift=2mm] {Do all  $\f p_{2k}$'s vanish?};
 \node [intb]  (ki2) [node distance=1.2cm and -0.2cm,below right=of kp,xshift=-2mm] {Find the maximal $(m,d)$-sequence.};
 \node [intl]  (ki11)[node distance=3cm and 1cm,below left of=ki1,xshift=-2mm] {$\exists$ a syzygy configuration?};
 \node [intl]  (ki12)[node distance=3cm and 1cm,below right of=ki1,xshift=2mm] {$\exists$ a syzygy configuration?};
 \node [intl]  (ki21)[node distance=2.05cm and 0.6cm,below of=ki2] {$\exists$ a syzygy configuration?};
 \node [sum]  (ki111)[node distance=2.8cm and 0cm,below left of=ki11,xshift=1cm] {$D_2$};
 \node [sum]  (ki112)[node distance=2.8cm and 0cm,below right of=ki11,xshift=-1cm] {$C_2$};
 \node [sum]  (ki121)[node distance=2.8cm and 0cm,below left of=ki12,xshift=1cm] {$D_1$};
 \node [sum]  (ki122)[node distance=2.8cm and 0cm,below right of=ki12,xshift=-1cm] {$C_1$};
 \node [sum]  (ki211)[node distance=2.8cm and 0cm,below left of=ki21,xshift=1cm] {$D_m$};
 \node [sum]  (ki212)[node distance=2.8cm and 0cm,below right of=ki21,xshift=-1cm] {$C_m$};

 \draw[->] (kp) -- ($(kp.south)+(0,-0.75)$) -| (ki1) node[above,pos=0.25] {Yes} ;
 \draw[->] (kp) -- ($(kp.south)+(0,-0.75)$) -| (ki2) node[above,pos=0.25] {No};
 \draw[->] (ki1) -- ($(ki1.south)+(0,-0.75)$) -| (ki11) node[above,pos=0.25] {Yes} ;
 \draw[->] (ki1) -- ($(ki1.south)+(0,-0.75)$) -| (ki12) node[above,pos=0.25] {No};
 \draw[->] (ki2) -- ($(ki2.south)+(0,-0.75)$) -| (ki21) node[above,pos=0.25] {} ;
 \draw[->] (ki11) -- ($(ki11.south)+(0,-0.75)$) -| (ki111) node[above,pos=0.25] {Yes} ;
 \draw[->] (ki11) -- ($(ki11.south)+(0,-0.75)$) -| (ki112) node[above,pos=0.25] {No};
 \draw[->] (ki12) -- ($(ki12.south)+(0,-0.75)$) -| (ki121) node[above,pos=0.25] {Yes} ;
 \draw[->] (ki12) -- ($(ki12.south)+(0,-0.75)$) -| (ki122) node[above,pos=0.25] {No};
 \draw[->] (ki21) -- ($(ki21.south)+(0,-0.75)$) -| (ki211) node[above,pos=0.25] {Yes} ;
 \draw[->] (ki21) -- ($(ki21.south)+(0,-0.75)$) -| (ki212) node[above,pos=0.25] {No};
\end{tikzpicture}
\caption{Symmetry groups of curves with trigonometric parameterization \eqref{eq:C_m}.}
    \label{fig:symmetry_groups}
\end{figure}
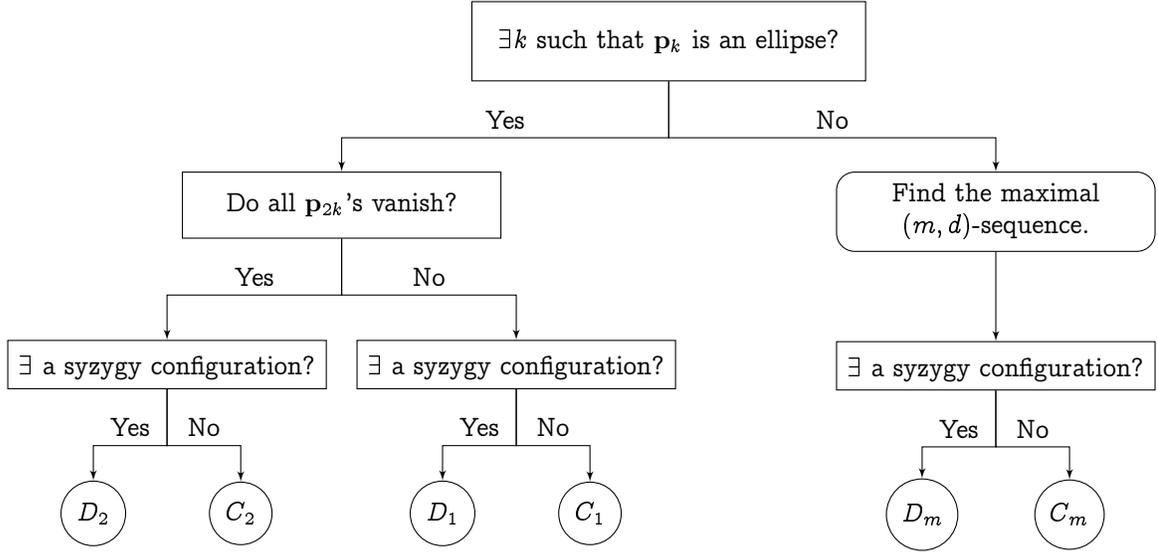
}

\section{Application to the symmetries of discrete curves}\label{sec:sym_polyline}

By a~ \emph{discrete curve} in $\R^2$ we understand a closed polyline defined by the set of its vertices $\f v_0,\f v_1,\ldots,\f v_n=\f v_0$.  If $C$ is the discrete curve then its symmetry $\phi\in\sym(C)$ is the isometry of $\R^2$ such that $\phi(C)=C$ and $\phi(V)=V$, where $V$ is the ordered set of its vertices. Notice that if any three consecutive vertices are not collinear (which can always be ensured by removing the middle vertex) then the latter condition is fulfilled automatically.

In what follows we will replace the discrete curve by a suitable trigonometric curve such that the conditions of Lemma~\ref{lem:komut} are satisfied in order to  detect the symmetries of $C$. A natural candidate is the trigonometric curve interpolating the vertices $\f v_i$. Let us fix the uniform distribution of parameter $t_j=\frac{2\pi}{n}j$, $j=0,\dots,n-1$. Then for  odd $n$ we obtain the unique interpolant $\f p(t)$, see Section~\ref{sec:trigcurves}  and \eqref{eq:interpol_coef_odd}. However when $n$ is even then there exists a one-parametric set of solutions, see again Section~\ref{sec:trigcurves}.  Because of the choice of the uniformly spaced $t_j$, we have $\sin(Nt_j)=0$ and thus $b_N$, where $N=\frac{n}{2}$, is the~free parameter. Following the convention from Section~\ref{sec:trigcurves} we set $b_N=0$ and compute the coefficients of the interpolant $\f p(t)$ via \eqref{eq:interpol_coef_even}. In both cases let $T_C$ denote the curve parameterized by $\f p(t)$. For later use it is essential to prove that the curve parameterized by $\f p(t)$ does not depend on the choice of the first vertex being interpolated and neither on the orientation.

\begin{lemma}\label{lem:TC}
  Let be given a trigonometric curve $C$, then $T_C$ is well defined.
\end{lemma}

\begin{proof}
We start with the odd case.  Let $\f p(t)$ be the interpolant corresponding to the ordered vertex set $\f v_0,\f v_1,\ldots$ and let $\f q(t)$ be another interpolant associated to the shifted vertex set $\f v_i,\f v_{i+1},\ldots$. Then  $\f p(t+\frac{2\pi i}{n})$ interpolates the shifted vertex in the sense that $\f p(t_j)=\f v_{i+j}$. Because of the uniqueness of the interpolants we see that $\f q(t)$ is a reparameterization of $\f p(t)$ and thus they define the same curve $T_C$. The argument for the change of the orientation is the same.

For $n$ being even, we have to realize first that because of the uniform distribution of $t$ it holds
$$
\cos\left[N\left(t+\frac{2\pi i}{n}\right)\right]=\cos\left[\frac{n}{2}\left(t+\frac{2\pi i}{n}\right)\right]=
\cos(Nt+\pi i)=(-1)^{i-1}\cos(Nt).
$$
Hence, the type of the trigonometric parameterization (distinguished by the property $b_N=0$) is preserved by this reparameteriation. The next argumentation is analogous to the odd case, i.e., we again obtain the same parameterizations $\f p\left(t+\frac{2\pi i}{n}\right)$ and $\f q(t)$. 
\end{proof}

\begin{corollary}\label{cor:subset}
Let $C$ be a discrete curve then   $\sym(C)\subset\sym(T_C)$.
\end{corollary}

\begin{proof}
  It follows from  Lemma~\ref{lem:TC} and Lemma~\ref{lem:komut}.
\end{proof}

It can happen that $\sym(T_C)$ is strictly bigger then $\sym(C)$, see Fig.~\ref{fig:narust_a_pokles_symetrie}, left.  In general, both groups are finite, and thus comparable, unless the curve $T_C$ is a straight line or a circle. Since we assume in what follows that the vertices $\f v_i$ are not collinear and the curve $T_C$ interpolates them, this can never be a straight line. 



\begin{figure}[t]
\begin{center}
\hspace{-10ex}
\begin{overpic}[height=0.2\textwidth]{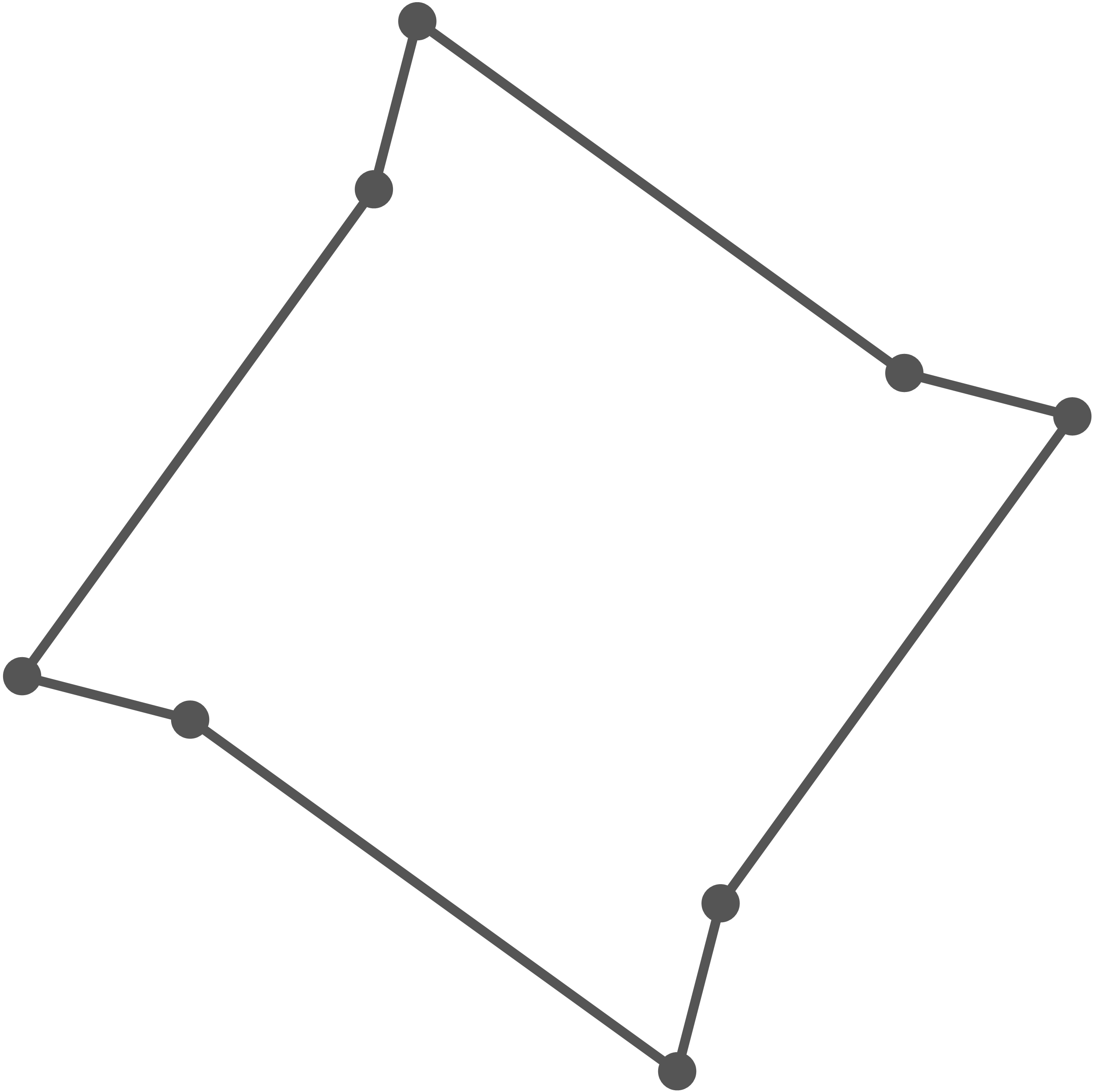}
\put(81,81){\fcolorbox{gray}{white}{\includegraphics[width=0.05\textwidth]{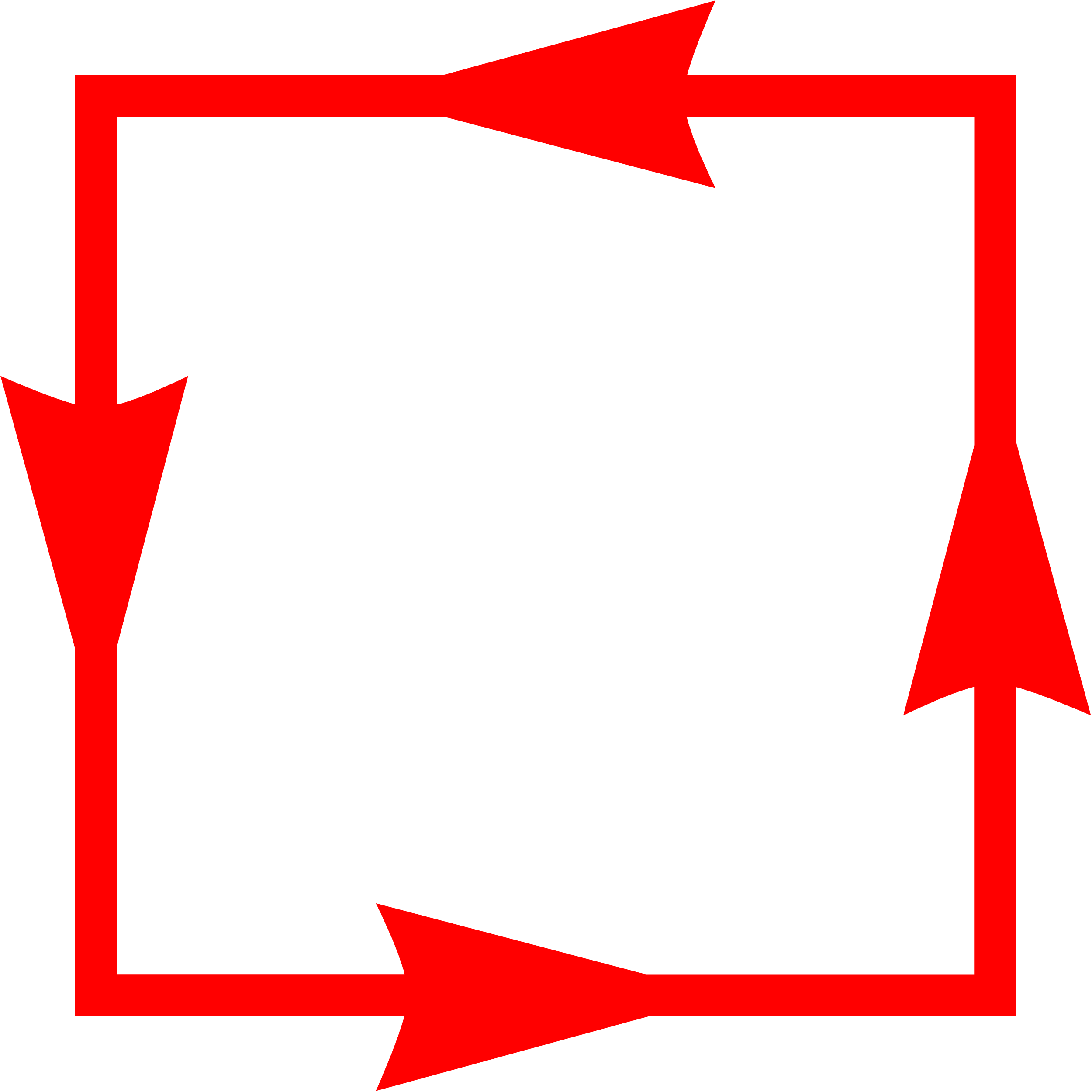}}}
\end{overpic}
\hspace{5ex}
\begin{overpic}[height=0.2\textwidth]{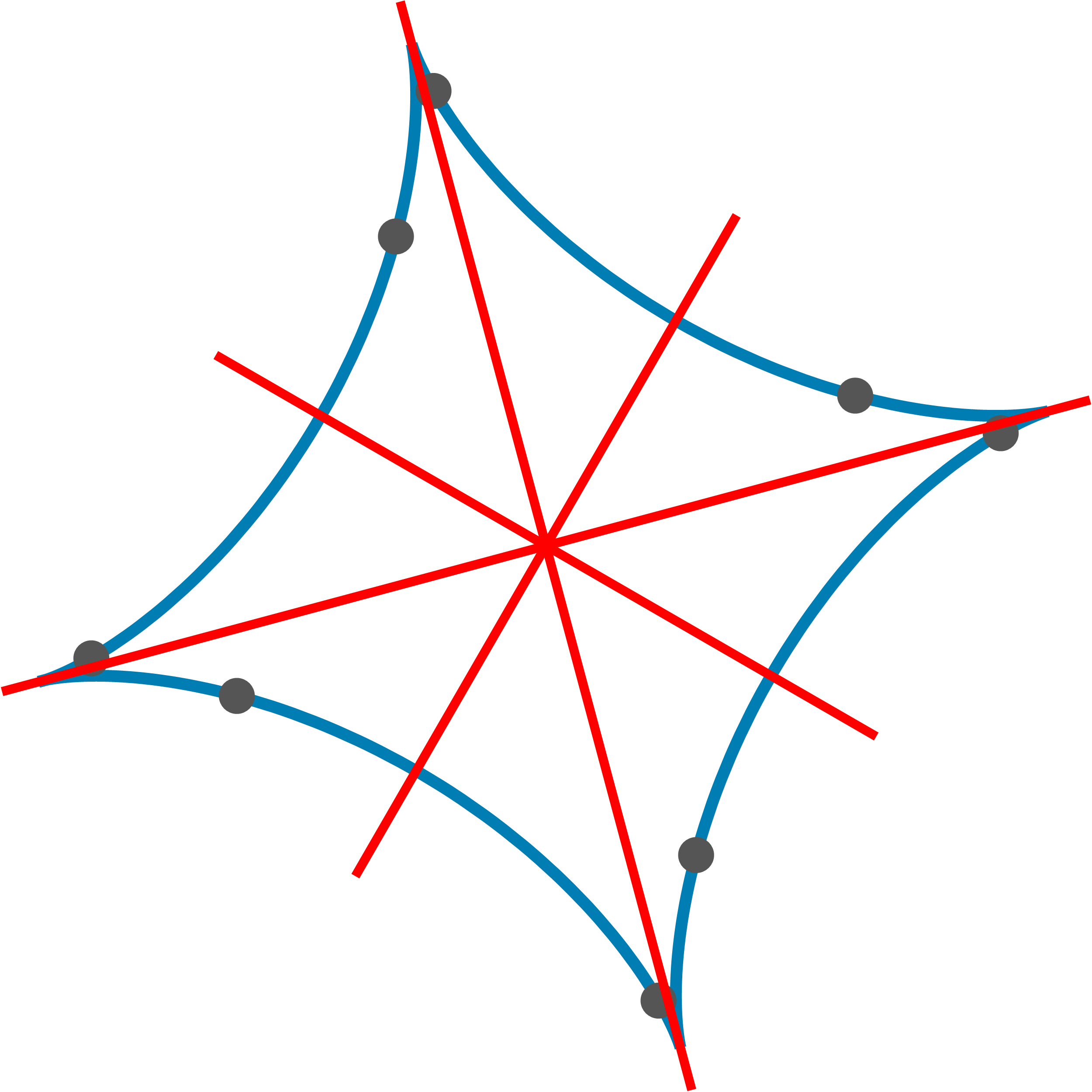}
\end{overpic}
\hspace{3ex}
\vrule width1pt
\hspace{3ex}
\begin{overpic}[height=0.13\textwidth]{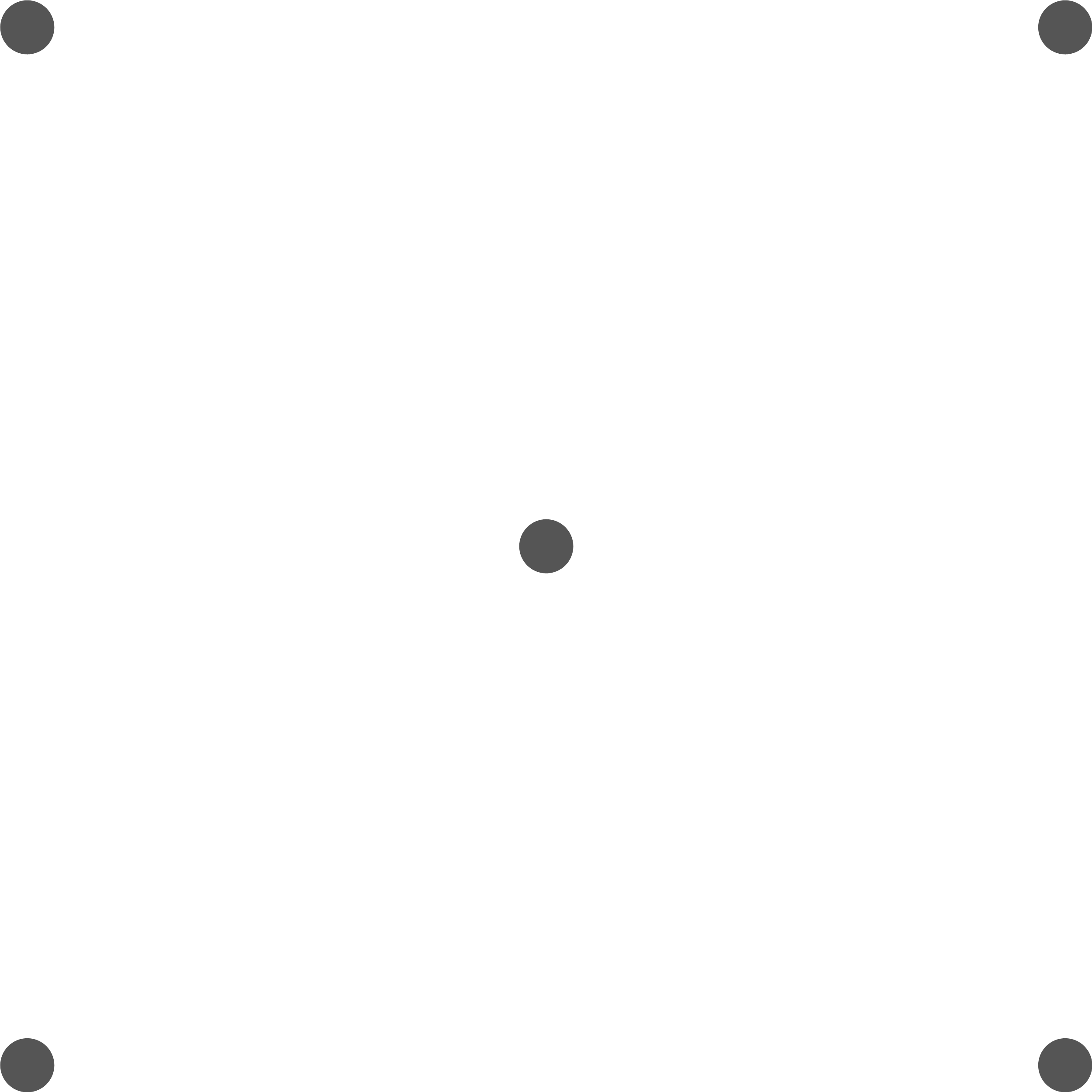}
\put(106,100){\fcolorbox{gray}{white}{\includegraphics[width=0.05\textwidth]{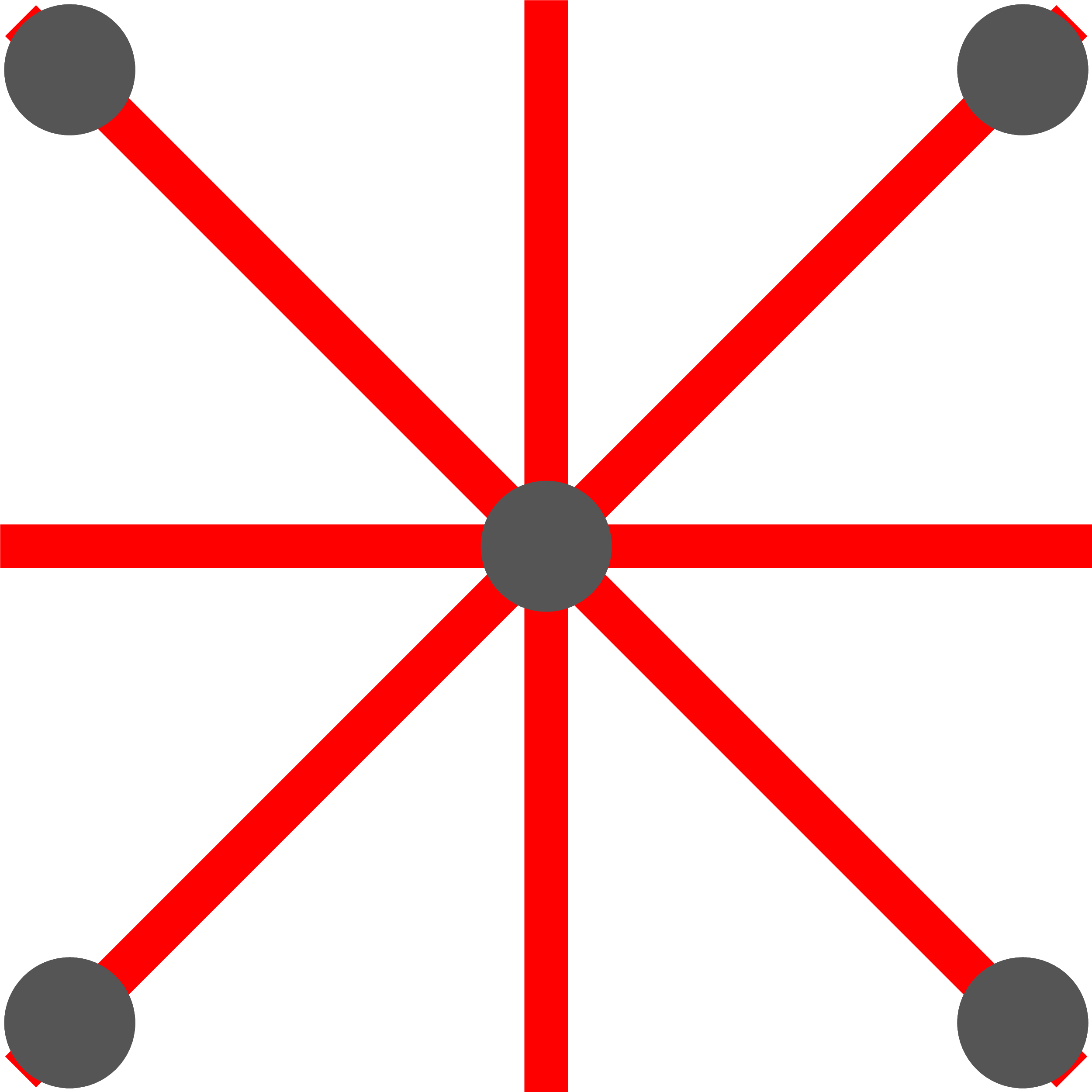}}}
\end{overpic}
\hspace{10ex}
\begin{overpic}[height=0.13\textwidth]{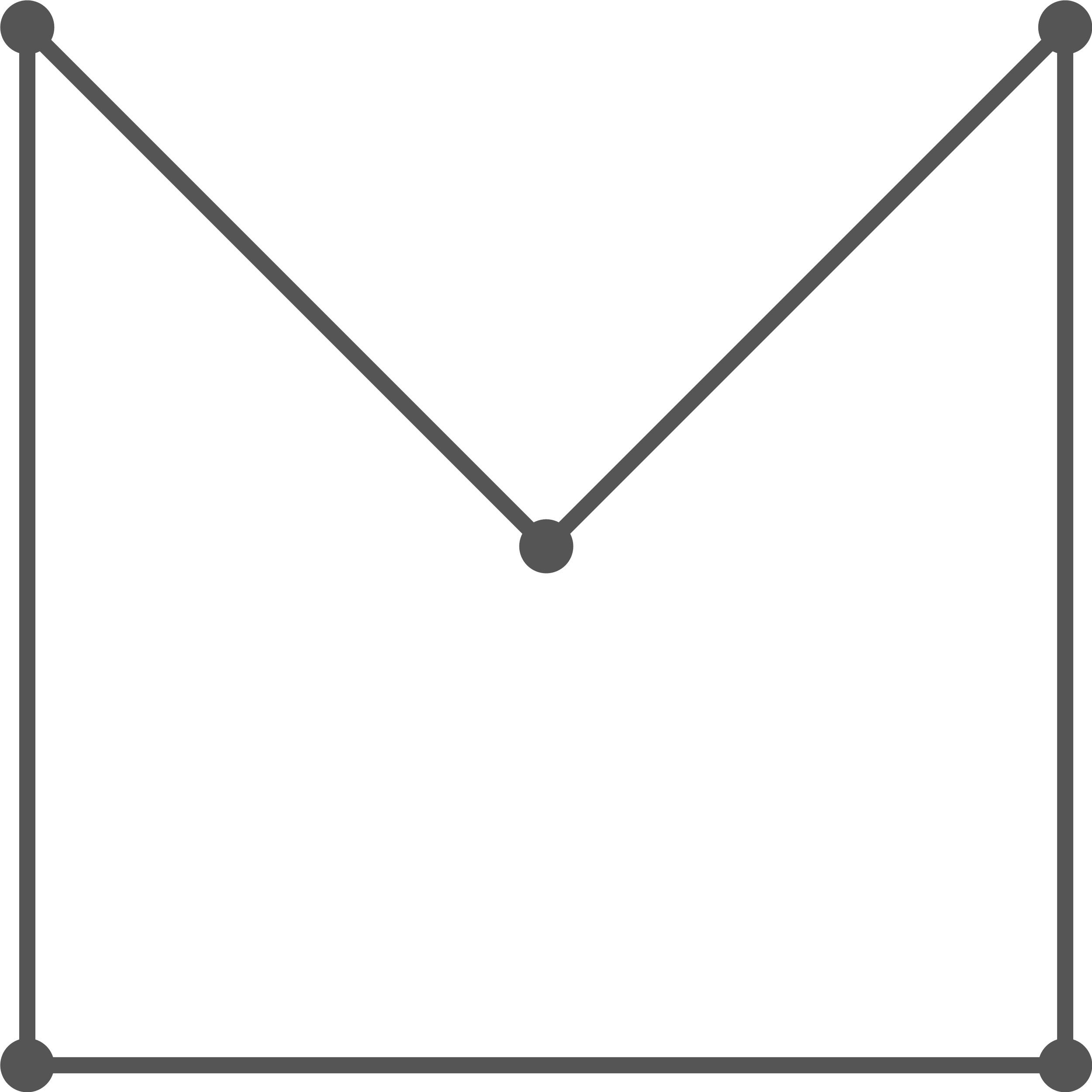}
\put(106,100){\fcolorbox{gray}{white}{\includegraphics[width=0.05\textwidth]{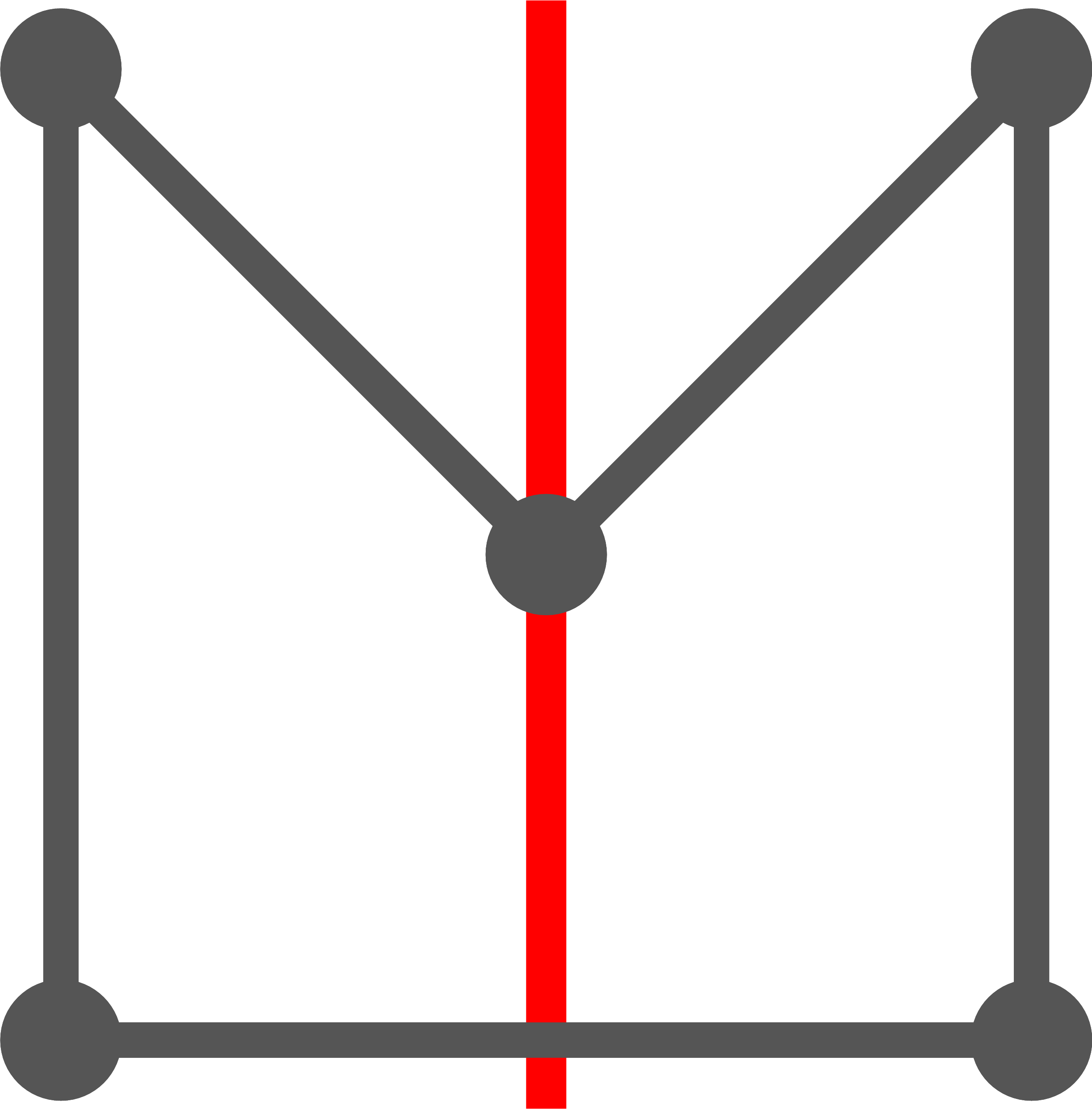}}}
\end{overpic}
\caption{Increasing and decreasing symmetries. Left: A discrete curve $C$  with $\sym(C)\cong C_4$ and the associated curve $T_C$ with the~bigger symmetry group $D_4$. Right: The given set of points with the symmetry group $D_4$ and the interpolating discrete curve with the symmetry group $D_1$.}\label{fig:narust_a_pokles_symetrie}
\end{center}
\end{figure}


\begin{lemma}\label{lem:infinite sym}
  $T_C$ is a circle if and only if the discrete curve $C$ is a~regular $n$-gon.
\end{lemma}

\begin{proof}
 If $C$ is a~regular $n$-gon then the circle clearly interpolates the vertices. Because of the uniqueness of the interpolant the curve $T_C$ is forced to be a circle.  Conversely if $T_C$ is a circle then the uniform distribution of the parameter implies that the points $\f v_i$ form the consecutive vertices of a~regular  $n$-gon. 
\end{proof}

Fortunately if $T_C$ is not a circle or a straight line then $\sym(T_C)$ is finite and thus $\sym(C)$ can be identified with a subgroup of $\iso({\R^2})$. The whole method for  determining symmetries of closed discrete curves is summarized in Algorithm~\ref{algoritmus1}.

\begin{algorithm}[ht]
\caption{Symmetries of closed discrete curves.}
\label{algoritmus1}
\begin{algorithmic}[1]
\medskip
\Require
A discrete curve $C$ given by the ordered set of its vertices $V=\{\f v_0,\ldots,\f v_{n-1}\}$.

\smallskip
\State
Compute the trigonometric interpolant $\f p(t)$ of $C$ via \eqref{eq:interpol_coef_odd} or \eqref{eq:interpol_coef_even}, which parameterizes $T_C$.    

\smallskip
\State
Determine $\sym(T_C)$ via the decision tree presented in Diagram~\ref{fig:symmetry_groups}.

\smallskip
\State
Include into the set $\sym(C)$ only such $\phi\in\sym(T_C)$ that satisfy $\phi(C)=C$ (and $\phi(V)=V$).

\medskip
\Ensure
$\sym(C)$.
\end{algorithmic}
\end{algorithm}

\begin{remark}\rm\label{rem:alg_step3}
It remains to comment on how to decide which symmetries from $\sym(T_C)$ to choose to $\sym(C)$, i.e., how to perform {\tt Step 3} of the algorithm. We will use the following simple test as the set of vertices $V$ is ordered. First we identify $\f v_j$ such that $\phi(\f v_1)=\f v_j$. If such $\f v_j$ does not exist then $\phi$ is not a symmetry of $C$. In the affirmative case we continue to compare the next vertices and their images, i.e., to include $\phi$ into $\sym(C)$ it has to hold either $\phi(\f v_2)=\f v_{j+1},\phi(\f v_3)=\f v_{j+2},\phi(\f v_4)=\f v_{j+3},\ldots $ for direct isometries, or $\phi(\f v_2)=\f v_{j-1},\phi(\f v_3)=\f v_{j-2},\phi(\f v_4)=\f v_{j-3},\ldots $ for indirect isometries (all indices are considered modulo $n$).
\end{remark}

\begin{figure}[t]
\begin{center}
\includegraphics[width=0.22\textwidth]{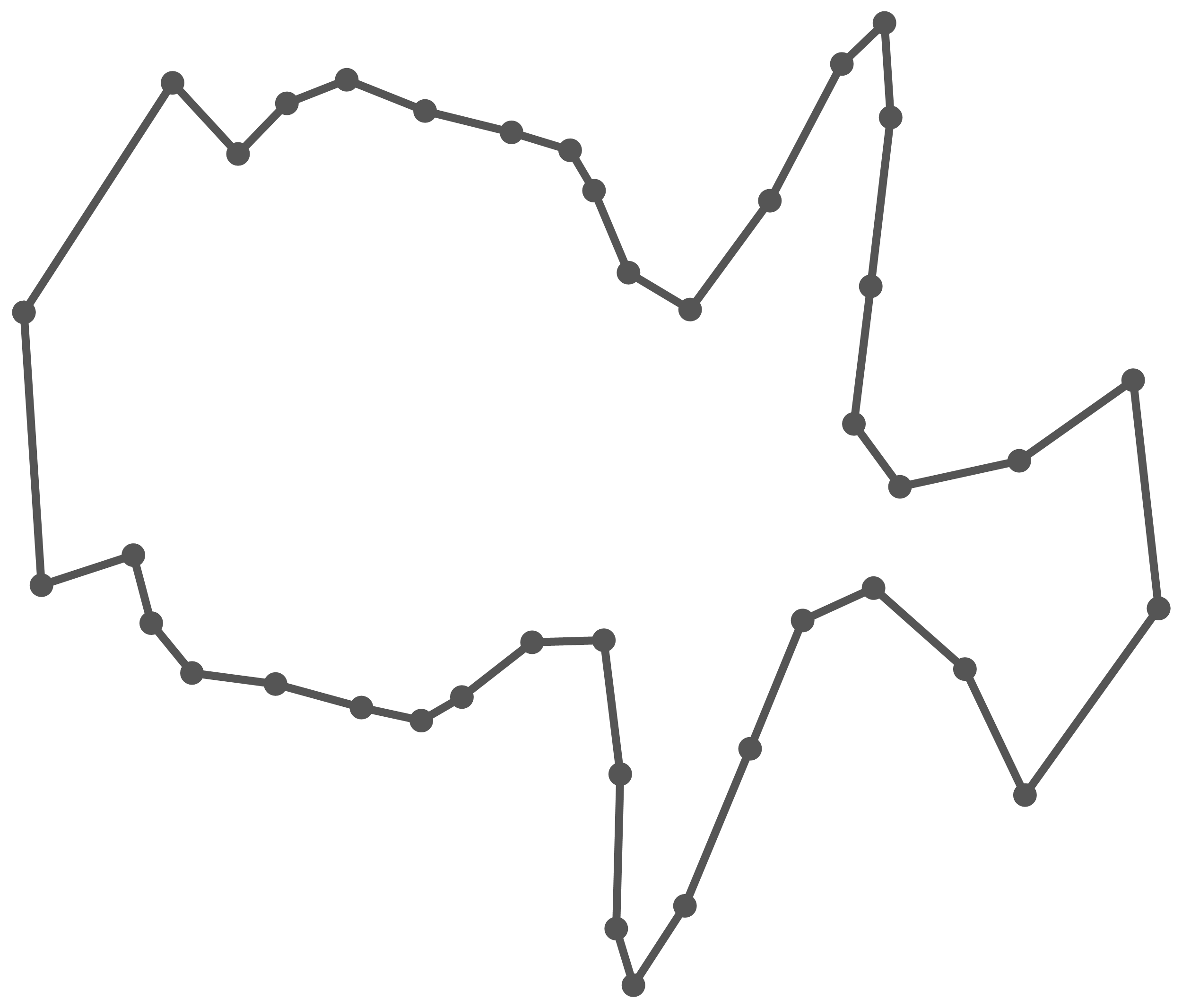}
\hspace{2ex}
\includegraphics[width=0.23\textwidth]{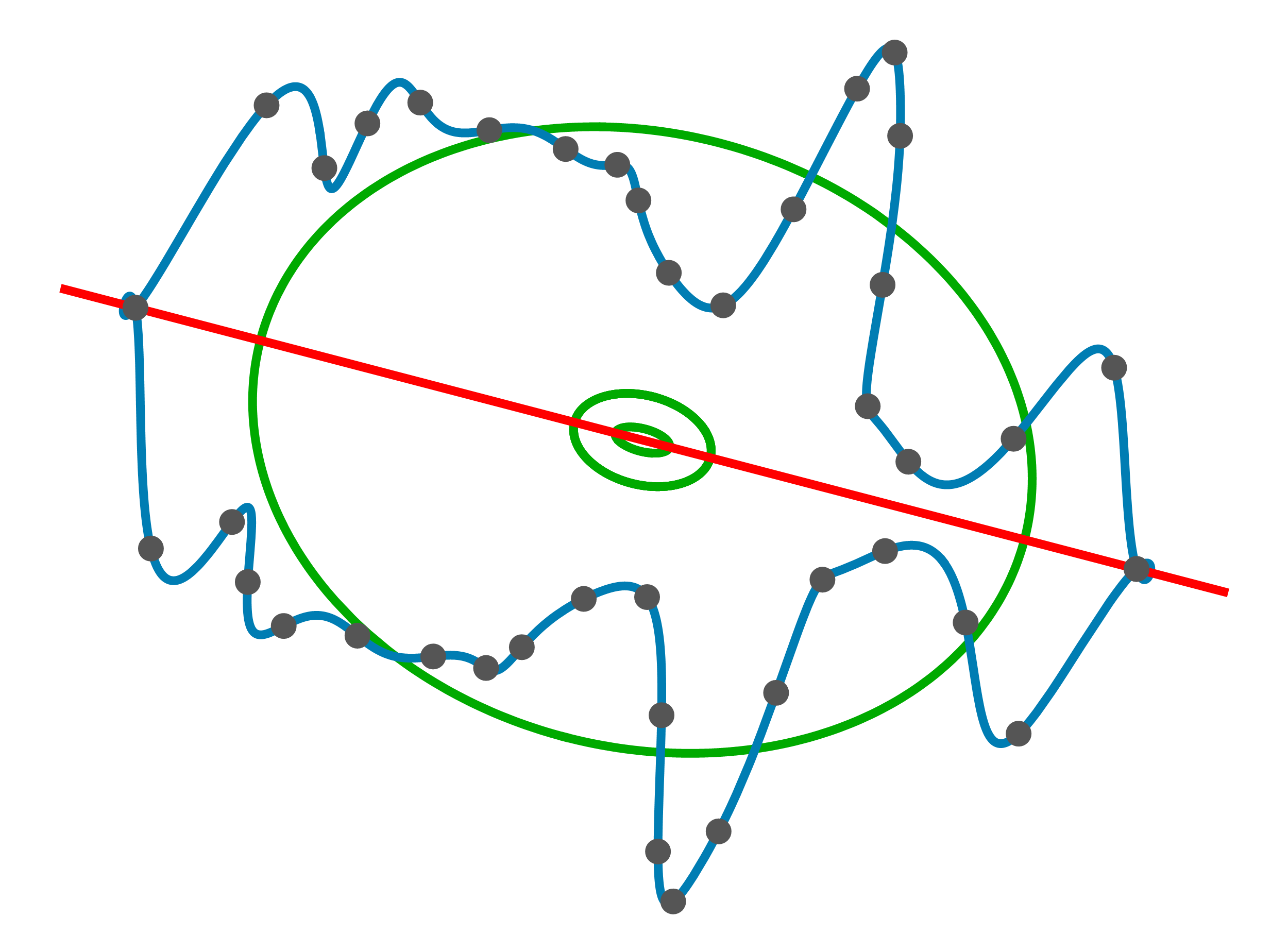}
\hfill
\vrule width1pt
\hfill
\includegraphics[width=0.22\textwidth]{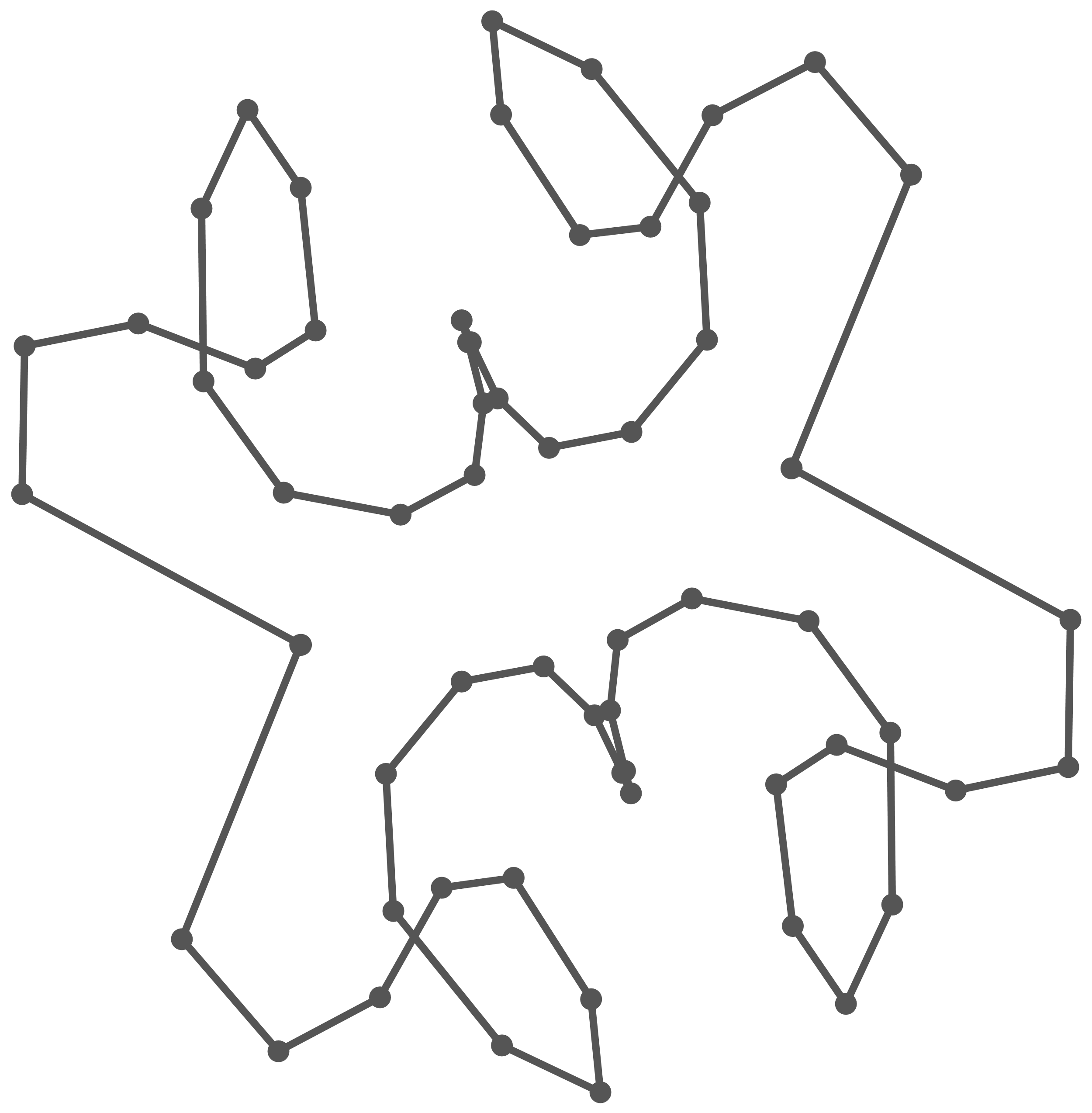}
\hspace{2ex}
\includegraphics[width=0.23\textwidth]{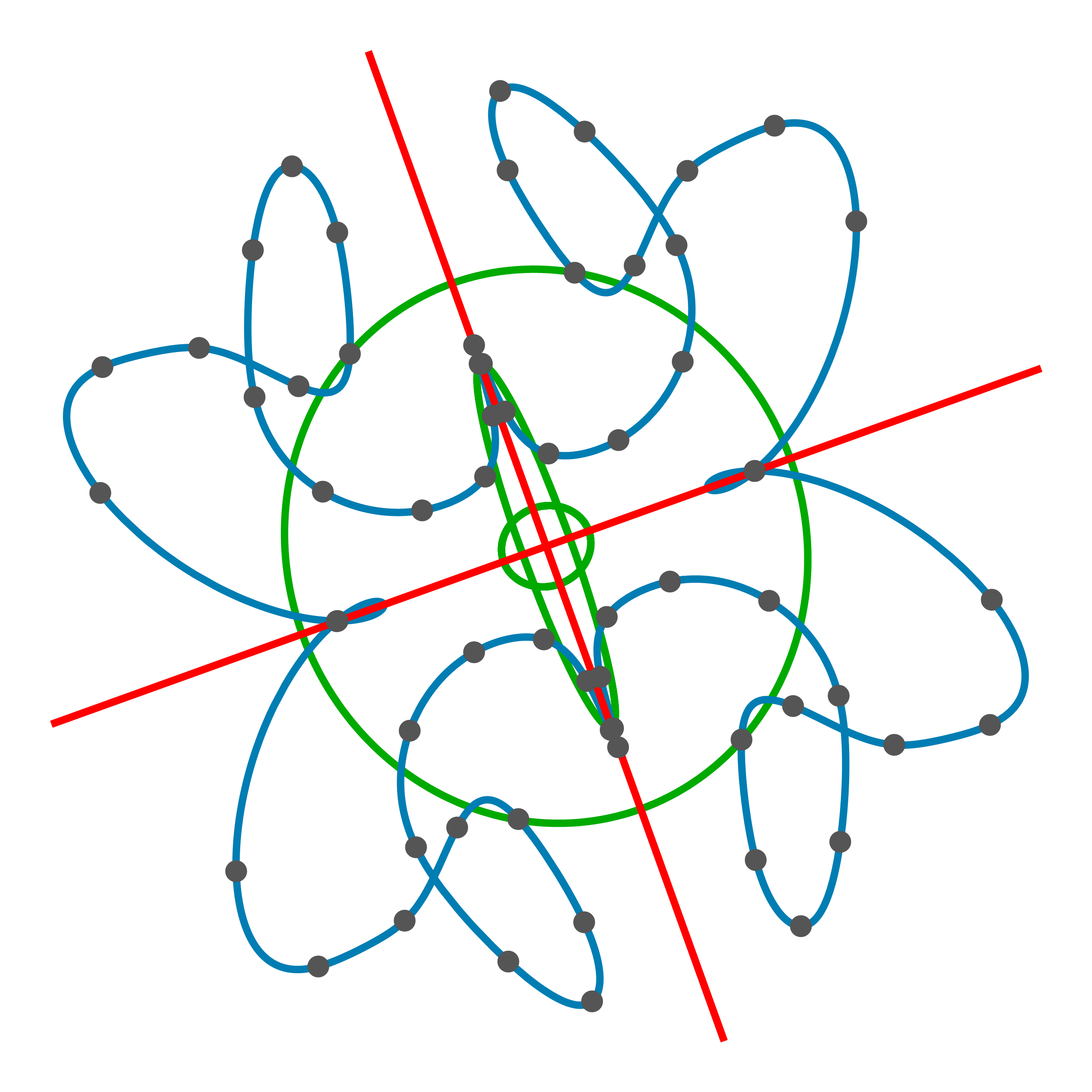}
\caption{Left: Discrete curve (gray) with the axial symmetry $D_1$ and trigonometric interpolant (blue) of their vertices (possessing the same symmetry).  Right: Discrete curve (gray) with the  symmetry group $D_2$ and the corresponding trigonometric interpolant (blue) with the same symmetry group. The first three ellipses $\f p_1, \f p_2, \f p_3$/ $\f p_1, \f p_3, \f p_5$ (green) and the axis/axes of symmetry (red) are also shown.}\label{fig:interpolant_axis}
\end{center}
\end{figure}



\begin{figure}[tbh]
\begin{center}
\hspace{-2ex}
\begin{overpic}[height=0.22\textwidth]{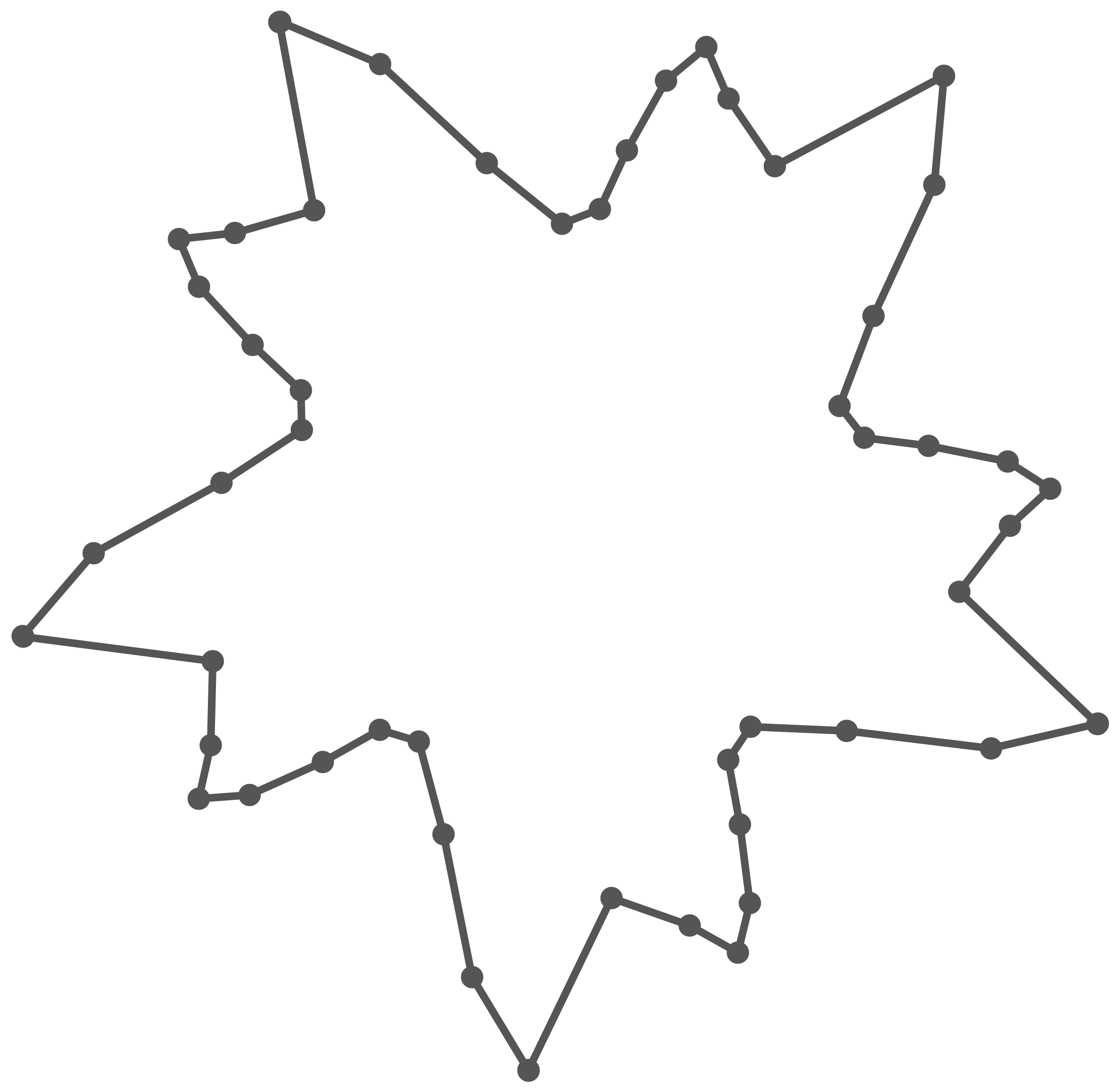}
\end{overpic}
\hspace{0ex}
\begin{overpic}[height=0.22\textwidth]{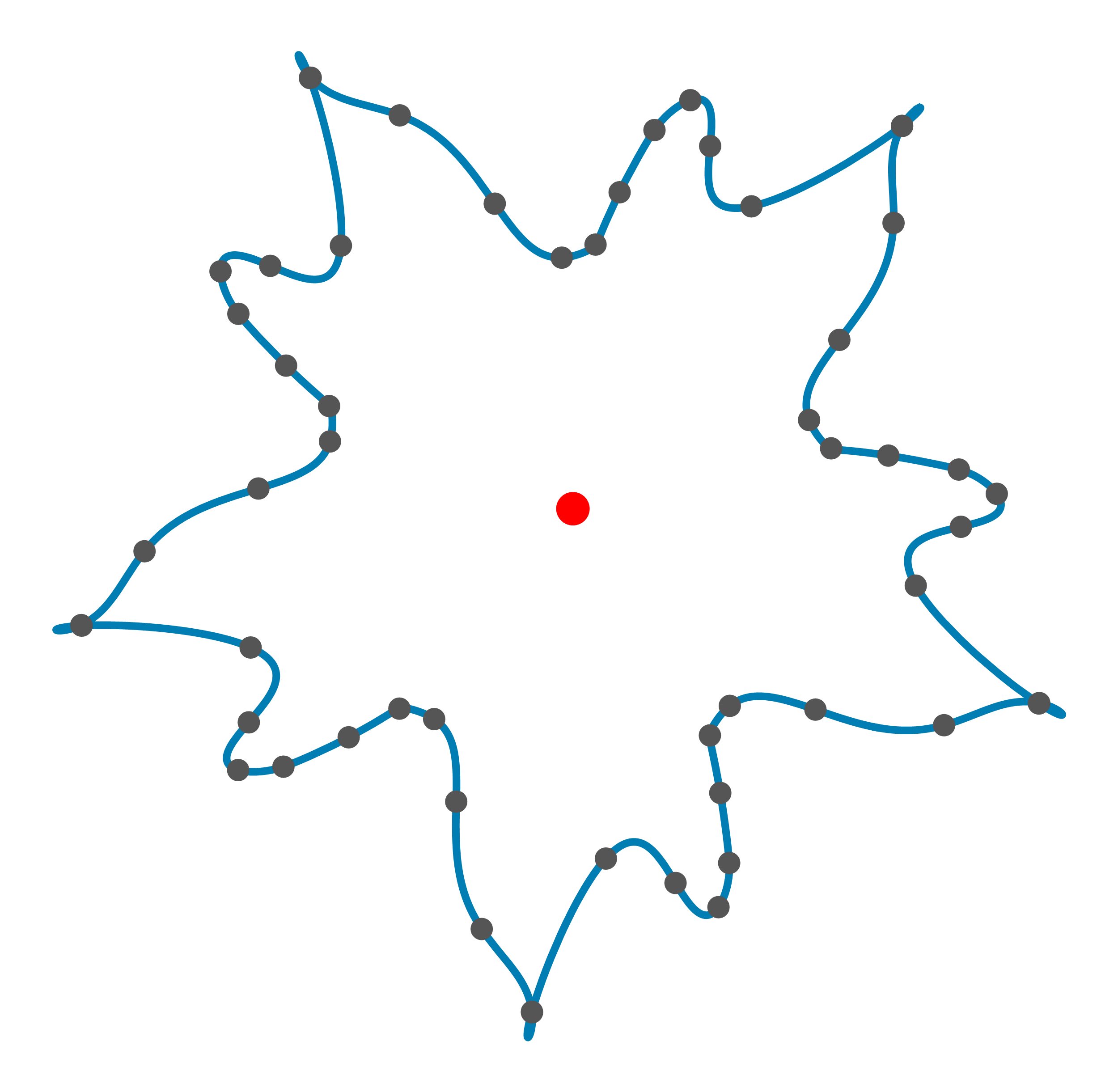}
\put(90,80){\fcolorbox{gray}{white}{\includegraphics[width=0.05\textwidth]{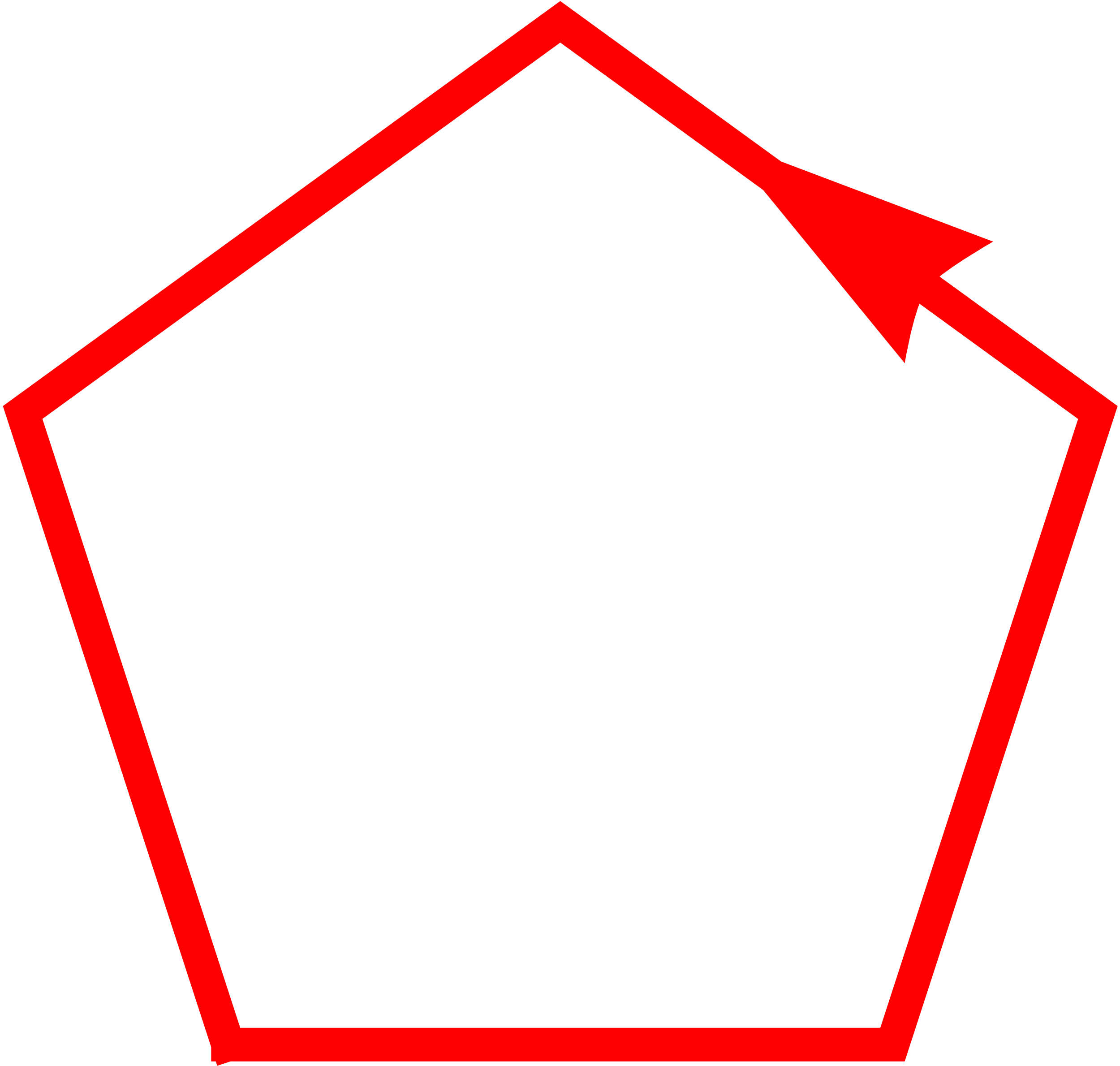}}}
\end{overpic}
\hfill
\vrule width1pt
\hfill
\begin{overpic}[height=0.22\textwidth]{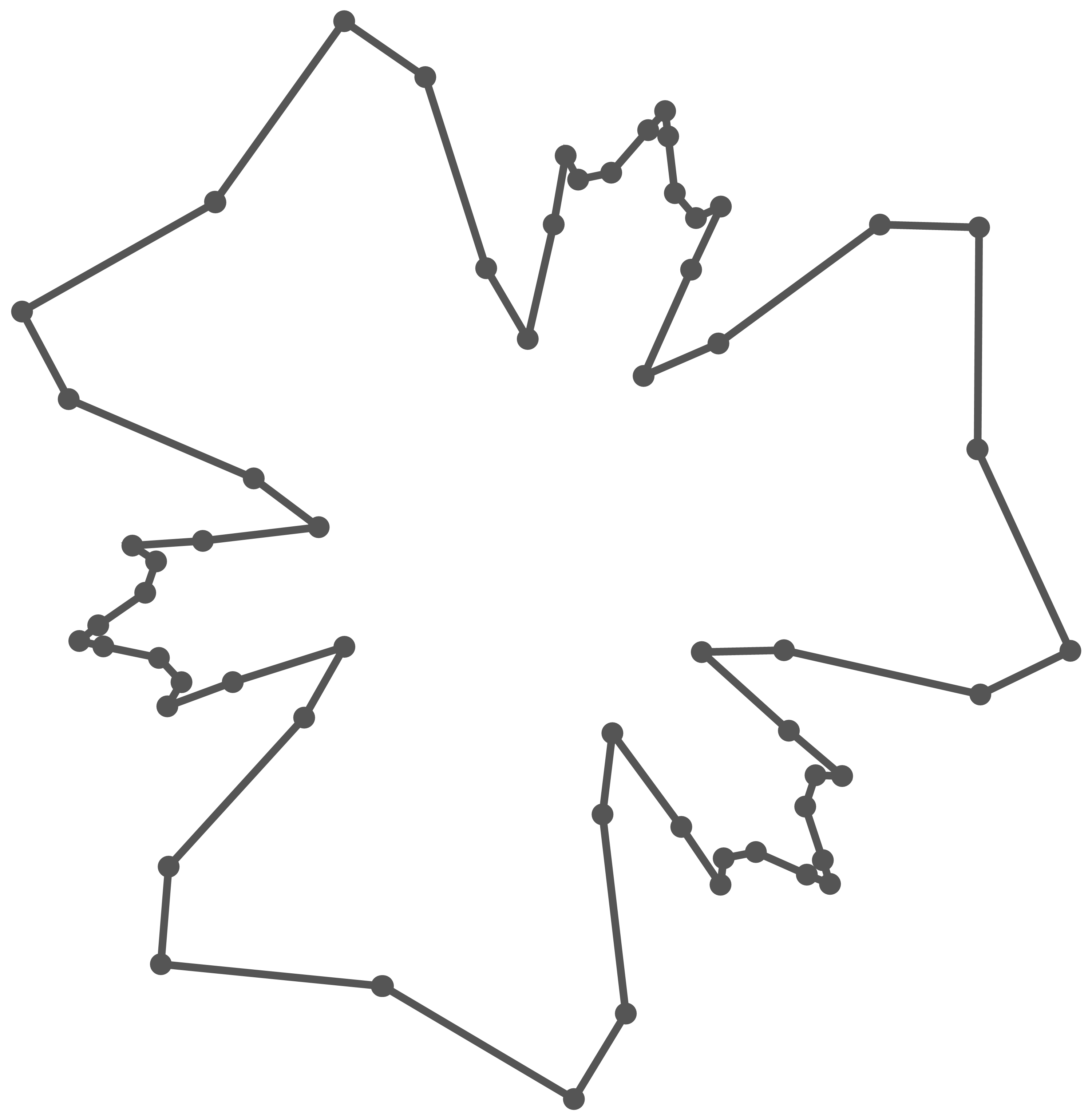}
\end{overpic}
\hspace{0ex}
\begin{overpic}[height=0.22\textwidth]{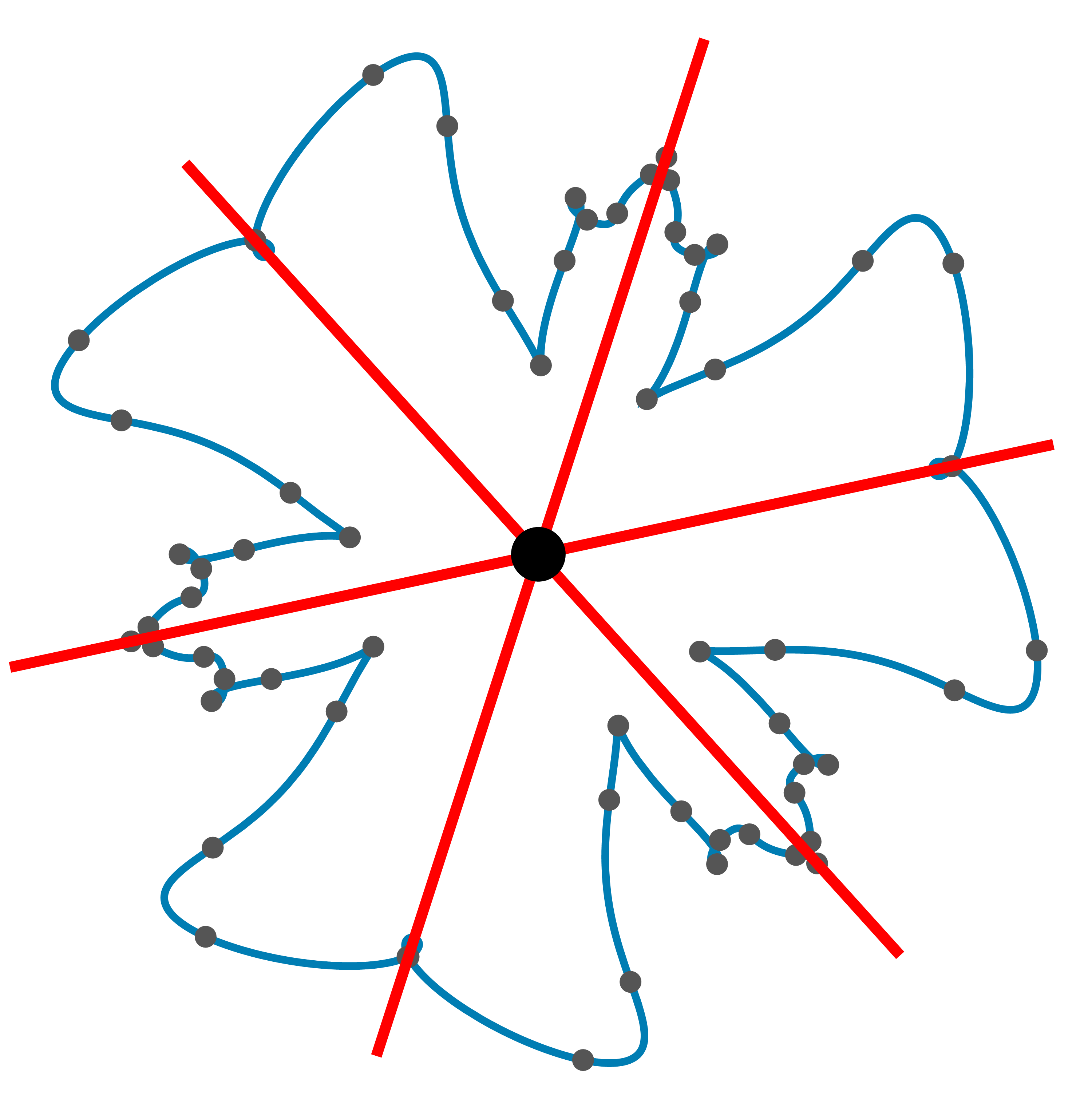}
\end{overpic}
\caption{Discrete curves (gray) with the symmetry groups $C_5$ and $D_3$ (red), and the trigonometric interpolants (blue) of their vertices with the same symmetry groups.}\label{fig:interpolant_rot}
\end{center}
\end{figure}

\begin{example}\rm
Consider four discrete curves $K_1$, $K_2$, $K_3$ and $K_4$. Following the steps of Algorithm~\ref{algoritmus1} we interpolate their vertices by the trigonometric curves $T_{K_1}$, $T_{K_2}$,  $T_{K_3}$ and $T_{K_4}$ and then determine their symmetry groups. 

The parameterizations of $T_{K_1}$ and $T_{K_2}$ contain $\f p_k$'s which are ellipses and moreover among them there are terms of even degree only in the case of $T_{K_1}$. In addition, since in both cases there exist  syzygy configurations we conclude that $T_{K_1}$ and $T_{K_2}$ possess the symmetry groups $D_1$ and $D_2$, respectively, see Fig.~\ref{fig:interpolant_axis}. Finally, it is verified that $\sym(K_1)=\sym(T_{K_1})$ and $\sym(K_2)=\sym(T_{K_2})$.

The curves $T_{K_3}$ and $T_{K_4}$ are higher cycloids with the associated sequences
$$
\begin{array}{rcl}
\sigma_{T_{K_3}} & = & 
\left( 
\{1\}, \emptyset, \emptyset, \{-1\}, \emptyset, \{1\}, \emptyset, 
\emptyset, \{-1\}, \emptyset, \{1\}, \emptyset, \emptyset, \{-1\},
\emptyset, \{1\}, \emptyset, \emptyset, \{-1\}, \emptyset, \{1\}, 
\emptyset, \emptyset, \{-1\}, \emptyset, \{1\}, \emptyset
 \right)\\[0.7ex]
& \preccurlyeq &  \vartheta^{5,1}=\left(\{1\}, \emptyset, \emptyset, \{-1\}, \emptyset, \ldots
 \right),
\end{array} 
$$
and
$$
\begin{array}{rcl}
\sigma_{T_{K_4}} & = & \left( 
\{1\}, \{-1\}, \emptyset, \{1\}, \{-1\}, \emptyset, \{1\}, \{-1\},
\emptyset, \{1\}, \{-1\}, \emptyset, \{1\}, \{-1\}, \emptyset, \{1\},
\{-1\}, \emptyset, \{1\}, \{-1\}, \emptyset,  \right. \\ 
& & \left. \{1\}, \{-1\}, \emptyset, \{1\}, \{-1\}, \emptyset,
\{1\}, \{-1\}, \emptyset, \{1\}
 \right)\\[0.7ex]
& \preccurlyeq &  \vartheta^{3,1}=\left(\{1\}, \{-1\}, \emptyset, \ldots
 \right).
\end{array} 
$$
Hence we have arrived at rotational symmetries of a regular 5- and 3-gon, respectively.

In the first case there does not exist any syzygy configuration, whereas in the second one we find three such configurations. Hence we conclude that the symmetry groups are $C_5$ and $D_3$, respectively, see Fig. \ref{fig:interpolant_rot}. Finally, it is confirmed that the discrete curves $K_3$ and $K_4$ have the same symmetry groups.
\end{example}

\begin{remark}\rm
The curve $T_C$ might possess too high degree with regard to practical purposes and further computations. Then it is advisable to work with suitable alternative curves instead. These curves can be obtained via {\em filtering} high harmonics of the original parameterization. Write $T^\ell_C$ for the {\em filtered curve} parameterized by $\f a_0+\sum_{k=1}^{N-\ell}\f p_k(t)$, i.e., the last $\ell$ terms of $\f p(t)$ are omitted. Analogously to Corollary~\ref{cor:subset} it can be shown that it holds
\begin{equation}
  \sym(T_C)\subset\sym(T^1_C)\subset\sym(T^2_C)\subset \cdots \subset \sym(T^{\ell}_C).
\end{equation}
This is a way how to deal with curves of manageable degrees. Nevertheless one has to be aware of a possible growth of the symmetry group when the filtering process is applied, see Fig.~\ref{fig:filtrace}

\begin{figure}[h!]
\begin{center}
\includegraphics[width=0.22\textwidth]{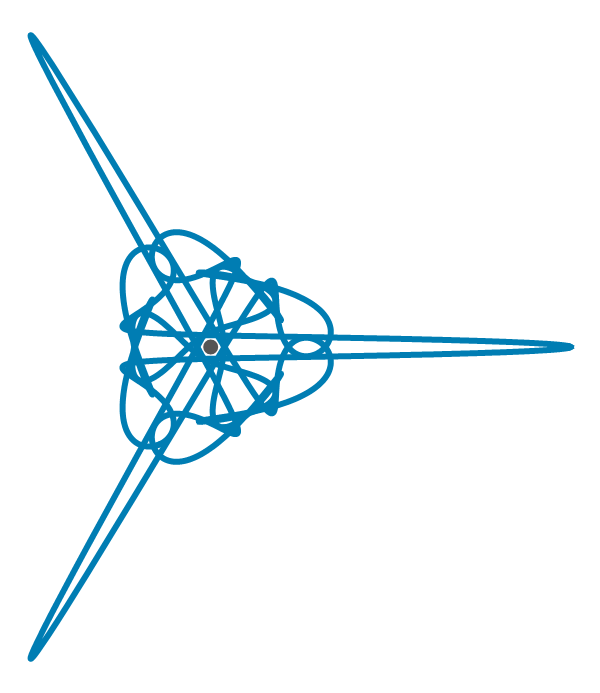}
\includegraphics[width=0.22\textwidth]{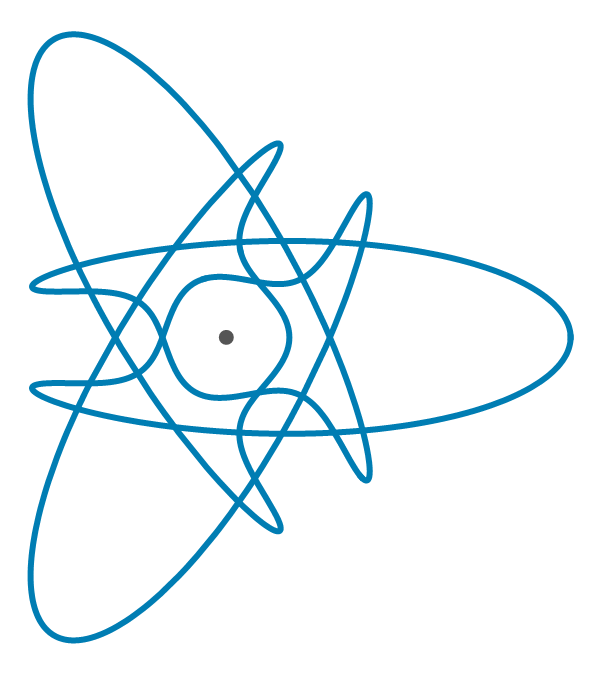}
\includegraphics[width=0.22\textwidth]{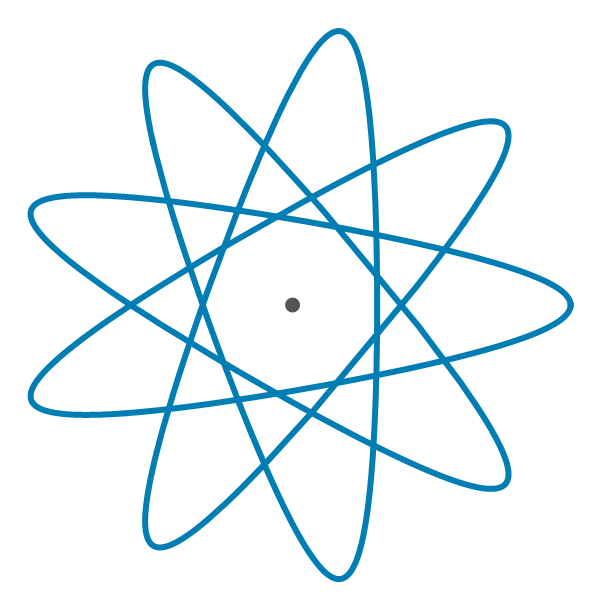}
\includegraphics[width=0.22\textwidth]{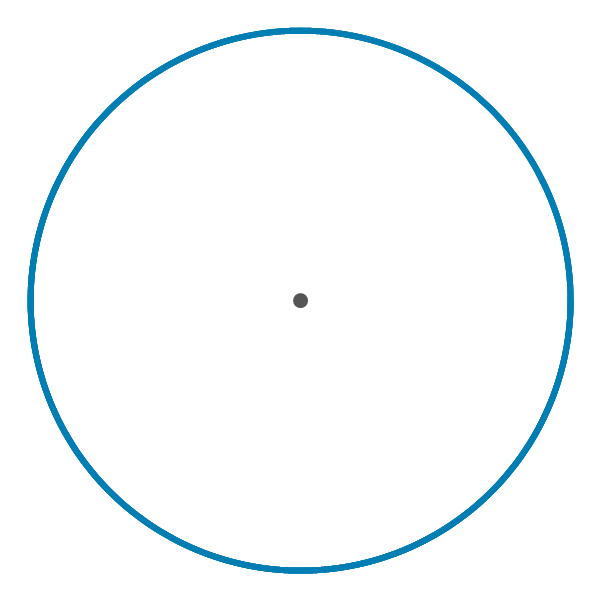}
\caption{The sequence of filtered curves demonstrating the growth of symmetry. From left to right: Curve $T_C$ with maximal trigonometric degree 20 and $\sym(T_C)\cong D_3$, $\sym(T_C^{10})\cong D_3$, $\sym(T_C^{14})\cong D_9$, and $\sym(T_C^{16})\cong \OO(2)$.}\label{fig:filtrace}
\end{center}
\end{figure}
 
\end{remark}

\section{Clouds of points}\label{sec:convexhull}

The previous section was devoted to discrete curves determined by their vertices, i.e., by ordered sets of points.
In what follows we address the problem where input data in plane is unorganized. Since these data might not be suitable for interpolation by a curve, or the interpolation can lead to loosing symmetries, see Fig.~\ref{fig:narust_a_pokles_symetrie} (right), we use an alternative approach based on the construction of the convex hull. By this we can efficiently solve the problem when symmetries of unorganized finite sets of points are to be find.


\begin{lemma}\label{lem:CHsubset}
Consider $\mathcal{X}$ as the set of all finite subsets of $\R^2$ and $\Psi$ as the map assigning to each set $X$ its convex hull $\mathrm{CH}(X)$. Then $\sym(X)\subset \sym(\mathrm{CH}(X))$.
\end{lemma}

\begin{proof}
Clearly, the convex hull is the affine, and thus also Euclidean invariant, so by Lemma~\ref{lem:komut} the property holds.
\end{proof}

The boundary of the convex hull of a finite set is a discrete curve. In addition, it obviously holds

\begin{lemma}\label{lem:polygonsubset}
Let $P$ be a polygon in plane and $\partial\, P$ be its boundary. Then $\sym(P)\subset\sym(\partial\, P)$.
\end{lemma}

Let be given a finite subset $X$ of $\R^2$. Then combination of Lemmas \ref{lem:CHsubset}, \ref{lem:polygonsubset} and Corollary \ref{cor:subset} yields the following chain  
\begin{equation}
\sym(X)\subset \sym(\mathrm{CH}(X))\subset \sym(\partial\,\mathrm{CH}(X))\subset \sym(T_{\partial\,\mathrm{CH}(X)}).
\end{equation}
Hence, we assign to the unorganized finite set of points the boundary of its convex hull and then we determine all symmetries of this closed discrete curve by the previously presented methods.

\smallskip
Let us emphasize again that the obtained symmetry group $\sym(T_{\partial\,\mathrm{CH}(X)})$ can be strictly larger then the sought group $\sym(X)$. However, in this case we cannot apply the approach discussed in Remark \ref{rem:alg_step3}, which is suitable for ordered point sets only. Instead, we determine the \emph{Hausdorff distance} between two finite sets $X,Y$ given by the formula
\begin{equation}
   \delta_H(X,Y)=\max_{x\in X}\{  \min_{y\in Y}\{\|xy\|\}\}.
\end{equation}
Then $\phi\in\sym(T_{\partial\,\mathrm{CH}(X)})$ belongs to $\sym(X)$ if and only if $\delta_H(X,\phi(X))=0$.

\begin{figure}[t]
\begin{center}
\includegraphics[width=0.2\textwidth]{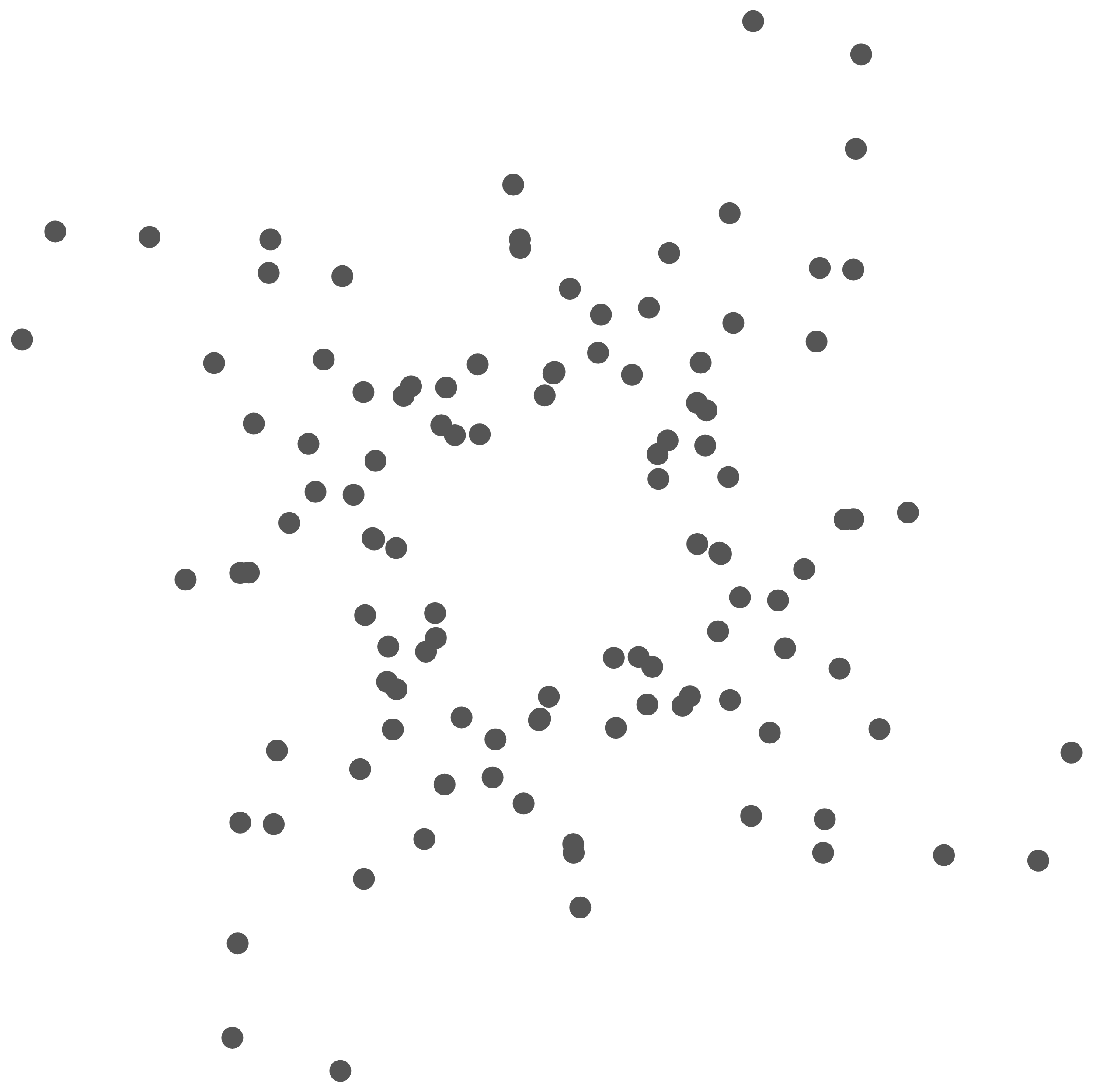}
\hspace{2ex}
\includegraphics[width=0.2\textwidth]{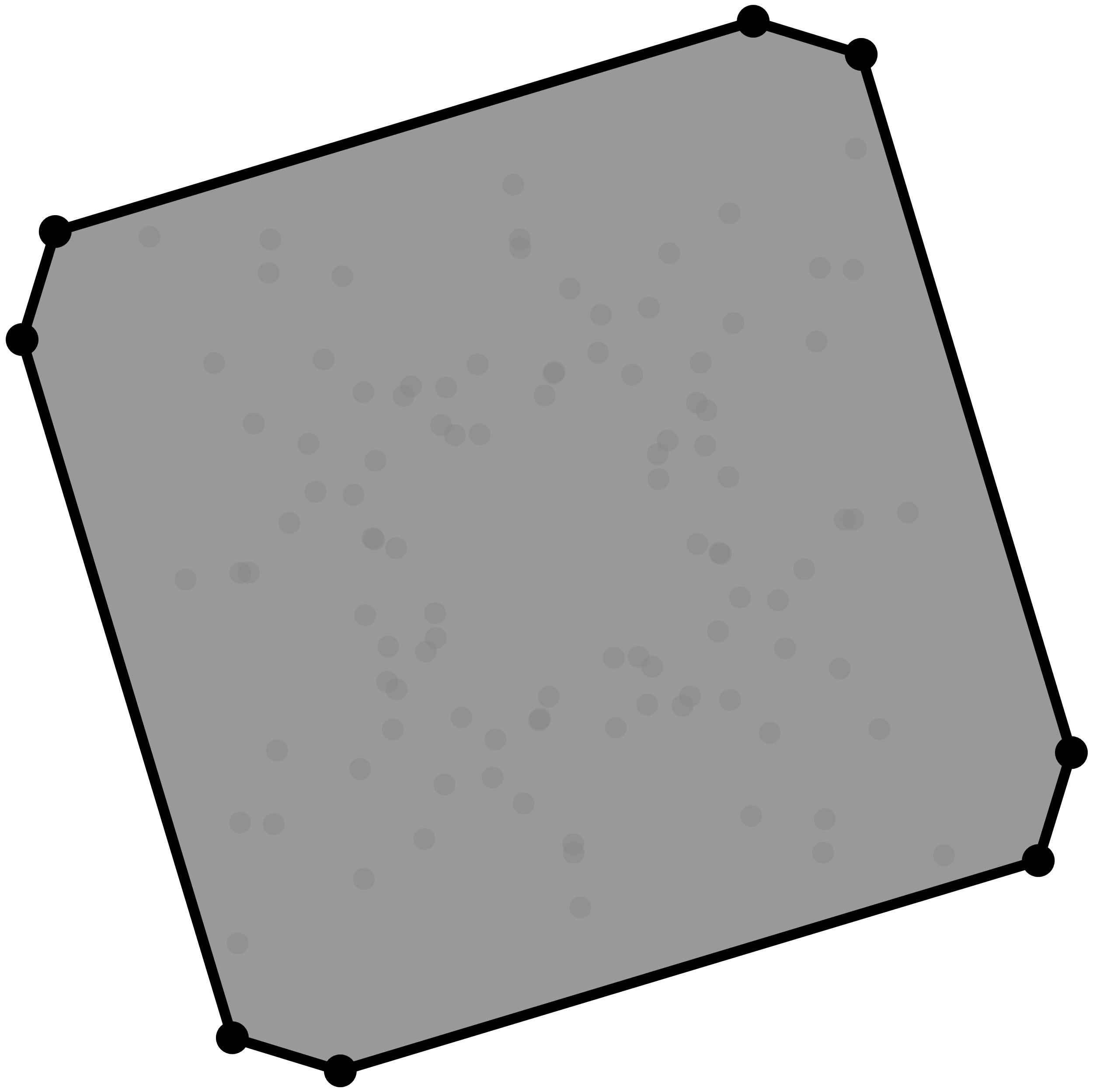}
\hspace{2ex}
\includegraphics[width=0.23\textwidth]{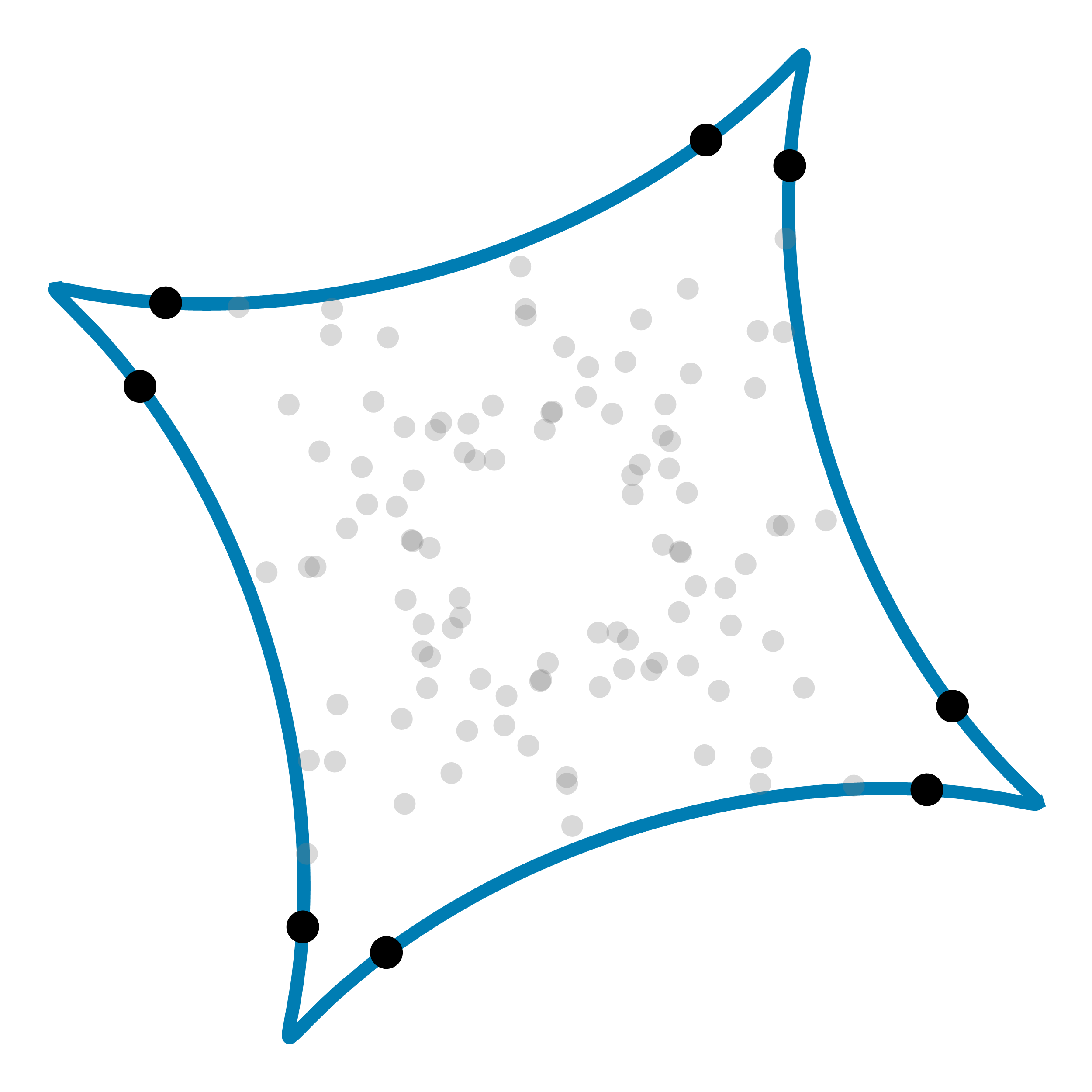}
\hspace{2ex}
\includegraphics[width=0.23\textwidth]{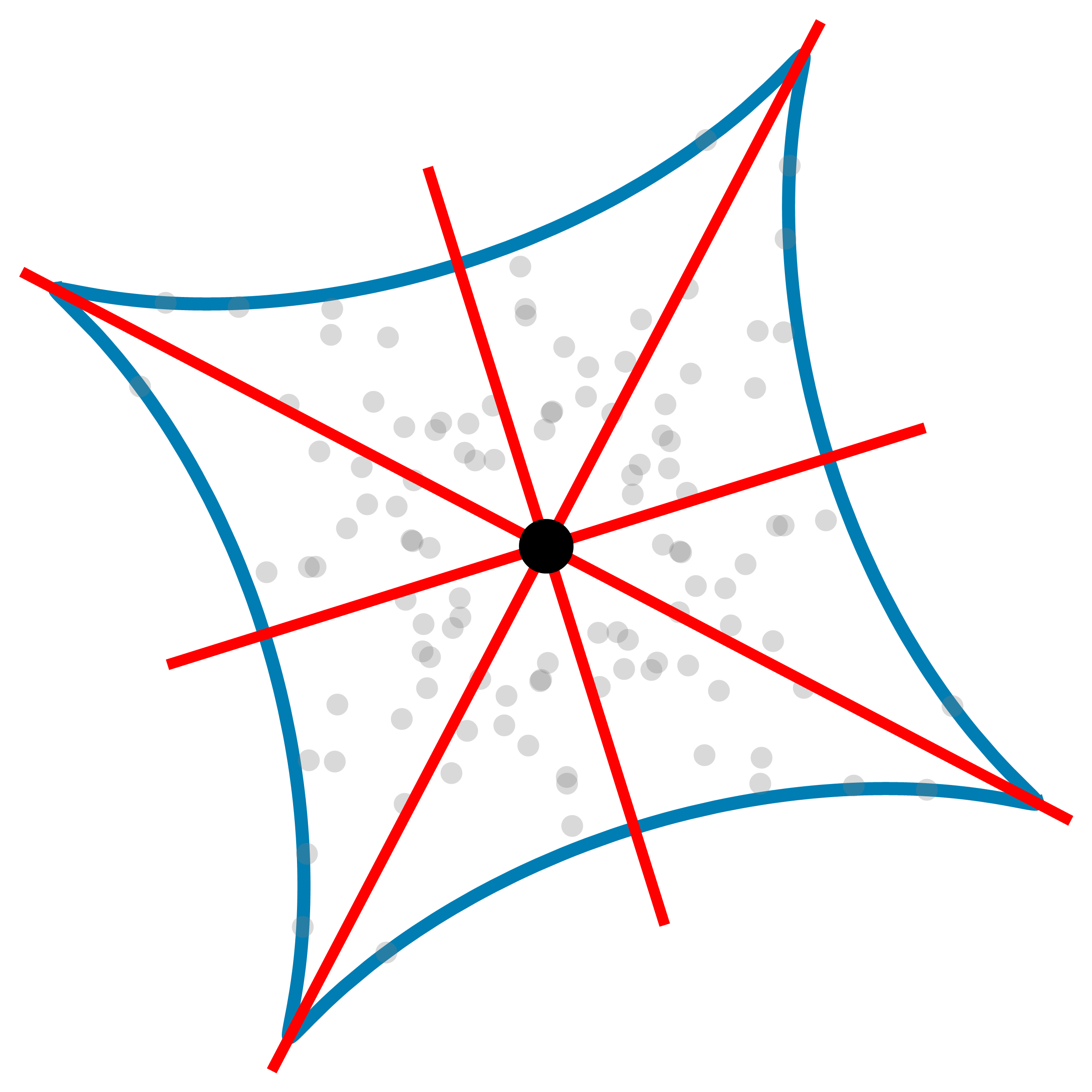}
\caption{Detection of symmetries of a point cloud. From left to right: The point cloud (with the symmetry group $C_4$), its convex hull and the trigonometric interpolant of the vertices of the convex hull having the symmetry group $D_4$.}\label{fig:cloud}
\end{center}
\end{figure}


\begin{figure}[t]
\begin{center}
\includegraphics[width=0.2\textwidth]{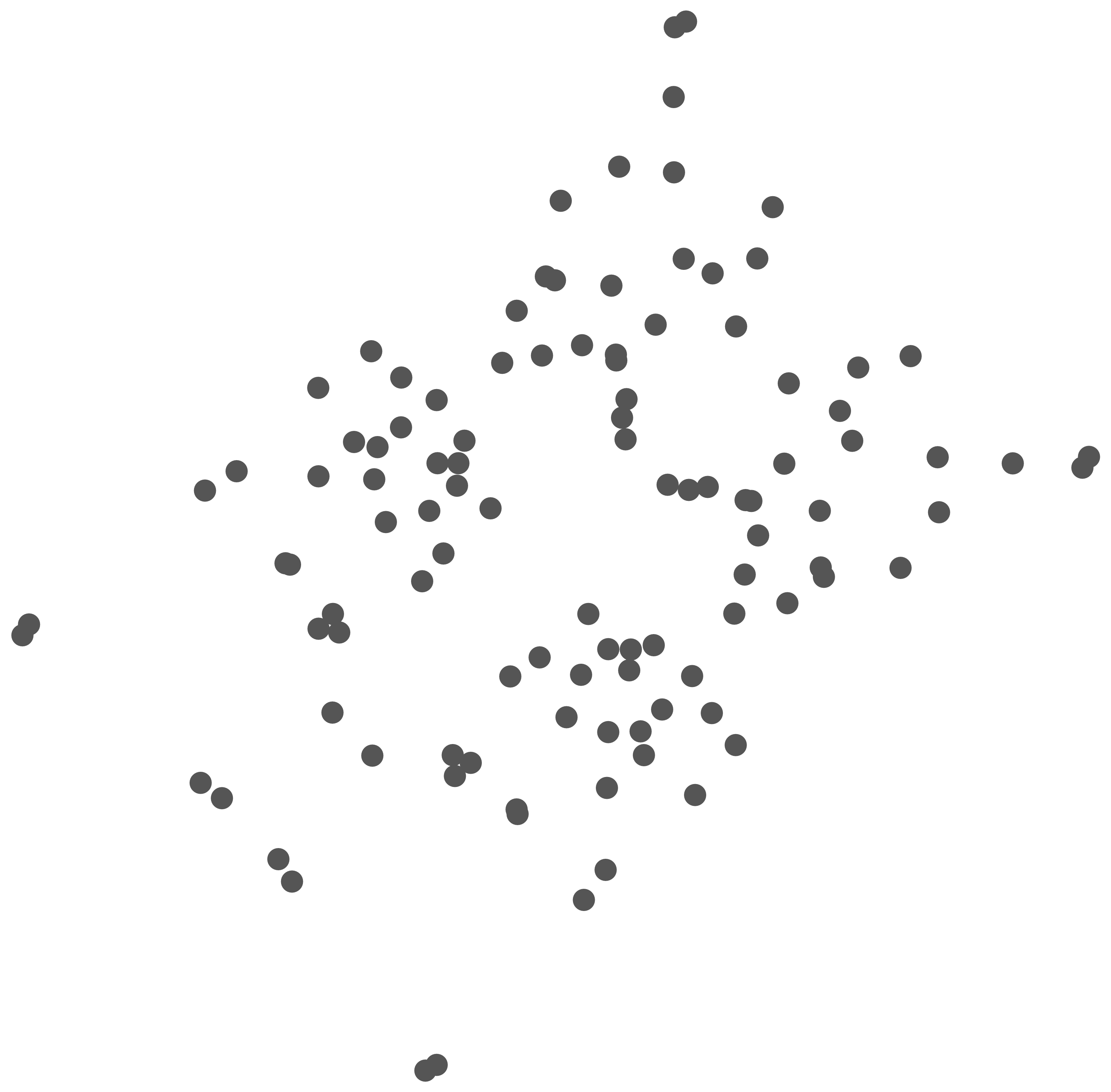}
\hspace{2ex}
\includegraphics[width=0.2\textwidth]{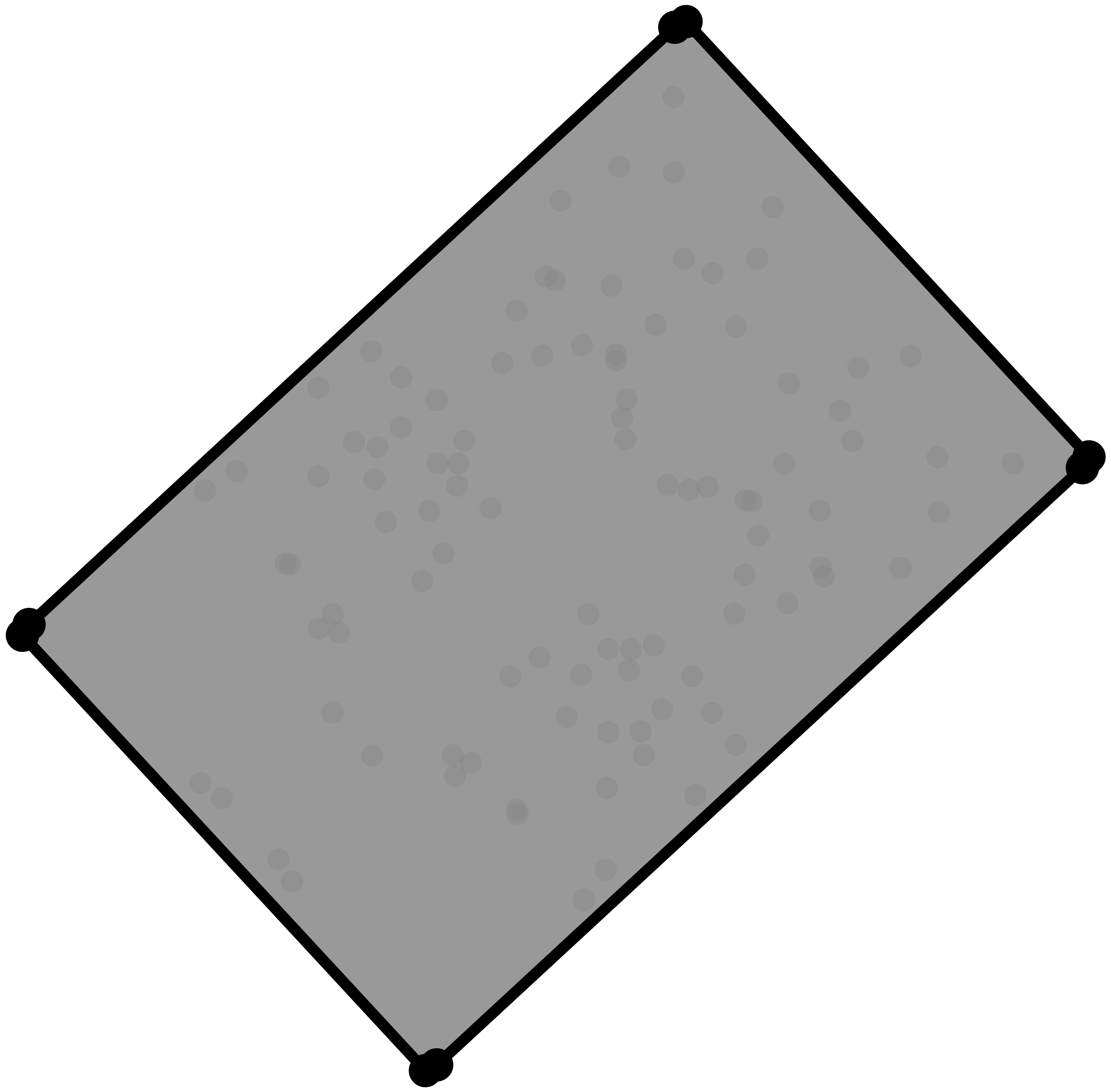}
\hspace{2ex}
\includegraphics[width=0.23\textwidth]{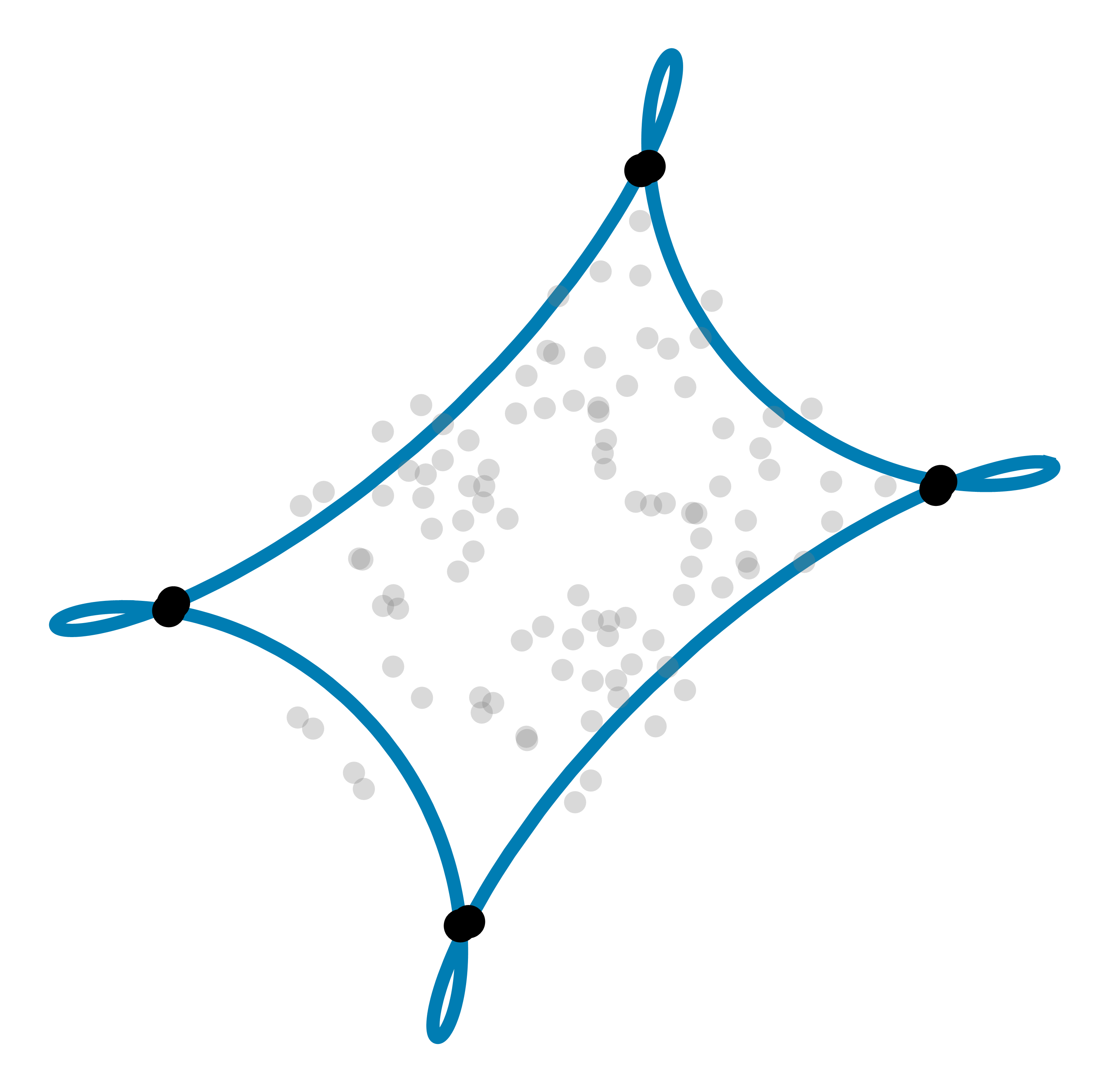}
\hspace{2ex}
\includegraphics[width=0.23\textwidth]{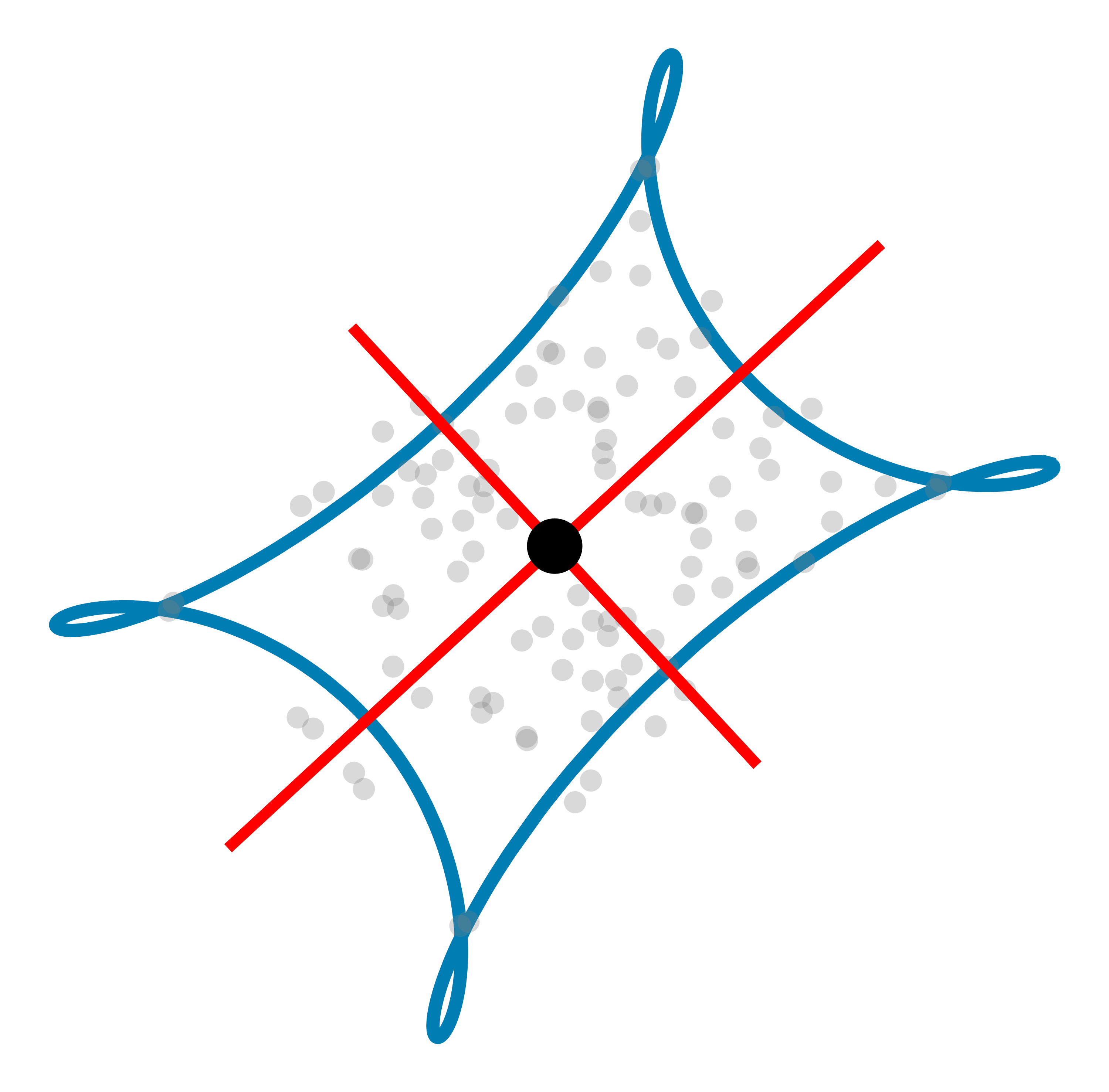}
\caption{Detection of symmetries of a point cloud (with the symmetry group $D_1$). From left to right: The point cloud, its convex hull,  the trigonometric interpolant of the vertices of the convex hull having the symmetry group $D_2$.}\label{fig:cloud2}
\end{center}
\end{figure}

\begin{figure}[t]
\begin{center}
\hspace{-5ex}
\includegraphics[width=0.21\textwidth]{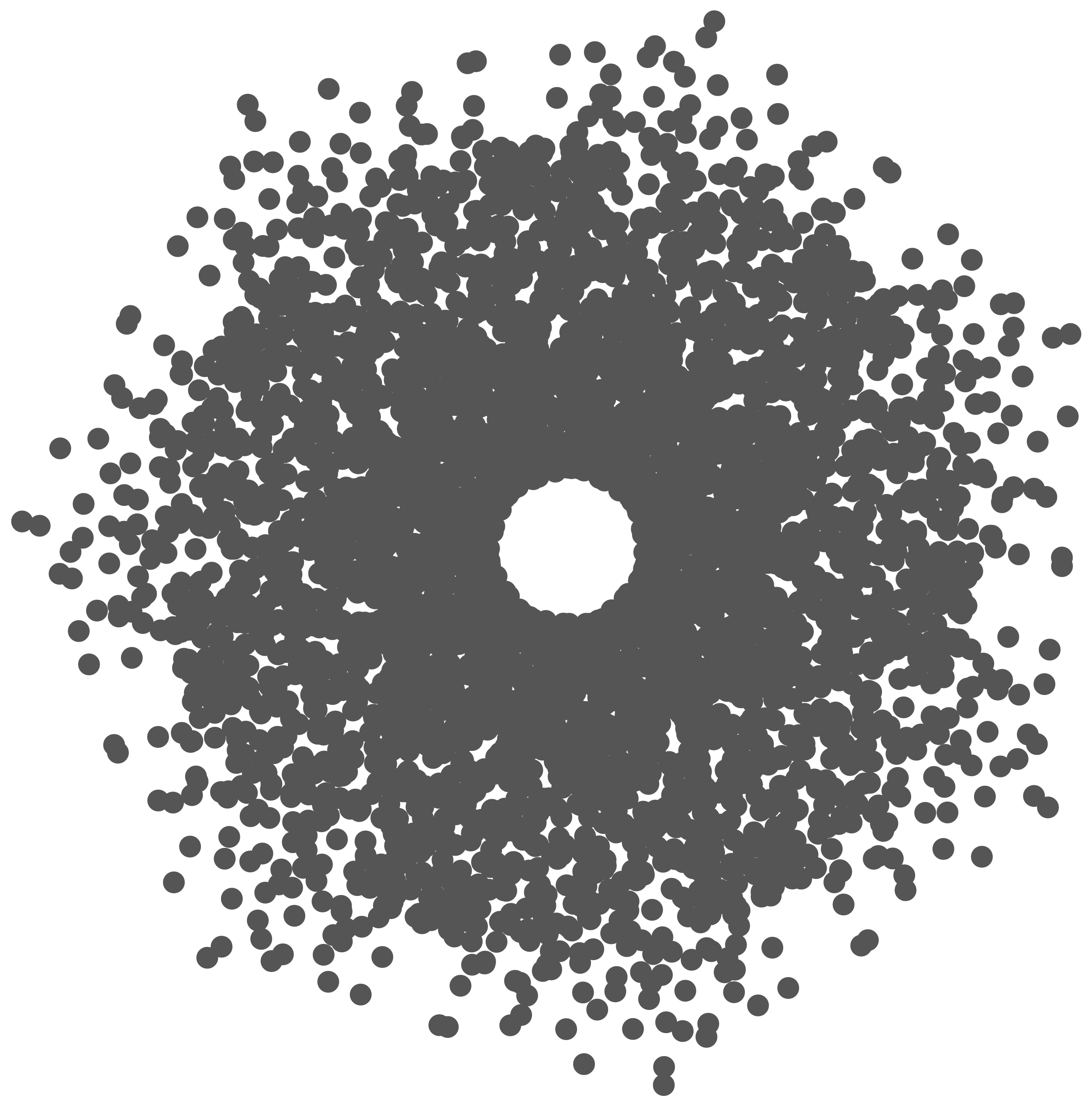}
\hspace{2ex}
\includegraphics[width=0.21\textwidth]{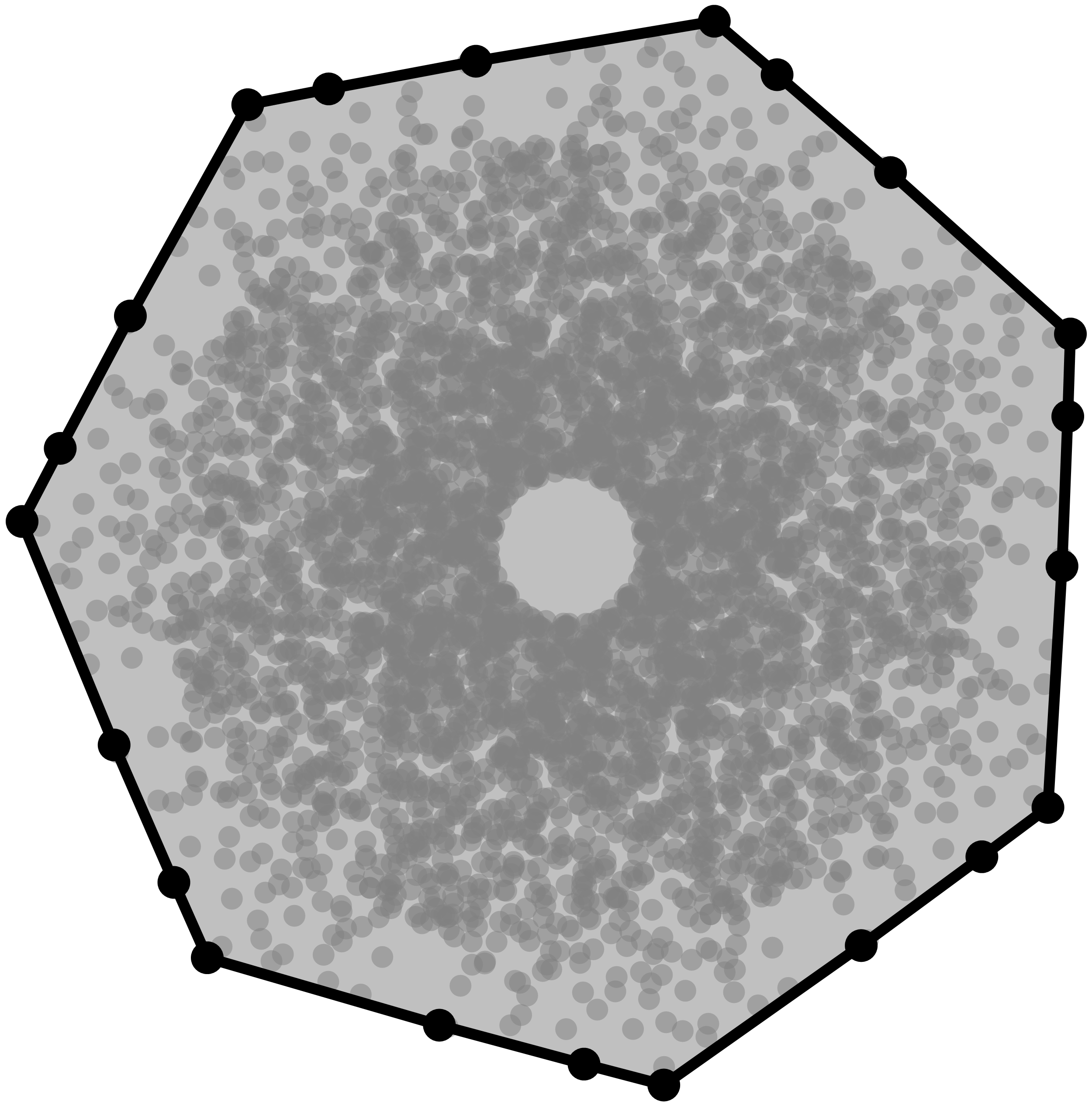}
\hspace{2ex}
\includegraphics[width=0.21\textwidth]{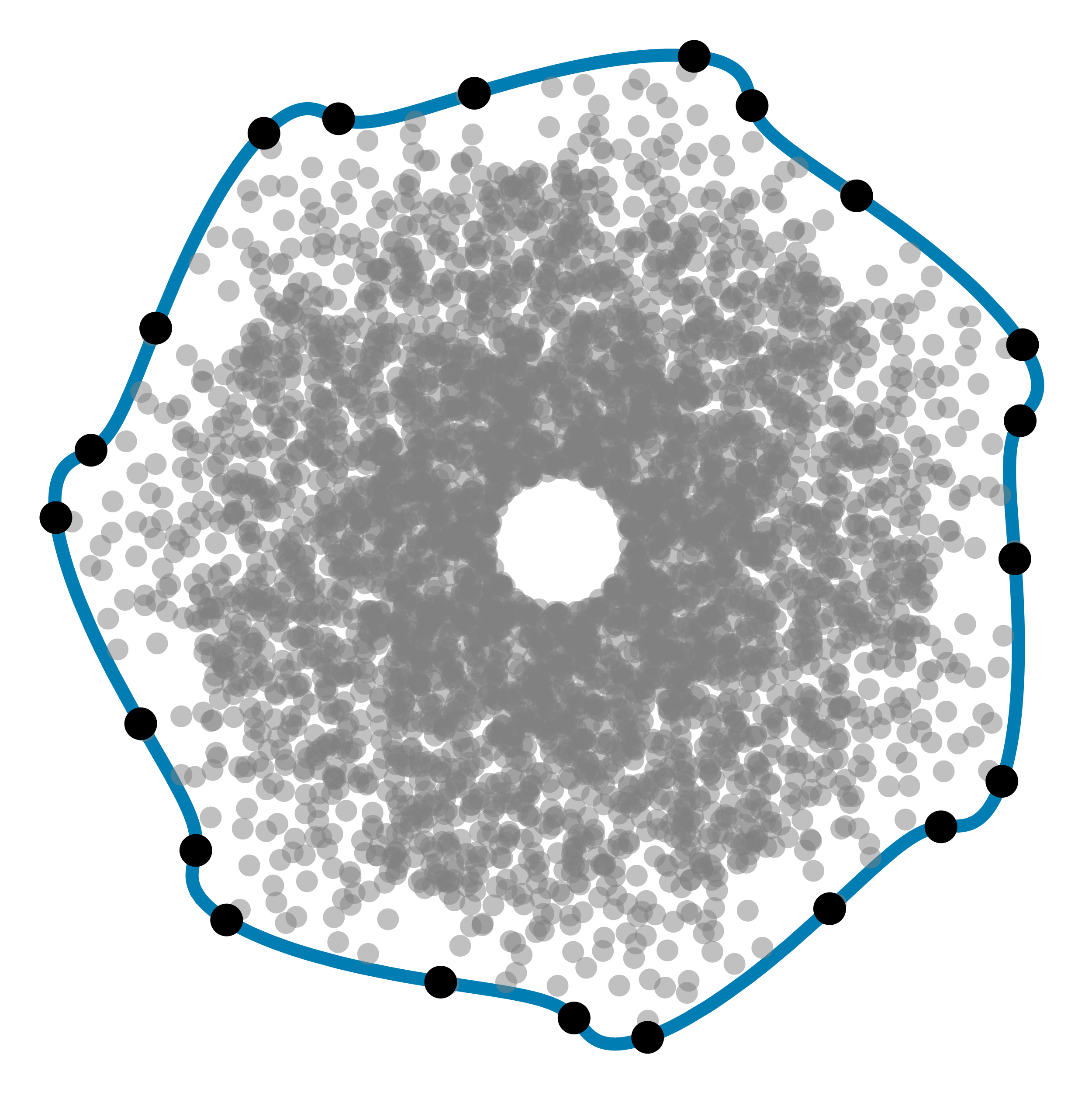}
\hspace{2ex}
\begin{overpic}[height=0.21\textwidth]{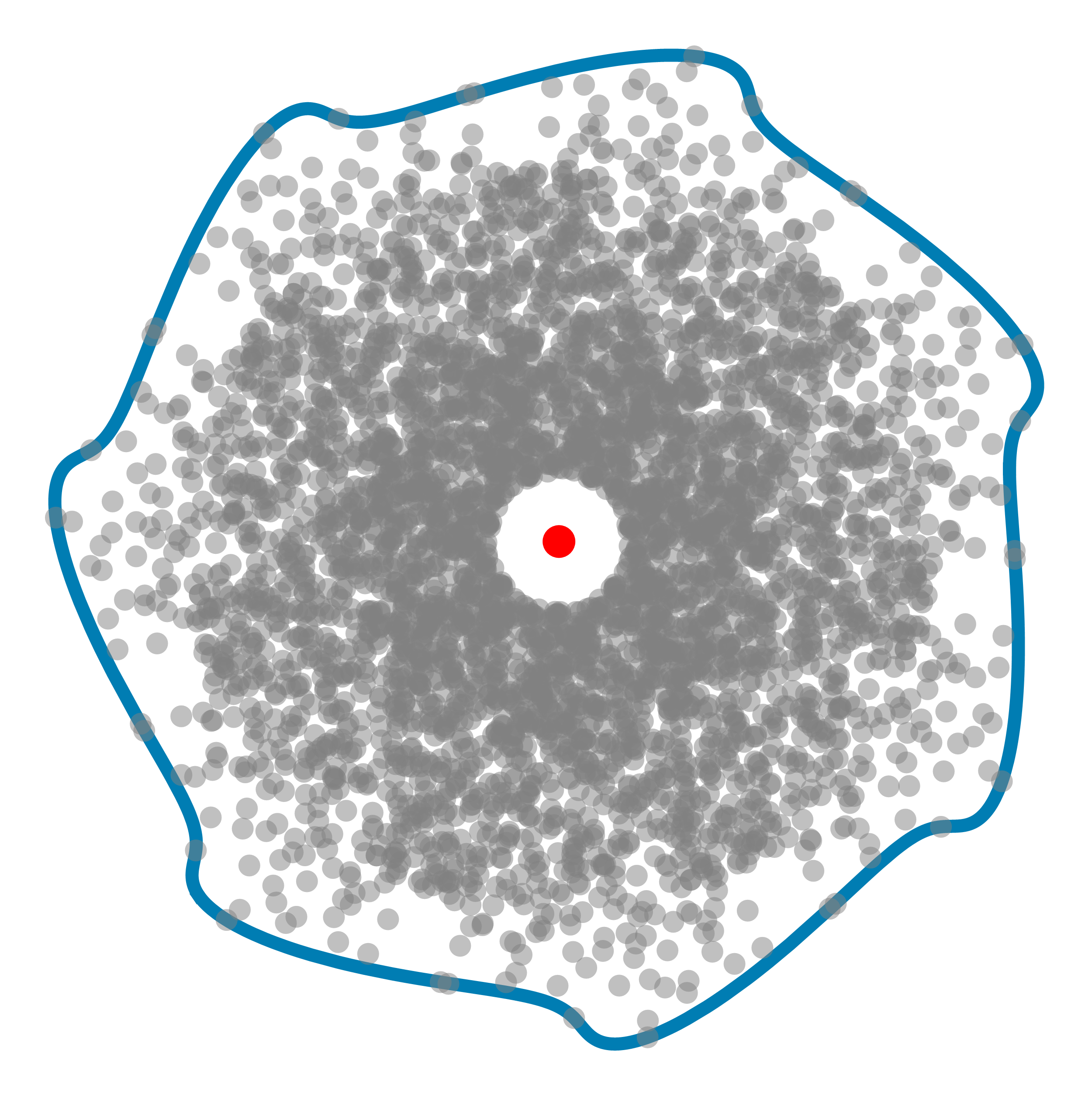}
\put(95,75){\fcolorbox{gray}{white}{\includegraphics[width=0.05\textwidth]{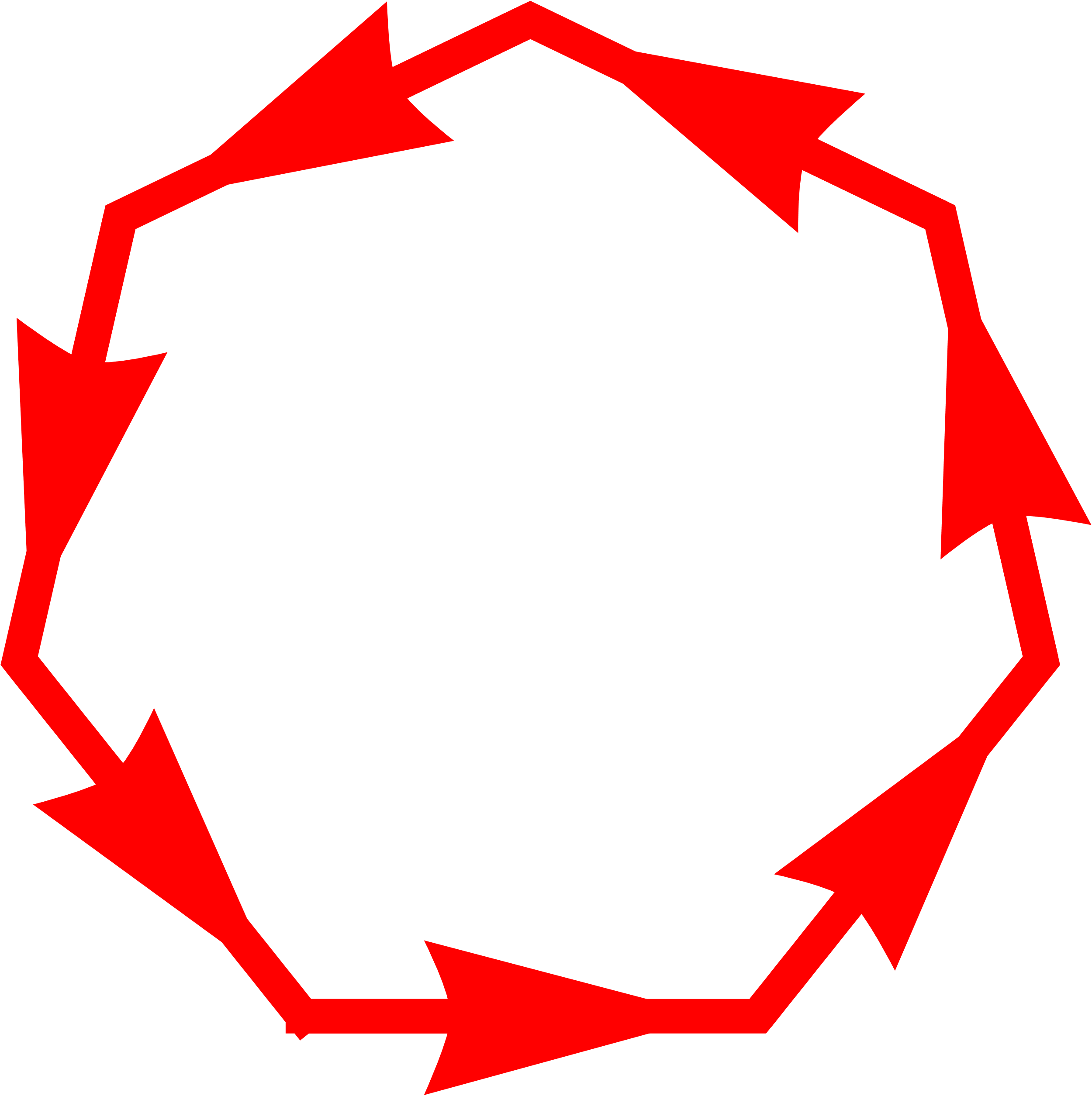}}}
\end{overpic}
\caption{Detection of symmetries of a point cloud (with the symmetry group $C_7$). From left to right: The point cloud, its convex hull,  the trigonometric interpolant of the vertices of the convex hull also having the symmetry group $C_7$.}\label{fig:cloud3}
\end{center}
\end{figure}


\begin{example}\rm
Consider three point clouds $X_1, X_2$ and $X_3$. We determine their convex hulls $\mathrm{CH}(X_i)$ with 
the boundaries $\partial\,\mathrm{CH}(X_i)$ being closed discrete curves and then interpolate their vertices by the trigonometric curves $T_{\partial\,\mathrm{CH}(X_i)}$. The interpolant $T_{\partial\,\mathrm{CH}(X_1)}$ possesses the symmetry group $D_4$, however the original point cloud $X_1$ has only $C_4$, see Fig.~\ref{fig:cloud}. The trigonometric curve $T_{\partial\,\mathrm{CH}(X_2)}$ has the symmetry group $D_2$, whereas the symmetry group of the point cloud $X_2$ is only $D_1$, see Fig.~\ref{fig:cloud2}. Finally, both $T_{\partial\,\mathrm{CH}(X_3)}$ and $X_3$ have the same symmetry group $C_7$, see Fig.~\ref{fig:cloud3}.
\end{example}

\begin{figure}[t]
\begin{center}
\hspace{-8ex}
\includegraphics[width=0.2\textwidth]{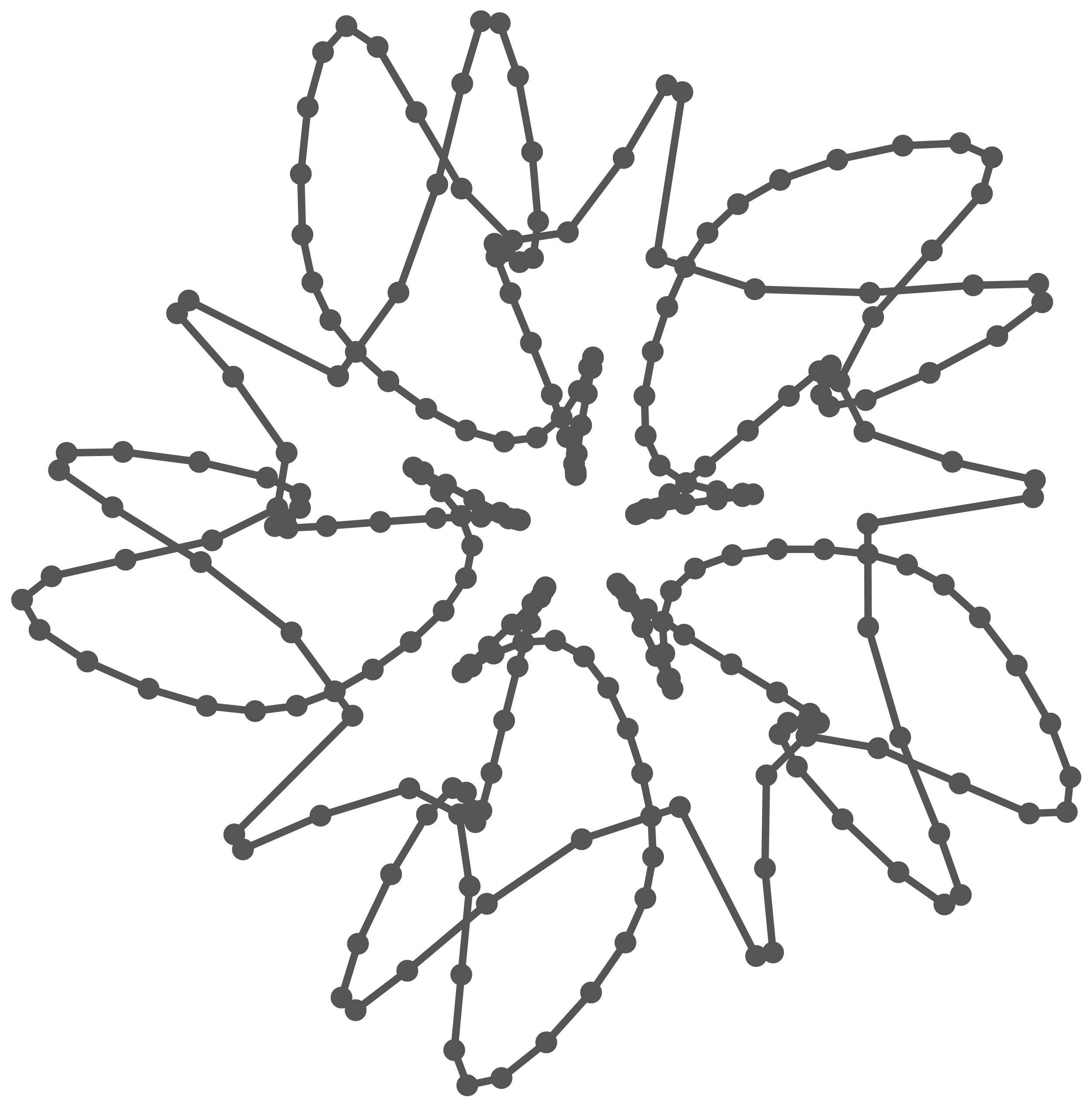}
\hspace{2ex}
\begin{overpic}[height=0.2\textwidth]{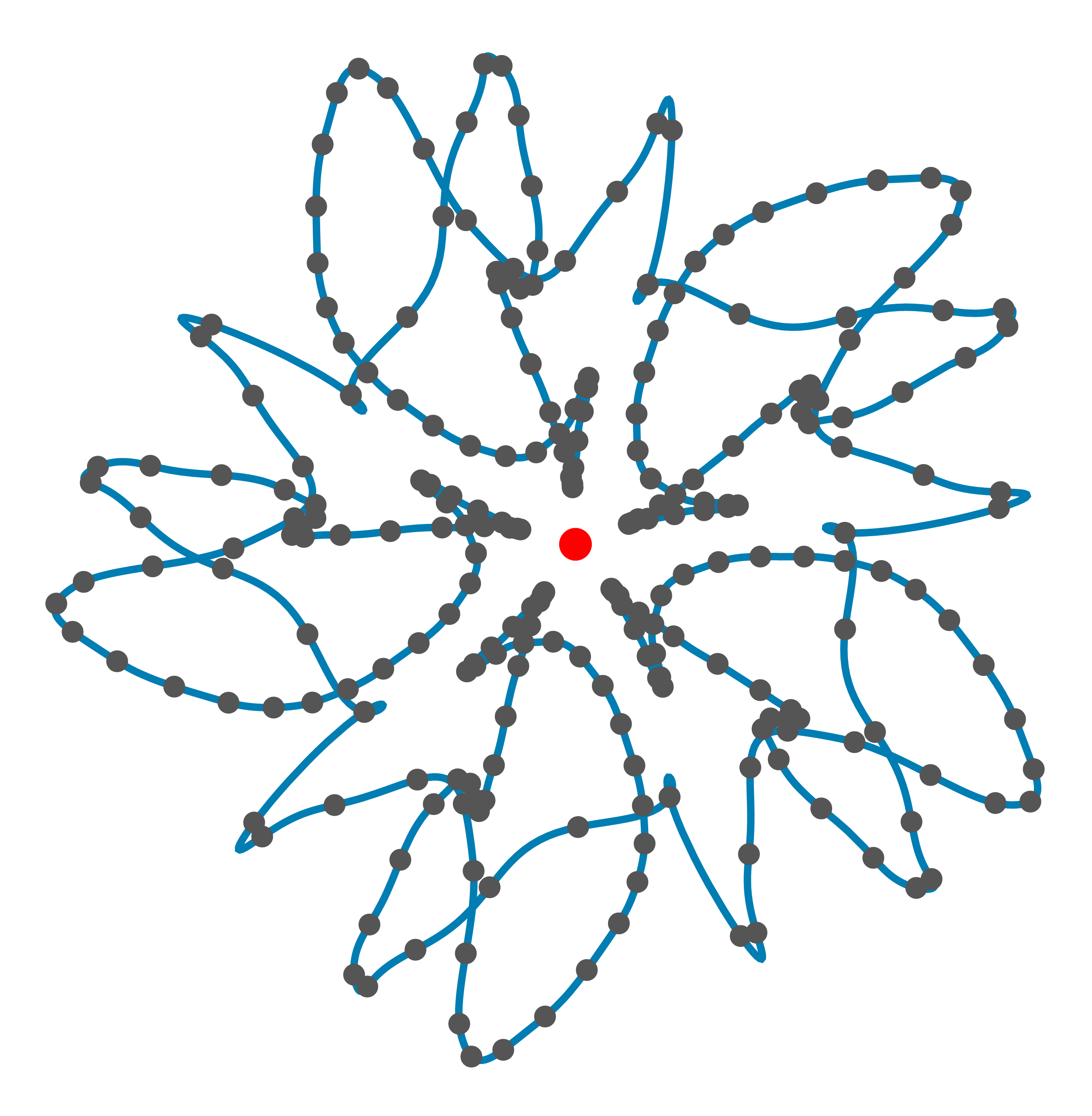}
\put(95,75){\fcolorbox{gray}{white}{\includegraphics[width=0.05\textwidth]{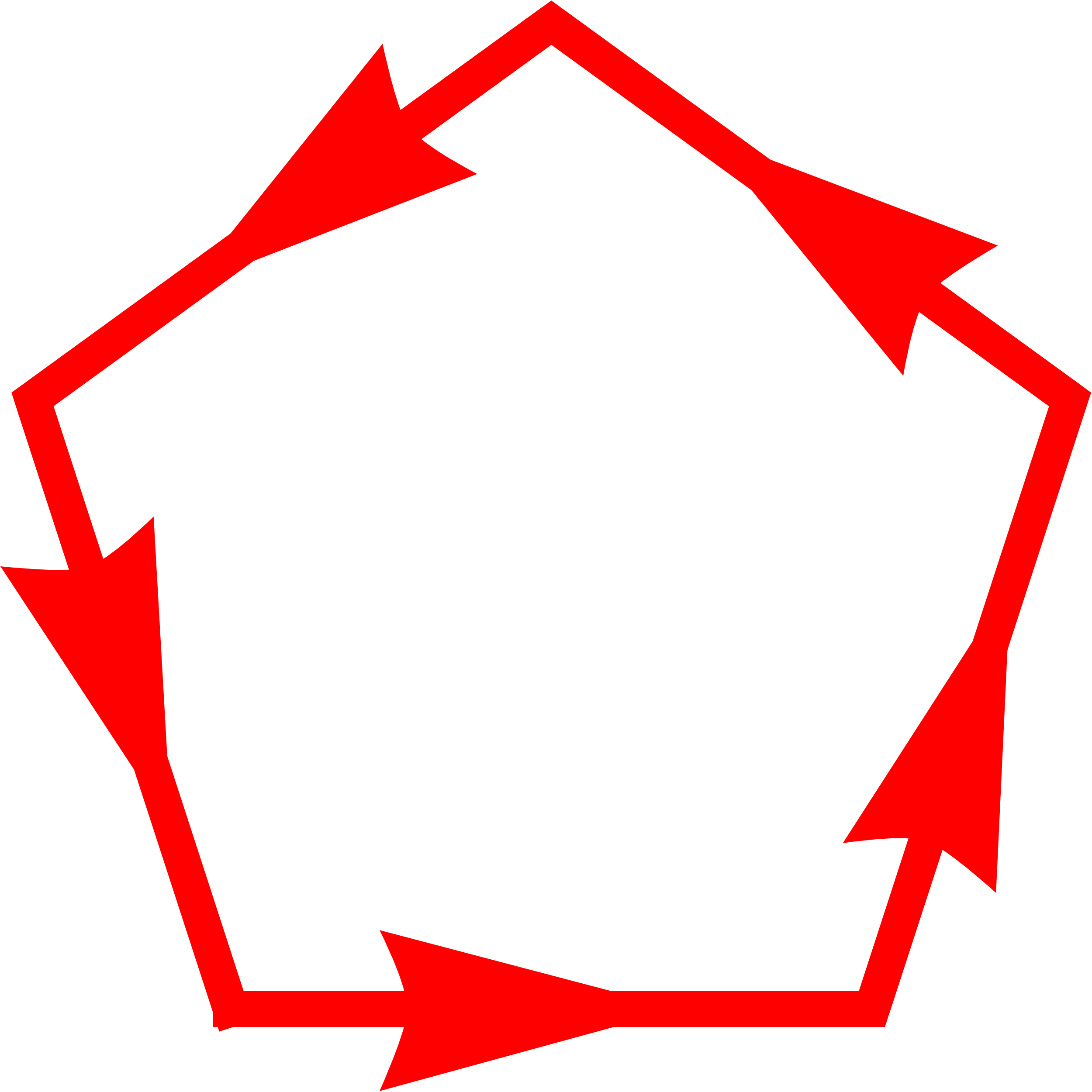}}}
\end{overpic}
\hspace{7ex}
\vrule width1pt
\hspace{1ex}
\includegraphics[width=0.2\textwidth]{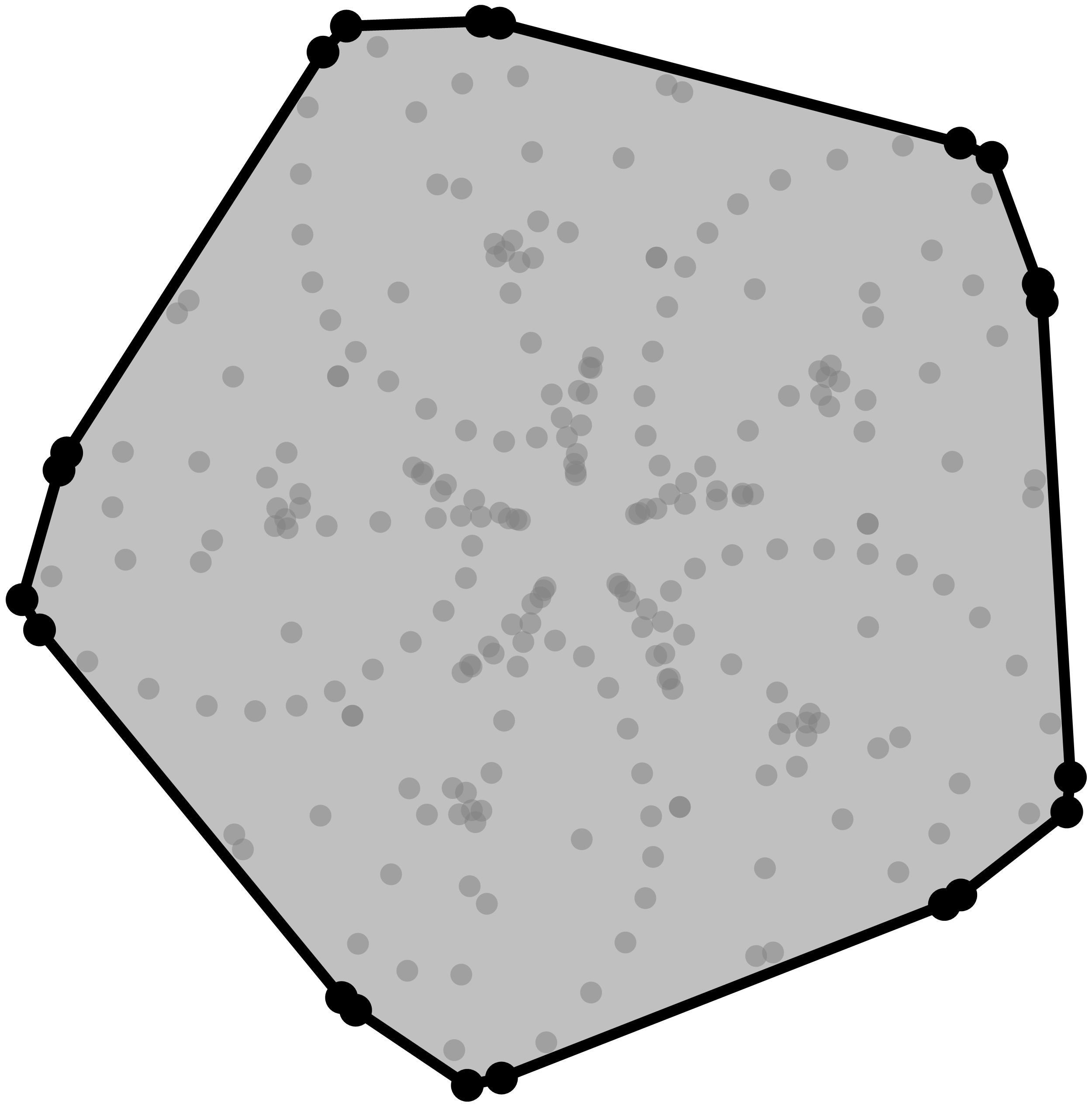}
\hspace{2ex}
\begin{overpic}[height=0.2\textwidth]{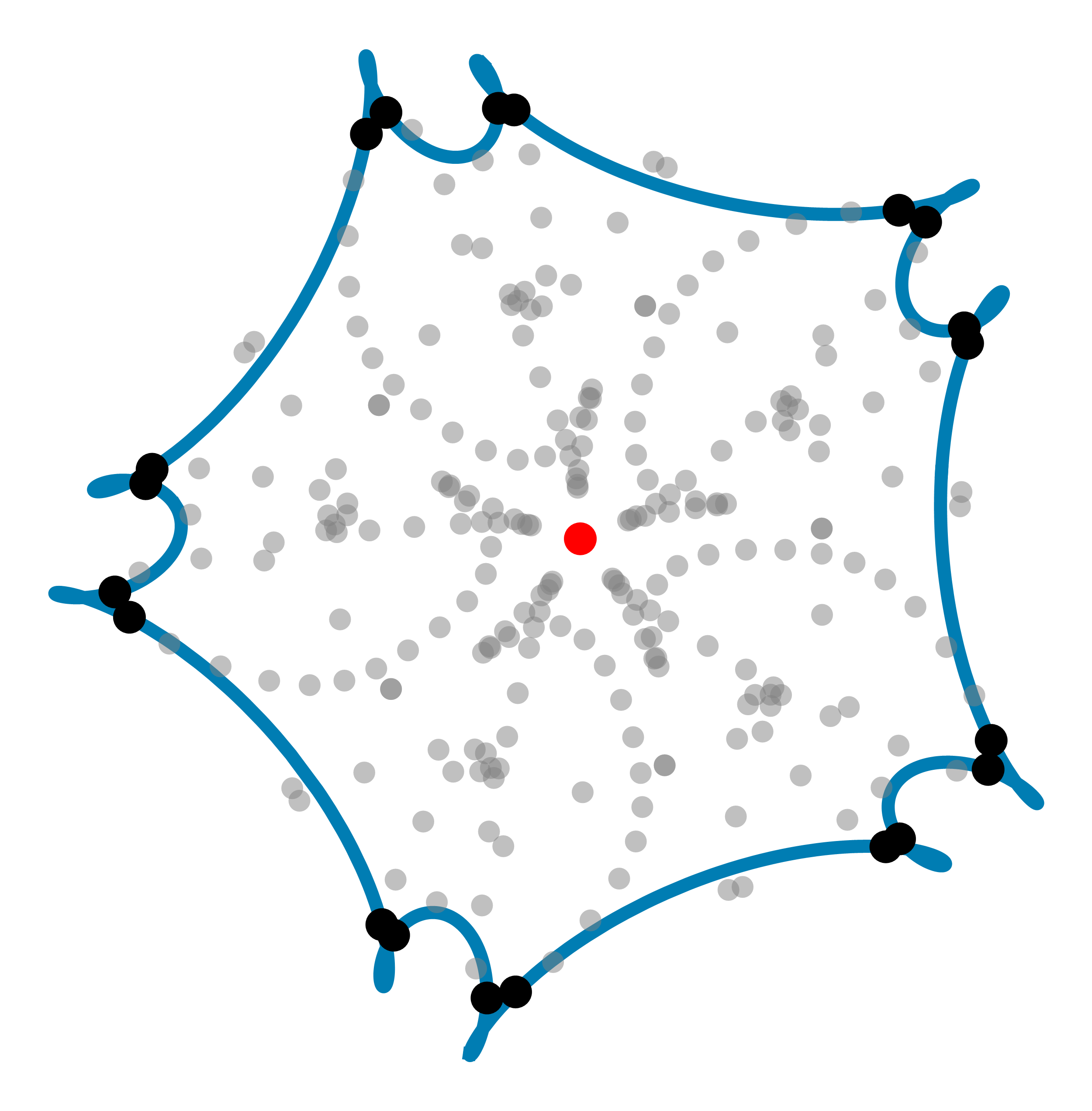}
\put(95,75){\fcolorbox{gray}{white}{\includegraphics[width=0.05\textwidth]{sym_C5.png}}}
\end{overpic}
\caption{Detection of symmetries of a discrete curve using both presented method. 
Left: We use the method for discrete curves. Right: The method for point clouds is employed. We emphasize that both method yields the symmetry group $C_5$.}\label{fig:all}
\end{center}
\end{figure}

\begin{example}\rm
As our final example we consider a discrete curve $C$. First, employing Algorithm~\ref{algoritmus1}  we conclude that $C$ has the symmetry group $C_5$, see Fig.~\ref{fig:all} (left). Second, we treat $C$ via its vertices as a point cloud $X$, see Fig.~\ref{fig:all} (right), and compute the symmetry group $\sym(X)$. Since this symmetry group can be bigger, in general, it is always necessary to apply the confirmation step, see Remark~\ref{rem:alg_step3}. Nonetheless, in this particular example the symmetry group of the point cloud is also  $C_5$.
\end{example}

\section{Summary}\label{sec:sum}

In this paper we have presented a novel, efficient method to compute the symmetries of planar trigonometric curves.
The axial and rotational symmetries were determined directly from their trigonometric parameterizations. The whole algorithm was summarized in a decision tree. Based on this a simple method for computing global exact symmetries of closed discrete curves in plane was formulated. Taking advantage of the fact that trigonometric interpolation of a given ordered set of point commutes with isometries, the studied problem dealing with discrete curves was transformed to determining symmetry groups of trigonometric curves. After a suitable modification we also applied the formulated approach on unorganized clouds of points. It is natural to wonder if the method can be suitably reformulated for 3D data. This is a question that we would like to investigate in the future. In addition we would like to extend the method for perturbed input data and approximate symmetries in our further research.

\bigskip
\section*{Acknowledgments}

The authors were supported by  the  grant  21-08009K  of  the  Grant  Agency  of  the  Czech Republic.

\end{document}